\documentclass[11pt]{article}

\usepackage{amsfonts,latexsym,amsthm,amssymb,amsmath,amscd,euscript,tikz,mathtools}
\usepackage{framed}
\usepackage[margin=1in]{geometry}
\usepackage{color}
\usepackage[colorlinks=true,citecolor=blue,linkcolor=blue]{hyperref}
\usepackage{enumitem}
\usepackage{siunitx}
\usepackage{textcomp}
\usepackage{physics}
\usepackage{mathrsfs}
\usepackage{dsfont}
\usepackage[ruled,vlined,linesnumbered]{algorithm2e}
\usepackage{titlesec}
\usepackage{tikz-cd}
\usepackage{thmtools}
\usepackage{thm-restate}
\usepackage{cleveref}
\usepackage{stmaryrd}
\usepackage{graphicx}
\usetikzlibrary{positioning,chains,fit,shapes,calc}

\allowdisplaybreaks[1]

\newtheorem{theorem}{Theorem}[section]

\newtheorem{lemma}[theorem]{Lemma}
\newtheorem{claim}[theorem]{Claim}
\newtheorem{corollary}[theorem]{Corollary}

\newtheorem{definition}[theorem]{Definition}
\newtheorem{example}[theorem]{Example}
\newtheorem{remark}[theorem]{Remark}

\theoremstyle{remark}

\newcommand{\cA}{\mathcal{A}}\newcommand{\cB}{\mathcal{B}}
\newcommand{\cC}{\mathcal{C}}\newcommand{\cD}{\mathcal{D}}
\newcommand{\cF}{\mathcal{F}}
\newcommand{\cG}{\mathcal{G}}
\newcommand{\cI}{\mathcal{I}}

\newcommand{\cS}{\mathcal{S}}

\newcommand{\cX}{\mathcal{X}}
\newcommand{\cY}{\mathcal{Y}}

\newcommand{\bC}{\mathbb{C}}
\newcommand{\bE}{\mathbb{E}}\newcommand{\bF}{\mathbb{F}}

\newcommand{\bN}{\mathbb{N}}

\newcommand{\bR}{\mathbb{R}}

\newcommand{\bZ}{\mathbb{Z}}

\newcommand{\1}{\mathds{1}}

\newcommand{\poly}{\operatorname{poly}}

\newcommand{\Enc}{\operatorname{Enc}}

\newcommand{\evl}{\operatorname{ev}}

\newcommand{\ver}{\operatorname{v}}
\newcommand{\dir}[1]{{#1}^{\operatorname{dir}}}
\newcommand{\spectrum}{\operatorname{spectrum}}
\newcommand{\labV}{\operatorname{L}_V}
\newcommand{\labE}{\operatorname{L}_E}
\newcommand{\liftlab}{\operatorname{L}_{\mathrm{lift}}}

\newcommand{\ind}[1]{\mathbf{1}_{#1}}
\newcommand{\subrank}{\operatorname{subrank}}
\newcommand{\loc}{\mathrm{loc}}
\newcommand{\direc}{\operatorname{Dir}}

\newcommand{\nc}{\newcommand}

\nc{\on}{\operatorname}
\nc{\Spec}{\on{Spec}}
\nc{\Aut}{\textit{Aut}}
\nc{\id}{\textit{id}}
\nc{\chr}{\on{char}}
\nc{\im}{\on{im}}
\nc{\Hom}{\on{Hom}}
\nc{\lcm}{\on{lcm}}
\nc{\dual}[1]{\prescript{t}{}{#1}}
\nc{\transpose}[1]{{#1}^{\intercal}}
\nc{\Sym}{\on{Sym}}
\nc{\End}{\on{End}}
\nc{\stab}{\on{stab}}
\nc{\Li}{\on{Li}}
\nc{\spn}{\on{span}}
\nc{\sgn}{\on{sign}}
\nc{\supp}{\on{supp}}
\nc{\Unif}{\on{Unif}}

\makeatletter
\newcommand\footnoteref[1]{\protected@xdef\@thefnmark{\ref{#1}}\@footnotemark}
\makeatother

\title{Quantum LDPC Codes with Transversal Non-Clifford Gates \\ via Products of Algebraic Codes}
\author{Louis Golowich\thanks{Supported by a National Science Foundation Graduate Research Fellowship under Grant No.~DGE 2146752. Research also supported in part by a ONR grant N00014-24-1-2491 and a UC Noyce initiative award.} \\
  UC Berkeley \\
  \href{mailto:lgolowich@berkeley.edu}{\texttt{lgolowich@berkeley.edu}}
  \and
  Ting-Chun Lin\thanks{Supported in part by funds provided by the U.S. Department of Energy (D.O.E.) under the cooperative research agreement DE-SC0009919 and by the Simons Collaboration on Ultra-Quantum Matter, which is a grant from the Simons Foundation (652264 JM).} \\
  UC San Diego \\
  Hon Hai Research Institute \\
  \href{mailto:til022@ucsd.edu}{\texttt{til022@ucsd.edu}}
}

\begin{document}

\pagenumbering{gobble}

\maketitle
\thispagestyle{empty}

\begin{abstract}
  For every integer $r\geq 2$ and every $\epsilon>0$, we construct an explicit infinite family of quantum LDPC codes supporting a transversal $C^{r-1}Z$ gate with length $N$, dimension $K\geq N^{1-\epsilon}$, distance $D\geq N^{1/r}/\poly(\log N)$, and stabilizer weight $w\leq\poly(\log N)$. The previous state of the art construction (in most parameter regimes) was the $r$-dimensional color code, which has only constant dimension $K=O(1)$, and otherwise has the same parameters up to polylogarithmic factors. Our construction provides the first known codes with low-weight stabilizers that are capable of magic state distillation with arbitrarily small yield parameter $\gamma=\log(N/K)/\log(D)>0$.
  

  A classical analogue of transversal $C^{r-1}Z$ gates is given by the multiplication property, which requires component-wise products of classical codewords to belong to another similar code. As a byproduct of our techniques, we also obtain a new construction of classical locally testable codes with such a multiplication property.

  We construct our codes as products of chain complexes associated to classical LDPC codes, which in turn we obtain by imposing local Reed-Solomon codes on a specific spectral expander that we construct. We prove that our codes support the desired transversal $C^{r-1}Z$ gates by using the multiplication property to combine local circuits based on the topological structure.
\end{abstract}

\newpage

\tableofcontents

\newpage

\pagenumbering{arabic}

\section{Introduction}
\label{sec:intro}
It is a major challenge in fault-tolerant quantum computation to efficiently perform non-Clifford gates, such as $CCZ$ or $T$. An ideal object for this task is a quantum LDPC (qLDPC) code of large rate and distance supporting transversal $CCZ$ (or $T$) gates. Such codes are defined to permit a constant-depth circuit of physical $CCZ$ (or $T$) gates that induces logical $CCZ$ gates on the encoded message. However, such codes have proven difficult to construct, with existing constructions suffering from poor dimension and/or distance.

Our main result in the theorem below addresses this question. In fact, we consider the more general $r$-qudit gate
\begin{equation}
  \label{eq:cczdefinf}
  C^{r-1}Z^a = \sum_{z_1,\dots,z_r\in\bF_q} e^{2\pi i\tr_{\bF_q/\bF_p}(a\cdot z_1\cdots z_r)/p}\ket{z_1,\dots,z_r}\bra{z_1,\dots,z_r}
\end{equation}
for $r\in\bN$ and $a\in\bF_q$, where $\bF_q$ is a finite field of characteristic $p$, and we let $C^{r-1}Z=C^{r-1}Z^1$. This gate lies outside the Clifford group for $r\geq 3$.

\begin{theorem}[Informal statement of Corollary~\ref{cor:qldpcmain}]
  \label{thm:qldpcmaininf}
  For every fixed real number $\epsilon>0$, every fixed integer $r\geq 2$, and every fixed prime power $q$ (including $q=2$), there exists an explicit infinite family of\footnote{Recall that an $[N,K,D]_q$ (resp.~$[[N,K,D]]_q$) code is a classical (resp.~quantum) code of length $N$, dimension (i.e.~message length) $K$, distance (i.e.~error tolerance) $D$, and alphabet size $q$.}
  \begin{equation*}
    [[N,\; K\geq N^{1-\epsilon},\; D\geq N^{1/r}/\poly(\log N)]]_q,
  \end{equation*}
  quantum LDPC codes of locality (i.e.~stabilizer weight) $w\leq\poly(\log N)$ that support a transversal $C^{r-1}Z$ gate in the following sense: There exists a depth-1 physical circuit consisting of $C^{r-1}Z^a$ gates acting across $r$ code states that induces logical $C^{r-1}Z$ gates on $N^{1-\epsilon}$ disjoint $r$-tuples of logical qudits.
\end{theorem}


Previously, the best known parameters in most regimes for qLDPC codes with transversal non-Clifford gates were achieved by the $r\geq 3$-dimensional color code \cite{bombin_exact_2007,bombin_topological_2007,bombin_gauge_2015}, which has only constant dimension $K=O(1)$, and otherwise matches our parameters in Theorem~\ref{thm:qldpcmaininf} up to polylogarithmic factors. Thus we improve this dimension $K$ from constant to almost linear, at the cost of a polylogarithmic loss in locality and distance. 

\begin{remark}
  \label{remark:logicalstructure}
  We emphasize that our transversal gates induce logical $C^{r-1}Z$ gates on an almost linear number of disjoint tuples of logical qudits. This condition is useful in applications to fault-tolerance, such as magic state distillation (see below), as it provides high (almost linear) logical gate parallelization without inducing undesired entanglement (due to the disjointness condition).

  Some prior works \cite{zhu_non-clifford_2023,scruby_quantum_2024} have constructed qLDPC codes of linear or almost linear dimension with transversal $T$ gates. However, these codes have only logarithmic distance, and furthermore, the transversal $T$ gates induce logical $CCZ$ gates on various overlapping triples of logical qubits. As a result, it remains an open question to determine how such codes perform in fault-tolerance applications, including magic state distillation.
\end{remark}

Applying the techniques of \cite{bravyi_magic-state_2012} to our codes in Theorem~\ref{thm:qldpcmaininf} with fixed $r$ (e.g.~$r=3$) and arbitrarily small $\epsilon>0$ gives a magic state distillation protocol with arbitrarily small yield parameter $\gamma=\log(N/K)/\log(D)=r\epsilon+o(1)$, meaning that $M$ copies of the magic state $C^{r-1}Z\ket{+}^{\otimes r}$ with some constant noise rate can be used to distill $M/\log^\gamma(1/\delta)$ copies of the magic state with noise rate $\delta$. To the best of our knowledge, these codes from Theorem~\ref{thm:qldpcmaininf} are the first to achieve $\gamma\rightarrow 0$ with low-weight stabilizers.\footnote{\label{footnote:rainbowcaveat} The recent work \cite{scruby_quantum_2024} claimed to construct qLDPC codes supporting transversal $T$ gates with $K=\Theta(N)$ and $D=\Theta(\log N)$, so that $\log(N/K)/\log(D)\rightarrow 0$. However, the optimal overhead for a magic state distillation protocol using the codes of \cite{scruby_quantum_2024} is unclear due to the discussion in Remark~\ref{remark:logicalstructure}. In particular, to the best of our knowledge no such protocol with yield parameter $\gamma\rightarrow 0$ has been constructed.} Indeed, the first known magic state distillation protocols with $\gamma\rightarrow 0$ were only recently constructed by \cite{wills_constant-overhead_2024,golowich_asymptotically_2024,nguyen_good_2024}, building on the work of \cite{krishna_towards_2019}. These works used codes supporting transversal $CCZ$ gates with linear-weight stabilizers, and left it as an open question to construct similar codes but with low-weight stabilizers. Theorem~\ref{thm:qldpcmaininf} helps address this question, and provides progress towards the larger goal of constructing asymptotically optimal qLDPC codes with transversal $CCZ$ gates.

While magic state distillation yield provides a useful metric for evaluating code parameters, we emphasize that qLDPC codes of close-to-linear dimension and polynomial distance supporting transversal $CCZ$ gates present the opportunity for alternative protocols for universal fault-tolerant quantum computation, which perform $CCZ$ gates directly on qLDPC code states and hence do not require magic state distillation. It is an interesting question to determine how our codes in Theorem~\ref{thm:qldpcmaininf} would perform in such a protocol.

A classical analogue of transversal quantum $CCZ$ gates is given by the \textit{multiplication property}, which requires the component-wise product of classical codewords to belong to another code of similar parameters. Using similar techniques as in our qLDPC construction, we also construct \textit{classical locally testable codes (cLTCs)} with such a multiplication property:

\begin{theorem}[Informal statement of Corollary~\ref{cor:cltcmain}]
  \label{thm:cltcmaininf}
  For every fixed $\epsilon>0$, it holds for every sufficiently large prime power $q$ that there exist explicit
  \begin{equation*}
    [N=q^{\poly(q)},\; K\geq N^{1-\epsilon}, D\geq N/\poly(\log N)]_q
  \end{equation*}
  classical LTCs of locality $w\leq\poly(\log N)$ and soundness $\rho\geq 1/\poly(\log N)$, which exhibit the multiplication property.
\end{theorem}


CLTCs are classical LDPC codes for which the distance of a given string to the code can be approximated (up to a multiplicative error of $\rho$) by querying a small number ($\leq w$) of random components of the string. The construction of cLTCs with the multiplication property is in part motivated by applications to complexity theory, such as for probabilistically checkable proofs (PCPs; see e.g.~\cite{dinur_new_2023} for more details on such motivation).

The parameters of our cLTCs described above are comparable to known constructions (up to polylogarithmic factors), such as Reed-Muller codes with the derandomized low-degree test of \cite{ben-sasson_randomness-efficient_2003}, as well as the codes of \cite{dinur_new_2023}. In fact, our cLTCs contain Reed-Muller codes, from which they inherit many properties. However, Theorem~\ref{thm:cltcmaininf} highlights the flexibility of our techniques, as we are able to obtain both qLDPC codes and cLTCs with desirable fault-tolerance properties, analogously to how the work of \cite{panteleev_asymptotically_2022} provided a unified construction of asymptotically good qLDPC codes and cLTCs.

\subsection{Our Techniques}
\label{sec:qldpcinf}
In this section, we outline our techniques for proving our main results described above. At a high level, our qLDPC codes in Theorem~\ref{thm:qldpcmaininf} are constructed as tensor products (also known as ``hypergraph products'' \cite{tillich_quantum_2014} or ``homological products'') of chain complexes associated to classical LDPC codes exhibiting a multiplication property, which helps ensure the resulting quantum codes support transversal $C^{r-1}Z$ gates. Our cLTCs in Theorem~\ref{thm:cltcmaininf} are constructed as balanced products \cite{breuckmann_balanced_2021} (or equivalently in our case, lifted products \cite{panteleev_quantum_2022,panteleev_asymptotically_2022}) of related chain complexes. Therefore we face two main challenges in proving our results:
\begin{enumerate}
\item\label{it:taskmp} Construct classical LDPC codes with an appropriate multiplication property.
\item\label{it:taskccz} Prove that quantum codes from products of such classical LDPC codes support transversal $C^{r-1}Z$ gates.
\end{enumerate}

To address item~\ref{it:taskmp} above, our key result is a new spectral expander graph construction with an appropriate embedding into a high-dimensional vector space. To address item~\ref{it:taskccz}, our main idea is to view the product of classical LDPC codes as a cubical complex with an associated system of local coefficients (i.e.~local codes). Within each ($r$-dimensional) cube in this complex, we use known techniques for performing $C^{r-1}Z$ gates on topological codes, and then we use the multiplication property to ensure that these local gates assemble to a global circuit with the desired logical action.

Specifically, we construct our classical LDPC codes in item~\ref{it:taskmp} using the Sipser-Spielman paradigm \cite{sipser_expander_1996}. In a standard code from this paradigm, codewords consist of assignments of elements of a finite field $\bF_q$ to edges of a low-degree expander graph, such that the values on edges incident to each vertex form a codeword of some ``local code.'' We choose the local codes to be Reed-Solomon codes, which satisfy the multiplication property; this multiplication property then naturally lifts to the global code as well.

However, there is a known barrier for this approach: the local codes will typically only satisfy a nontrivial multiplication property if they have (encoding) rate $<1/2$, whereas the global code will typically only have good dimension of the local codes have rate $>1/2$. To circumvent this barrier, we construct expander graphs that embed into the vector space $\bF_q^t$ for $t=\poly(q)$, in such a way that when the local codes are Reed-Solomon (even of low rate $<1/2$), the resulting global code will contain a large ``planted'' Reed-Muller code. At a high level, our graphs arise from affine lines in $\bF_q^t$ pointing in directions sampled from a pseudorandom (or more specifically, low-bias) distribution, in a manner reminiscent of the derandomized low-degree tests of \cite{ben-sasson_randomness-efficient_2003}. We prove expansion of our graphs using the trace power method (see e.g.~\cite{mohanty_explicit_2021,jeronimo_explicit_2022,paredes_expansion_2022}), and we specifically adapt the derandomization technique of \cite{jeronimo_explicit_2022} to obtain an explicit construction.

This ``planting'' approach is inspired by the works \cite{dinur_new_2023,golowich_nlts_2024}, which used related approaches for planting codewords in LDPC codes that may otherwise have vanishing dimension. Similar constructions are also sometimes called ``lifted codes'' (see e.g.~\cite{guo_new_2013,frank-fischer_locality_2017}). We remark that we construct classical Sipser-Spielman codes containing planted Reed-Muller codes, instead of directly using ordinary Reed-Muller codes, to ensure that the resulting product construction of quantum LDPC codes has good distance (see the discussion of ``transpose code'' distance in Section~\ref{sec:ssinitial}).

To address item~\ref{it:taskccz} above, we view the product of Sipser-Spielman codes as a system of local coefficients on a product of the underlying graphs. This graph product can be viewed as a cubical complex, where for instance the product of $r$ edges is an $r$-dimensional cube. Therefore this complex locally consists of $r$-dimensional cubes, i.e.~manifolds with boundaries. There are known techniques for performing transversal $C^{r-1}Z$ gates on such manifolds, such as by using $r$-dimensional color codes \cite{bombin_topological_2007}. We use related ideas to construct local $C^{r-1}Z$ gates within the cubes of our complex (though for simplicity, we do not use color codes).

We then combine these local $C^{r-1}Z$ gates into a global circuit. We use the multiplication property of the local Reed-Solomon codes to show that the resulting circuit preserves the code space of our quantum code. Meanwhile, we use the multiplication property of the planted Reed-Muller codes to show that the logical action induced by this circuit consists of $C^{r-1}Z$ gates on many disjoint tuples of logical qudits.

A related approach based on gluing together local color codes was studied in the recent paper \cite{scruby_quantum_2024}, building on prior works including \cite{bombin_topological_2007,vuillot_quantum_2022,zhu_non-clifford_2023}. Our techniques allow for more general complexes and local codes, resulting in parameter improvements (see Remark~\ref{remark:logicalstructure} for a comparison).

In a companion paper, \cite{lin_transversal_2024} shows that our approach to constructing transversal $C^{r-1}Z$ gates generalizes to a broad class of complexes based on local coefficient systems. In contrast, in the present paper, we focus on our instantiation with cubical complexes, and hence provide a less general (but more concrete) presentation.

We remark that we prove the distance and soundness bounds for the classical balanced product codes in Theorem~\ref{thm:cltcmaininf} by applying the results of \cite{polishchuk_nearly-linear_1994,dinur_expansion_2024} (the latter of which builds upon \cite{panteleev_asymptotically_2022,dinur_good_2023,leverrier_quantum_2022-1}). We are not able to also obtain near-linear distance quantum LDPC or locally testable codes with transversal $C^{r-1}Z$ gates using balanced products because as shown in \cite{kalachev_two-sided_2023}, the \textit{product-expansion} bound of \cite{polishchuk_nearly-linear_1994} (see Definition~\ref{def:prodexp}) is too weak to show good quantum code distance using existing techniques. It is an interesting open question to address this challenge.

\section{Technical Overview}
\label{sec:techoverview}
In this section, we provide an overview of our qLDPC codes with transversal $C^{r-1}Z$ gates. Our construction of cLTCs with the multiplication property uses related techniques.

As described in Section~\ref{sec:intro}, we construct our qLDPC codes with transversal $C^{r-1}Z$ gates as the tensor product (also called a ``hypergraph'' or ``homological'' product) of $r$ 2-term chain complexes associated to classical LDPC codes. We construct these classical LDPC codes using the well-known Sipser-Spielman paradigm \cite{sipser_expander_1996}, by imposing parity-checks (i.e.~linear constraints) from small local codes according to the incidence structure of a low-degree expander graph.
We specifically construct these classical LDPC codes to possess the multiplication property, which as described in Section~\ref{sec:intro}, requires component-wise products of codewords to belong to similar codes. We then show that products of such classical codes yield quantum codes with transversal $C^{r-1}Z$ gates.

Therefore we must perform the following:

\begin{enumerate}
\item\label{it:taskldpcmp} Construct classical LDPC codes with good parameters that have the multiplication property.
\item\label{it:taskmptoccz} Show that tensor products of chain complexes associated to classical LDPC codes with the multiplication property yield quantum codes with transversal $C^{r-1}Z$ gates.
\end{enumerate}

Section~\ref{sec:qldpcinf} provided a brief overview of our approach to addressing the two problems above; the following two sections provide more details.

\subsection{Classical LDPC Codes with the Multiplication Property}
\label{sec:ldpcmp}
In this section, we describe how we construct classical LDPC codes with the multiplication property using the Sipser-Spielman paradigm \cite{sipser_expander_1996}. Our key result here is a construction of spectral expander graphs with an appropriate linear-algebraic structure, which interacts well with local Reed-Solomon codes.

For a classical code $C=\ker(H)$ with specified parity-check matrix $H$, we define the \textit{transpose code} by $C^\top:=\ker(H^\top)$. In order to ultimately construct quantum LDPC codes with transversal $C^{r-1}Z$ gates via products, we want to construct classical codes $C$ for which the following hold:
\begin{enumerate}
\item $C$ satisfies an appropriate multiplication property.
\item $C$ and $C^\top$ have large distance and small locality (meaning that the rows and columns of $H$ have small Hamming weight).
\item $C$ or $C^\top$ has large dimension.
\end{enumerate}
As described above, we will use the multiplication property of $C$ construct transversal $C^{r-1}Z$ gates on the resulting product quantum codes. The distance, dimension, and locality of the quantum codes will follow from the respective properties for $C$ and $C^\top$, as was shown in \cite{tillich_quantum_2014}. Though we will ultimately need to be able to ensure that either $C$ or $C^\top$ has large dimension, in the remainder of this section for simplicity we focus on ensuring that $C$ has large dimension, and refer the reader to Section~\ref{sec:cldpc} and Section~\ref{sec:qldpc} for more details.

The requirement that the transpose code $C^\top$ has good distance is perhaps unnatural from a classical coding theoretic perspective. However, the Sipser-Spielman paradigm naturally yields LDPC codes with good distance for both $C$ and $C^\top$, as long as the local codes and their duals have good distance (see e.g.~\cite{panteleev_quantum_2022,breuckmann_balanced_2021}). For this reason, we use the Sipser-Spielman paradigm to construct construct our desired codes achieving the properties outline above.

\subsubsection{Known Obstruction to Achieving Positive Dimension}
\label{sec:ssbarrier}
Unfortunately, there is a known obstruction to achieving such a construction of Sipser-Spielman codes with positive dimension supporting the multiplication property. Recall that (the most basic form of) a Sipser-Spielman code is defined by a $\Delta$-regular expander graph $\Gamma=(V,E)$ with an assignment of a length-$\Delta$ local code $C_v$ (over some field $\bF_q$) to each vertex $v\in V$, whose components are labeled with the $\Delta$ edges incident to vertex $v$. The global Sipser-Spielman code $C$ then consists of all elements $c\in\bF_q^E$ such that for each vertex $v$, the assignment to the edges incident to $v$ lies in the local code $C_v$. \cite{sipser_expander_1996} showed that if the local codes $C_v$ have large enough distance relative to the expansion of $\Gamma$, then the global code $C$ will have large distance.

Assume that each $C_v$ satisfies the multiplication property, so that for instance $C_v*C_v$ lies inside some appropriate $C_v'\subseteq\bF_q^\Delta$, where $*$ denotes component-wise multiplication. Then the global code $C$ satisfies the multiplication property $C*C\subseteq C'$, where $C'$ is the Sipser-Spielman code with local codes $C_v'$. Therefore Sipser-Spielman codes naturally lift local multiplication properties to global multiplication properties.

Perhaps the most natural choice of a local code with a multiplication property is a Reed-Solomon code $\evl_S(\bF_q[X]^{<\ell})$, whose codewords are evaluations of degree $<\ell$ polynomials over $\bF_q$ on points in a set $S\subseteq\bF_q$; here we take $|S|=\Delta$. Indeed, by definition we have $\evl_S(\bF_q[X]^{<\ell})^{*2}=\evl_S(\bF_q[X]^{<2\ell-1})$. However, this equality is only nontrivial when $\ell\leq|S|/2=\Delta/2$, as otherwise $\evl_S(\bF_q[X]^{<2\ell-1})=\evl_S(\bF_q[X]^{<\Delta})=\bF_q^\Delta$, because every element of $\bF_q^\Delta$ is the evaluation of some degree $<\Delta$ polynomial.

However, the only known general lower bound on the dimension of the Sipser-Spielman code $C$ is $\dim(C)\geq|E|-\sum_{v\in V}(\Delta-\dim(C_v))$, which follows from counting linear constraints. When $C_v=\evl_S(\bF_q[X]^{<\ell})$, this bound is positive precisely when $\ell=\dim(C_v)>\Delta/2$. That is, for the Sipser-Spielman construction with local Reed-Solomon codes, we obtain a positive dimension bound under the precise condition $\ell>\Delta/2$ that the multiplication property becomes trivial and therefore not useful.

Note that the only property of Reed-Solomon codes we used in the reasoning above is that the dimension (approximately) adds under component-wise multiplication, so that $\dim(C_v^{*2})\geq\min\{2\cdot\dim(C_v)-1,\Delta\}$. Although this inequality does not hold for some edge cases such as $C_v=\bF_q^\ell\times\{0\}^{\Delta-\ell}$, it tends to hold for local codes with the properties we want, such as good distance.

\subsubsection{Initial Attempt to Circumvent Barrier: Reed-Muller Codes}
\label{sec:ssinitial}
In this section, we describe an initial attempt at circumventing the apparent incompatibility of positive dimension and the multiplication property described in Section~\ref{sec:ssbarrier}, namely, that Reed-Muller codes appear to solve our problem. However, we will see that if we try to directly instantiate $C$ as a Reed-Muller code, we will be unable to ensure good distance of the transpose code $C^\top$. Instead, as described in Section~\ref{sec:sscircumvent} below, we will need to let $C$ be a carefully constructed Sipser-Spielman code that contains a planted Reed-Muller code.

Recall that a codeword of a Reed-Muller code $\evl_S(\bF_q[X_1,\dots,X_t]^{<\ell})$ consists of the evaluations of a multivariate polynomial of degree $<\ell$ on points in some set $S$; here it is convenient to think of $S=\bF_q^t$. Therefore Reed-Muller codes satisfy the multiplication property by the same reasoning that Reed-Solomon codes do, as the product of two degree $<\ell$ polynomials has degree $<2\ell-1$. Furthermore, Reed-Muller codes are LDPC when the number of variables $t$ is large relative to $q$. Specifically, if $\ell<q$, then the Reed-Muller code $\evl_S(\bF_q[X_1,\dots,X_t]^{<\ell})$ has parity-checks of weight $q$, as the restriction of a multivariate polynomial of total degree $<\ell$ to an affine line is a univariate polynomial of degree $<\ell$.


However, a naive choice of parity-check matrix $H$ for a Reed-Muller code, such as one that places checks on all affine lines in $\bF_q^t$, will have many $(\tilde{\Omega}(q^t))$ checks touching each code component, and thus $C^\top=\ker(H^\top)$ will not even be LDPC. This issue is mitigated by instead choosing a parity-check matrix $H$ associated to a derandomized low-degree test, such as that of \cite{ben-sasson_randomness-efficient_2003}, which only queries a much sparser set of affine lines in $\bF_q^t$. However, the associated transpose code $C^\top=\ker(H^\top)$ may still have poor distance. For instance, inside an appropriate 2-dimensional subspace of $\bF_q^t$, the test of \cite{ben-sasson_randomness-efficient_2003} queries all $2q$ affine lines pointing in two different directions within the subspace; using this fact, we may construct a codeword in $\ker(H^\top)$ of weight $O(q^2)$, which is much smaller than the block length $|S|=q^t$. Despite this challenge, as described in Section~\ref{sec:sscircumvent} below, our ultimate construction will draw on related ideas as in \cite{ben-sasson_randomness-efficient_2003}.

\subsubsection{Circumventing the Barrier by Planting Reed-Muller Codes}
\label{sec:sscircumvent}
We now describe how we circumvent the problems described in Section~\ref{sec:ssbarrier} and Section~\ref{sec:ssinitial} to obtain classical LDPC codes $C=\ker(H)$ with the multiplication property that have good dimension, such that both $C$ and $C^\top$ have good distance. Our solution is to construct a specific $\Delta$-regular graph $\bar{\Gamma}$ on which we can ensure that an associated Sipser-Spielman code $C$ contains a large ``planted'' Reed-Muller code, when the local codes are Reed-Solomon of dimension $\ell\leq\Delta/2$. As described above, the Sipser-Spielman paradigm yields good distance for both $C$ and $C^\top$ as long as the local codes and their duals have good distance, which is indeed the case for Reed-Solomon local codes (as Reed-Solomon codes are self-dual).



Our approach is inspired by \cite{dinur_new_2023}, who used similar ideas to construct cLTCs with the multiplication property, and by \cite{golowich_nlts_2024}, who showed how to plant codewords (though in their case, only the all-1s codeword) in quantum product codes. However, \cite{dinur_new_2023} use a high-dimensional expander (HDX) constructed by \cite{kaufman_construction_2018} in the place of $\bar{\Gamma}$, which makes it difficult to obtain good distance for the associated transpose codes. As a related note, the codes of \cite{dinur_new_2023} are constructed using chain complexes and hence have a natural quantum analogue with good $Z$ distance, but it is unclear if they also can be made to have good $X$ distance given the ``one-sided'' nature of their construction from HDXs. Indeed, prior constructions of qLDPC codes from HDXs have been limited by poor distance in one of the two bases (e.g.~\cite{evra_decodable_2020,kaufman_new_2021}).


We therefore instead construct a new graph $\bar{\Gamma}=(\bar{V},\bar{E})$, which we show has good spectral expansion and yields Sipser-Spielman codes with the desired planted Reed-Muller codes. At a high level, we construct $\bar{\Gamma}$ as follows. For some small constant $\tau>0$, we choose an integer $t\approx q^\tau=\poly(q)$. We then let the vertices in $\bar{V}$ correspond to certain affine lines in $\bF_q^{t+1}$, and we let the edges in $\bar{E}$ correspond to points in $\bF_q^{t+1}$. Each edge $\bar{e}\in\bar{E}=\bF_q^{t+1}$ connects two vertices whose associated affine lines intersect at the point $\bar{e}$. (In fact $\bar{E}$ may be a subset of $\bF_q^{t+1}$, but the distinction is not important for this overview.)

Our construction of $\bar{\Gamma}$ shares similarities with the derandomized low-degree test of \cite{ben-sasson_randomness-efficient_2003}, which is also based on choosing parity-checks of Reed-Muller codes associated to certain affine lines in $\bF_q^t$; both our work and \cite{ben-sasson_randomness-efficient_2003} also choose such lines using a low-bias distribution (see Setion~\ref{sec:pseudorandom}). However, we only associate two affine lines to each point in $\bF_q^t$, whereas \cite{ben-sasson_randomness-efficient_2003} associates a growing number of affine lines to each point; this difference ultimately allows us to obtain Sipser-Spielman LDPC codes, whose transpose codes therefore have good distance.

We remark that our construction $\bar{\Gamma}$ can be viewed as a ``coset complex,'' as vertices of $\bar{\Gamma}$ correspond to affine lines, which are simply cosets of 1-dimensional subspaces. The HDXs of \cite{kaufman_construction_2018} used by \cite{dinur_new_2023} can also be viewed as coset complexes, though with a more involved group theoretic structure.

Formally, to construct $\bar{\Gamma}$, we first define a smaller ``base graph'' $\Gamma=(V,E=\bF_q)$, whose edges are each associated with a unique element of $\bF_q$ (in fact we may take $E\subseteq\bF_q$, but for simplicity in this section we assume $E=\bF_q$). It will in fact suffice to take $\Gamma$ to be a $\Delta$-regular complete bipartite graph with multiedges allowed, where $\Delta=\poly(q)$. For each $v\in V$, we choose a random label $\labV(v)\in\bF_q^t$. (Following \cite{jeronimo_explicit_2022}, we are in fact able to derandomize this construction using the low-bias distributions of \cite{jalan_near-optimal_2021}.)

We then define the vertex set of $\bar{\Gamma}$ to be the set
\begin{align*}
  \bar{V} &= \bigsqcup_{v\in V}\bF_q^{t+1}/\spn\{(1,\labV(v))\}
\end{align*}
of affine lines in $\bF_q^{t+1}$ that point in the direction of $(1,\labV(v))\in\bF_q\times\bF_q^t$ for some $v\in V$. We define the edge set
\begin{align*}
  \bar{E} &= \bF_q^{t+1},
\end{align*}
where an edge $\bar{e}=(\bar{e}_0,\dots,\bar{e}_t)\in\bF_q^{t+1}$ connects the unique two vertices associated to respective affine lines $\bar{e}+\spn\{(1,\labV(v_0))\}$ and $\bar{e}+\spn\{(1,\labV(v_1))\}$ passing through the point $\bar{e}$, for which the associated vertices $v_0,v_1\in V$ are connected by edge $\bar{e}_0$ in $\Gamma$.

Recall that a Sipser-Spielman code has large distance if the local codes have large distance, and if the graph has good (i.e.~small) spectral expansion, defined as the second largest eigenvalue of the adjacency matrix. Therefore we show the following:

\begin{theorem}[Informal statement of Corollary~\ref{cor:expinst}]
  \label{thm:expinstinf}
  For every fixed $\delta,\eta\in(0,1)$, there exists a sufficiently small $\tau>0$ such that for every sufficiently large prime power $q$, the graph $\bar{\Gamma}$ defined above is an $\eta\Delta$-spectral expander of degree $\Delta=q^\delta$, for an appropriate choice of the vertex labels $(\labV(v))_{v\in V}$ that can be found in $\poly(|\bar{V}|)$ time.
\end{theorem}

Recall that the spectral expansion of a $\Delta$-regular graph is at most $\Delta$. For our application, we will want spectral expansion $\leq\eta\Delta$ for a sufficiently small constant $\eta>0$, as is provided by Theorem~\ref{thm:expinstinf}. Note that as $|\bar{V}|,|\bar{E}|=q^{\Theta(t)}$ with $t=\poly(q)$, the degree $\Delta=\poly(q)$ is polylogarithmic in the size of $\bar{\Gamma}$.

To prove Theorem~\ref{thm:expinstinf}, we apply the trace power method, as mentioned in Section~\ref{sec:intro}. This technique has previously been used (e.g.~\cite{mohanty_explicit_2021,jeronimo_explicit_2022,paredes_expansion_2022}) to bound the expansion of certain random \textit{abelian lifts} of graphs (see Section~\ref{sec:ablift}) by counting walks on the graph. Our graph $\bar{\Gamma}$ is by definition a lift of $\Gamma$ by the abelian (additive) group $\bF_q^t$, though we place random labels $\labV(v)$ on vertices $v$ of the base graph $\Gamma$, whereas prior works typically considered random labels placed on edges of the base graph. Fortunately, we show that the trace power method be extended to our setting. This approach gives a randomized construction of $\bar{\Gamma}$ that has good expansion with high probability. We obtain an explicit (i.e.~polynomial-time computable) construction in Theorem~\ref{thm:expinstinf} by adapting the derandomization techniques of \cite{jeronimo_explicit_2022}, specifically using the low-bias distributions of \cite{jalan_near-optimal_2021}.

Given some $\ell\leq\Delta$, we can then define a Sipser-Spielman code $C$ on $\bar{\Gamma}$ as follows: we let the local code $C_{\bar{v}}$ for vertex $\bar{v}\in\bar{V}$ be a Reed-Solomon code of degree (i.e.~dimension) $\ell$ on the affine line associated to $\bar{v}$, where the Reed-Solomon evaluation points are those points on the affine line associated to edges $\bar{e}$ that are incident to vertex $\bar{v}$ in the graph $\bar{\Gamma}$.

In particular, recall that the Reed-Muller code $\evl_{\bF_q^{t+1}}(\bF_q[X_0,\dots,X_t]^{<\ell})$ has codewords given by evaluations of $(t+1)$-variate polynomials on all points in $\bF_q^{t+1}$. Because the restriction of a degree $<\ell$ multivariate polynomial to an affine line is a degree $<\ell$ univariate polynomial, all such Reed-Muller codewords lie in our Sipser-Spielman code $C$. Thus $C$ has dimension
\begin{equation*}
  K = \dim(C) = \dim(\bF_q[X_0,\dots,X_t]^{<\ell}) = {\ell+t\choose t+1} \geq \left(\ell/t\right)^{t+1}.
\end{equation*}
In particular, given an arbitrarily small constant $\epsilon>0$, if we choose $\delta=1-\epsilon/4$, $\tau<\epsilon/4$, and $\ell<\Delta/2$ with $\ell\geq\Omega(\Delta)$ so that $\ell/t\geq q^{1-\epsilon}$ (for $q$ sufficiently large), the resulting code $C$ has parameters
\begin{equation*}
  \left[N=|\bar{E}|=q^{t+1},\; K\geq N^{1-\epsilon},\; D\geq N/\poly(\log N)\right]_q,
\end{equation*}
and the transpose code $C^\top$ also has distance $D^\top\geq N/\poly(\log N)$, where the distance bounds follow from \cite{sipser_expander_1996}. By construction, $C$ and $C^\top$ are LDPC of locality $O(q)=\poly(\log N)$. Thus $C$ satisfies the desired properties outlined at the start of Section~\ref{sec:ldpcmp} above.

We remark that we actually use the local codes $C_v$ described above for our cLTC construction, and we use slightly different Reed-Solomon local codes for our qLDPC construction; see Definition~\ref{def:RMplant} for these alternative local codes. At a high level, the difference arises because in the quantum case, it is more convenient to work with a multiplication property for the transpose codes~$C^\top$.

\subsection{Quantum LDPC Codes with Transversal $C^{r-1}Z$ Gates}
\label{sec:mptoccz}
In this section, we describe the main ideas behind our proof that the product of $r$ classical LDPC codes with an appropriate multiplication property yield quantum LDPC codes with transversal $C^{r-1}Z$ gates. In this paper, we consider products of chain complexes associated to Sipser-Spielman codes similar to those defined in Section~\ref{sec:sscircumvent} above, and leverage the cubical nature of the underlying graph product. The companion paper \cite{lin_transversal_2024} applies similar ideas to present a more general framework for transversal $C^{r-1}Z$ gates on a large class of chain complexes. Our presentation here is less general but more concrete, leading to more explicit expressions for the relevant circuits of $C^{r-1}Z$ gates on our codes.

Throughout this section, we fix an arbitrary integer $r\geq 2$ and an arbitrary finite field $\bF_q$ of characteristic $p$. Prior constructions of qLDPC codes with transversal $C^{r-1}Z$ gates were primarily $r$-dimensional toplogical codes based on manifolds. Our high-level approach is to generalize such ideas to more general complexes based on Sipser-Spielman codes by leveraging the local topological structure of such complexes, and ensuring these local structures fit together globally using a multiplication property of the local codes. A similar idea of gluing together local topological codes was also used in the recent work of \cite{scruby_quantum_2024}, though with less general complexes and local codes, resulting in poor (i.e.~logarithmic) distance.

\subsubsection{Transversal $C^{r-1}Z$ Gates Preliminaries}
\label{sec:cczpreliminf}
We first briefly describe some necessary notation for discussing transversal $C^{r-1}Z$ gates on quantum CSS codes (see Definition~\ref{def:css}). If we have $r$ code states of length-$N$ codes, which form a superposition of length-$rN$ strings $z=(z^{(1)},\dots,z^{(r)})\in(\bF_q^N)^r$, then by definition (see~(\ref{eq:cczdefinf})), a $C^{r-1}Z^a$ gate acting on qudit $i_h$ of $z^{(h)}$ for each $h\in[r]$ simply induces a phase on $\ket{z}$ with exponent proportional to $a\cdot z^{(1)}_{i_1}\cdots z^{(r)}_{i_r}$. That is, such a $C^{r-1}Z^a$ gate maps
\begin{equation*}
  \ket{z} \mapsto e^{2\pi i\tr_{\bF_q/\bF_p}(a\cdot z^{(1)}_{i_1}\cdots z^{(r)}_{i_r})/p}\ket{z}.
\end{equation*}

Therefore if we apply many $C^{r-1}Z^a$ gates to different $r$-tuples of length-$N$ code states, where each gate acts on one qudit from each of the $r$ code states, the resulting unitary $U^\zeta$ acts on a basis element $\ket{z}=\ket{z^{(1)},\dots,z^{(r)}}$ by
\begin{equation*}
  U^\zeta:\ket{z} \mapsto e^{2\pi i\tr_{\bF_q/\bF_p}(\zeta(z))/p}\ket{z}
\end{equation*}
for some $r$-multilinear form
\begin{equation*}
  \zeta:(\bF_q^N)^r\rightarrow\bF_q.
\end{equation*}
Specifically, $\zeta$ can be expressed as homogeneous degree-$r$ polynomial in the $rN$ variables $(z^{(h)}_i)_{i\in[N]}^{h\in[r]}$, which is defined to have a term $a\cdot z^{(1)}_{i_1}\cdots z^{(r)}_{i_r}$ associated to every $C^{r-1}Z^a$ gate acting on qudits $\ket{z^{(1)}_{i_1}},\dots,\ket{z^{(r)}_{i_r}}$ respectively.

We are interested in the case where we have $r$ quantum CSS codes $(Q^{(h)}=(Q^{(h)}_X,Q^{(h)}_Z))_{h\in[r]}$, for which a basis of code states is given by the equally weighted superpositions $\sum_{x\in{Q^{(h)}_X}^\perp}\ket{z+x}$ over the elements of cosets $z+{Q^{(h)}_X}^\perp\in Q^{(h)}_Z/{Q^{(h)}_X}^\perp$. In order for the unitary $U^\zeta$ defined above to preserve the code space (i.e.~map code states to code states), we therefore need $\zeta$ to be invariant on cosets, meaning that
\begin{equation}
  \label{eq:cobinvinf}
  \zeta(z^{(1)}+x^{(1)},\dots,z^{(r)}+x^{(r)}) = \zeta(z^{(1)},\dots,z^{(r)})
\end{equation}
for every $(z^{(h)}\in Q^{(h)}_Z)_{h\in[r]}$ and every $(x^{(h)}\in{Q^{(h)}_Z}^\perp)_{h\in[r]}$. We call this condition in~(\ref{eq:cobinvinf}) \textit{coboundary-invariance}.

If coboundary-invariance holds, then $\zeta$ naturally induces a well-defined multilinear form $\zeta'$ on cosets
\begin{equation*}
  \zeta':(Q^{(1)}_Z/{Q^{(1)}_X}^\perp)\times\cdots\times(Q^{(r)}_Z/{Q^{(r)}_X}^\perp) \rightarrow \bF_q
\end{equation*}
given by applying $\zeta$ to arbitrary coset representatives. The number of logical $C^{r-1}Z$ gates on disjoint triples of logical (message) qudits that can be extracted from the circuit $U^\zeta$ is precisely the \textit{subrank} of $\zeta'$, which essentially measures the largest identity tensor contained within $\zeta'$ (see Definition~\ref{def:subrank}).

\subsubsection{Background on Transversal $C^{r-1}Z$ Gates on Topological Codes}
\label{sec:topcczinf}
We now briefly describe the (well known) intuition behind transversal $C^{r-1}Z$ gates on topological codes obtained by tiling $r$-dimensional manifolds. Fix such a manifold (such as an $r$-dimensional torus). Consider $r$ CSS codes $Q^{(1)},\dots,Q^{(r)}$ associated to some tilings of this manifold, such that the physical qudits (i.e.~code components) are associated to $(r-1)$-dimensional faces in the tilings. Assuming for simplicity that our local qudit dimension is $q=2$, then a codeword in $Q^{(h)}_Z$ can be intuitively viewed as a set of tiles that collectively form a closed $(r-1)$-dimensional submanifold. Meanwhile, a codeword of ${Q^{(h)}_X}^\perp$ can be viewed as such a submanifold that is the boundary of some $r$-dimensional volume within the manifold.

Now we define the multilinear form $\zeta$ to contain a nonzero term (corresponding to a physical $C^{r-1}Z$ gate) for each $r$-tuple of tiles (with one tile associated to a qubit from each of the $r$ codes) that collectively intersect at a point. Then for codewords $z^{(h)}\in Q^{(h)}_Z$, we can interpret $\zeta(z^{(1)},\dots,z^{(r)})$ as the parity of the number of collective intersection points of the $(r-1)$-dimensional submanifolds $z^{(1)},\dots,z^{(r)}$.

The coboundary-invariance condition~(\ref{eq:cobinvinf}) can then be enforced by showing that the intersection number is a topological invariant. That is, adding elements of ${Q^{(h)}_X}^\perp$ can be interpreted as performing local perturbations of the submanifold $z^{(h)}$. However, in a manifold such as an $r$-dimensional torus, one can visually see (at least for $r=2$ or $r=3$) that such local perturbations will never change the parity of the number intersections of $r$ different $(r-1)$-dimensional submanifolds. Here we must assume that the intersections are ``generic,'' so that for instance no two submanifolds are tangent to each other at some point, though this condition can be enforced by choosing different tilings for each of the $r$ codes $Q^{(h)}$.

Thus we obtain transversal $C^{r-1}Z$ gates on codes associated to tilings of such manifolds. Indeed, the standard $r$-dimensional color code (see e.g.~\cite{bombin_topological_2007}) is based on the reasoning above for the $r$-dimensional torus. We remark that the intuition described above can be formalized and generalized using the notion of a \textit{cup product} from algebraic topology.

\subsubsection{Extending to More General QLDPC Codes}
We now describe how we extend the techniques outlined in Section~\ref{sec:topcczinf} to products of Sipser-Spielman codes. We can view a Sipser-Spielman code as a graph with an associated \textit{system of local coefficients} (i.e.~local codes). The tensor product of the chain complexes associated to $r$ Sipser-Spielman codes can then be viewed as an $r$-dimensional cubical complex given by the product of the $r$ underlying graphs (where the product of $r$ edges is an $r$-dimensional cube), with a system of local coefficients given by tensor products of the local codes. For more details, see Example~\ref{example:graphproduct} in Section~\ref{sec:ccdef}.

Now we construct our desired multilinear form $\zeta$ (corresponding to our circuit of physical $C^{r-1}Z$ gates) on the associated product code as the sum of local multilinear forms within each $r$-dimensional cube of the cubical complex. Each such local form is obtained as described in Section~\ref{sec:topcczinf}, and simply computes an intersection number within the $r$-dimensional cube, viewed as an $r$-dimensional manifold with boundary. However, these boundaries pose an additional challenge to proving the coboundary-invariance property~(\ref{eq:cobinvinf}), which we address with the multiplication property of the local codes.

Specifically, if we interpret elements of ${Q^{(h)}_X}^\perp$ as local perturbations of submanifolds as described in Section~\ref{sec:topcczinf} above, then to prove coboundary-invariance, we now must ensure that when we perturb a submanifold across a boundary between $r$-dimensional cubes, we preserve the desired intersection numbers. As the elements of ${Q^{(h)}_X}^\perp$ correspond to parity-checks of $Q^{(h)}_X$, which are simply parity-checks of the local codes, the local codes (at a high level) dictate the rules for how submanifolds transform across boundaries. It turns out that if the local codes satisfy an appropriate multiplication property, then these local transformations behave in a way amenable to coboundary-invariance. We therefore obtain the following result.

\begin{theorem}[Informal statement of transversal $C^{r-1}Z$ property in Theorem~\ref{thm:qldpcmain}]
  \label{thm:cczinf}
  The product of $r$ chain complexes associated to classical Sipser-Spielman LDPC codes whose local codes satisfy an appropriate multiplication property yields a quantum LDPC code supporting a transversal $C^{r-1}Z$ gate.
\end{theorem}

We have intentionally stated Theorem~\ref{thm:cczinf} informally, as there are many technical details regarding the local codes and the structure of the product that we have omitted in this intuitive overview.

While Theorem~\ref{thm:cczinf} provides a circuit of physical $C^{r-1}Z$ gates that preserves the logical code space, we also must show that the induced logical circuit consists of $C^{r-1}Z$ gates on many disjoint $r$-tuples of logical qudits, as stated in Theorem~\ref{thm:qldpcmaininf}. That is, as described in Section~\ref{sec:cczpreliminf}, we must show that the multilinear form $\zeta'$ associated to the physical circuit $U^\zeta$ has large subrank.

For this purpose, we again use the planted Reed-Muller codes in our Sipser-Spielman codes described in Section~\ref{sec:sscircumvent}. Specifically, we are able to show that the qLDPC codes we obtain inherit large planted Reed-Muller codes from the underlying classical Sipser-Spielman codes, and that the multiplication property of these planted Reed-Muller codes directly translates into a good bound on the subrank of $\zeta'$. See Lemma~\ref{lem:subrankbound} for more details (where $\zeta'$ is denoted $\zeta_{H'}$).


\section{Preliminaries}
This section presents preliminary various preliminary notions, results, and notation that we will use throughout the paper.

\subsection{Notation}
\label{sec:notation}
We use standard ``big-$O$'' notation $O,o,\Omega,\omega,\Theta$, and use subscripts to denote variables on which the hidden constants may depend. For instance, $O_\alpha(1)$ describes any function bounded above by a value depending only on $\alpha$.

For a prime power $q$, we let $\bF_q$ denote the finite field of order $q$, and we let $\bF_q^*=\bF_q\setminus\{0\}$. We let $[n]=\{1,\dots,n\}$, and for an element $x\in\bF_q^n$, we let $|x|=|\{i\in[n]:x_i\neq 0\}|$ denote the Hamming norm of $x$. For vectors $x,y\in\bF_q^n$, we let $x\cdot y=\sum_{i\in[n]}x_iy_i\in\bF_q$ denote the standard bilinear form, and we let $x*y=(x_iy_i)_{i\in[n]}\in\bF_q^n$ denote the component-wise product. We extend this notation to subspaces, so that for $C,C'\subseteq\bF_q^n$, then $C*C'=\spn\{x*x':x\in C,x'\in C'\}\subseteq\bF_q^n$. For a matrix $M\in\bF_q^{m\times n}$, the locality refers to the maximum number of nonzero entries in any row or column.

For an event $E$, we let $\1_E$ denote the indicator random variable for $E$ occurring. We often also use this notation for deterministic events. For instance, for variables $i,j$, then $\1_{i=j}$ equals $1$ if $i=j$ and $0$ if $i\neq j$. For a set $S$ and an element $s\in S$, we let $\ind{s}\in\bF_q^S$ denote the indicator vector for element $s$. In particular, for $i\in[n]$, we let $\ind{i}\in\bF_q^{[n]}=\bF_q^n$ denote the $i$th standard basis vector.

For a complex-valued square matrix $M$, we $\spectrum(M)$ denote the multiset of eigenvalues of $M$, counting multiplicities.

We let $\cS_n$ denote the symmetric group on $n$ elements, that is, the group of permutations of $[n]$.

\subsection{Graphs and Expansion}
This section presents preliminary notions pertaining to graphs and graph expansion. To begin, we present our basic terminology and notation for graphs.

\begin{definition}
  A \textbf{directed (multi)graph} is a pair $\Gamma=(V,E,\ver)$ consisting of a vertex set $V$ and an edge set $E$ with a function $\ver:E\rightarrow V\times V$. We let $\ver_0,\ver_1:E\rightarrow V$ be the projections of $\ver$ onto the two components respectively, so that $\ver(e)=(\ver_0(e),\ver_1(e))$.

  An \textbf{undirected (multi)graph} is a pair $\Gamma=(V,E,\ver)$ consisting of a vertex set $V$ and an edge set $E$ with a function $\ver:E\rightarrow V\times V/\sim$, where $\sim$ denotes the equivalence relation $(v_0,v_1)\sim(v_1,v_0)$. We denote elements of $V\times V/\sim$ as unordered pairs $\{v_0,v_1\}=\{v_1,v_0\}$. We define $\ver_0,\ver_1:E\rightarrow V$ such that $\ver(e)=\{\ver_0(e),\ver_1(e)\}$; thus $\ver_0,\ver_1$ implicitly define an arbitrary but fixed orientation of each edge.

  For simplicity, we restrict attention to graphs with no self-loops, meaning that every $e\in E$ has $\ver_0(e)\neq\ver_1(e)$.

  We sometimes abuse notation and use the same symbol $\ver$ for different graphs, and rely on the argument to disambiguate the graph at hand. We then write $\Gamma=(V,E)$ as a shorthand for $\Gamma=(V,E,\ver)$. 

  An undirected graph $\Gamma=(V,E,\ver)$ has a naturally associated directed graph $\dir{\Gamma}=(V,\dir{E},\dir{\ver})$, where $\dir{E}=E\times\{0,1\}$, such that $\dir{\ver}:\dir{E}\rightarrow V$ is given by $\dir{\ver}(e,0)=(\ver_0(e),\ver_1(e))$ and $\dir{\ver}(e,1)=(\ver_1(e),\ver_0(e))$.

  If $\Gamma$ is directed, then for subsets $S,T\subseteq V$, we let $E(S,T)=\ver^{-1}(S\times T)\subseteq E$. If $\Gamma$ is undirected, we let $E(S,T)\subseteq E$ denote the projection of $(\dir{\ver})^{-1}(S\times T)$ onto the first coordinate. For directed or undirected $\Gamma$, we also let $E(S)=E(S,V)$. In particular, $E(v)$ denotes the set of edges containing vertex $v\in V$.

  For an undirected graph $\Gamma$, the \textbf{degree} of a vertex $v\in V$ equals $|E(v)|$, and $\Gamma$ is said to be \textbf{$\Delta$-regular} if every $v\in V$ has degree $\Delta$. We say that $\Gamma$ is \textbf{simple} if $\Gamma$ has no self-loops and if every pair $\{v_0,v_1\}$ is the image under $\ver$ of at most one edge.
  
  An (undirected or directed) graph is \textbf{bipartite} if there exists a partition $V=V_0\sqcup V_1$ such that every edge contains one vertex in each of $V_0$ and $V_1$. The graph is \textbf{balanced} if $|V_0|=|V_1|$. If $\Gamma$ is undirected, we slightly abuse notation and assume for every $e\in E$ that $\ver_b(e)\in V_b$, and we write $(v_0,v_1)$ to denote a pair $\{v_0,v_1\}$ of vertices for which $v_0\in V_0$ and $v_1\in V_1$.
\end{definition}

Unless explicitly stated otherwise, we allow arbitrarily many edges between the same two vertices, and for notational convenience call such objects ``graphs'' as a shorthand for ``multigraphs.'' We will also explicitly say when a graph is directed, and otherwise it can be assumed to be undirected.

\begin{definition}
  \label{def:expansion}
  The \textbf{adjacency matrix $A_\Gamma$} of an undirected (resp.~directed) graph $\Gamma=(V,E)$ is the matrix $A_\Gamma\in\bZ_{\geq 0}^{V\times V}\subseteq\bR^{V\times V}$ given by $(A_\Gamma)_{v_0,v_1}=|\ver^{-1}(\{v_0,v_1\})|$ (resp.~$(A_\Gamma)_{v_0,v_1}=|\ver^{-1}(v_0,v_1)|$).

  The \textbf{(onesided) spectral expansion $\lambda_2(\Gamma)$} of an undirected graph $\Gamma$ equals the second largest eigenvalue of the adjacency matrix $A_\Gamma$.
  If $\Gamma$ is an undirected $\Delta$-regular graph, then equivalently
  \begin{equation}
    \label{eq:expansionform}
    \lambda_2(\Gamma) = \max_{x\in\vec{1}^\perp\setminus\{0\}}\frac{x^\top A_\Gamma x}{\|x\|^2},
  \end{equation}
  where $\vec{1}^\perp\subseteq\bR^V$ denotes the orthogonal complement of the all-1s vector $\vec{1}\in\bR^V$.
\end{definition}

As an undirected graph $\Gamma=(V,E)$ has a symmetric adjacency matrix $A_\Gamma$, there are $|V|$ eigenvalues of $A_\Gamma$ (counting multiplicites), all of which are real. Therefore $\lambda_2(\Gamma)$ is well-defined. If $\Gamma$ is $\Delta$-regular, then the largest eigenvalue equals $\Delta$, corresponding to the all-1s vector. Thus as the sum of the eigenvalues $\tr(A_\Gamma)\geq 0$, we have $\lambda_2(\Gamma)\in[-\Delta/(|V|-1),\Delta]$.

In this paper, we say ``spectral expansion'' or simply ``expansion'' to refer to onesided spectral expansion $\lambda_2$. Note that other works sometimes use this phrase to refer to \textbf{twosided spectral expansion}, which equals the second largest \textit{absolute value} of an eigenvalue of the adjacency matrix. However, we will rarely need to consider twosided expansion.

The following property of spectral expanders is well known; a proof can for instance be found in \cite[Lemma~2.1]{hsieh_explicit_2024}.

\begin{lemma}[Expander Mixing Lemma]
  \label{lem:expmix}
  Let $\Gamma=(V,E,\ver)$ be a $\Delta$-regular undirected graph. For every set $S\subseteq V$,
  \begin{equation*}
    |\dir{E}(S,S)| = 2|E(S,S)| \leq |S|\cdot\left(\lambda_2(\Gamma)+\Delta\cdot\frac{|S|}{|V|}\right).
  \end{equation*}
\end{lemma}

\subsection{Abelian Lifts of Graphs}
\label{sec:ablift}
Below, we present the notion of the lift of a graph by an abelian group. A similar definition exists for non-abelian groups, but we will not need such generality.

\begin{definition}
  \label{def:graphlift}
  Let $G$ be an abelian group, and let $\Gamma=(V,E,\ver)$ be an undirected graph with a $G$-valued labeling $\liftlab:E\rightarrow G$. Then the \textbf{$G$-lift of $\Gamma$ with labeling $\liftlab$} is the graph $\tilde{\Gamma}=(\tilde{V},\tilde{E},\tilde{\ver})$ with vertex set $\tilde{V}=V\times G$ and edge set $\tilde{E}=E\times G$ such that
  \begin{align*}
    \tilde{\ver}_0(e,g) &= (\ver_0(e),g) \\
    \tilde{\ver}_1(e,g) &= (\ver_1(e),g+\liftlab(e)).
  \end{align*}
\end{definition}

Below, we state a well-known lemma showing that the spectrum of $A_{\tilde{\Gamma}}$ is the union of the spectra of certain signed (or more precisely, phase-adjusted) versions of $A_\Gamma$. We first must formally define these signings.

\begin{definition}
  A \textbf{character} of a finite abelian group $G$ is a homomorphism $\chi:G\rightarrow\{z\in\bC:|z|=1\}$, meaning that $\chi(x+y)=\chi(x)\chi(y)$ for every $x,y\in G$. The \textbf{trivial character $\chi_0$} is given by $\chi_0(x)=1$ for every $x\in G$.
\end{definition}

The following fact is well known:

\begin{lemma}
  \label{lem:charbasic}
  A finite abelian group $G$ has precisely $|G|$ distinct characters. In particular, if $G$ is the additive group $\bF_p^n$, then the characters are indexed by elements of $\bF_p^n$, where for $a\in\bF_p^n$ we have $\chi_a(x)=e^{2\pi i(a\cdot x)/p}$.
\end{lemma}

We now define signings of adjacency matrices by characters:

\begin{definition}
  \label{def:signedadjmat}
  Let $\Gamma=(V,E,\ver)$ be an undirected graph with labeling $\liftlab:E\rightarrow G$. For a character $\chi:G\rightarrow\bC$, we define the \textbf{$\chi$-signed adjacency matrix $A_{\Gamma,\chi}\in\bC^{V\times V}$} by
  \begin{equation}
    \label{eq:signedadjmat}
    (A_{\Gamma,\chi})_{v_0,v_1} = \sum_{(e,b)\in \dir{E}:\ver(e,b)=(v_0,v_1)}\chi((-1)^b\liftlab(e)).
  \end{equation}
  We let $\lambda(A_{\Gamma,\chi})\in\bR_{\geq 0}$ denote the largest absolute value of any eigenvalue of $A_{\Gamma,\chi}$.
\end{definition}

The following lemma is well known; we provide a proof for completeness.

\begin{lemma}
  \label{lem:signedadjherm}
  The matrix $A_{\Gamma,\chi}$ in Definition~\ref{def:signedadjmat} is Hermitian.
\end{lemma}
\begin{proof}
  By definition
  \begin{align*}
    (A_{\Gamma,\chi}^\dagger)_{v_0,v_1}
    &= (A_{\Gamma,\chi})_{v_1,v_0}^* \\
    &= \sum_{(e,b)\in\dir{E}:\ver(e,b)=(v_1,v_0)}\chi((-1)^b\liftlab(e))^* \\
    &= \sum_{(e,b)\in\dir{E}:\ver(e,b)=(v_1,v_0)}\chi((-1)^{1-b}\liftlab(e)) \\
    &= \sum_{(e,b')\in\dir{E}:\ver(e,b')=(v_0,v_1)}\chi((-1)^{b'}\liftlab(e)) \\
    &= (A_{\Gamma,\chi})_{v_0,v_1},
  \end{align*}
  where the third equality above holds beacuse $\chi$ is a homomorphism, and the fourth equality follows by letting $b'=1-b$.
\end{proof}

As mentioned above and stated in the lemma below, the spectrum of $A_{\tilde{\Gamma}}$ is the union over all characters $\chi$ of the spectrum of~$A_{\Gamma,\chi}$. For completeness, we provide a proof, which is based on the fact that the (appropriately normalized) characters of $G$ form an orthonormal basis for $\bC^G$.

\begin{lemma}[Well known]
  \label{lem:specunion}
  Let $\tilde{\Gamma}=(\tilde{V},\tilde{E},\tilde{\ver})$ be a $G$-lift of $\Gamma=(V,E,\ver)$ with labeling $\liftlab$. Then
  \begin{equation}
    \label{eq:specuniongen}
    \spectrum(A_{\tilde{\Gamma}}) = \bigsqcup_{\text{characters }\chi:G\rightarrow\bC}\spectrum(A_{\Gamma,\chi}).
  \end{equation}
\end{lemma}
\begin{proof}
  For every character $\chi:G\rightarrow\bC$ and every vertex $v\in V$, define a unit vector $x^\chi_v\in\bC^{\tilde{V}}=\bC^{V\times G}$ by
  \begin{equation*}
    (x^\chi_v)_{(v',g)} = \begin{cases}
      \chi(g)/\sqrt{|G|},&v=v'\\
      0,&v\neq v'.
    \end{cases}
  \end{equation*}
  Then for every $v_0,v_1\in V$ and for every pair of characters $\chi_0,\chi_1$,
  \begin{equation*}
    {x_{v_0}^{\chi_0}}^\dagger x_{v_1}^{\chi_1} = \1_{v_0=v_1}\cdot\1_{\chi_0=\chi_1}.
  \end{equation*}
  Therefore by Definition~\ref{def:graphlift}, for every $e\in E$, letting $(v_0,v_1)=\ver(e)$, then
  \begin{align*}
    {x^{\chi_0}_{v_0}}^\dagger A_{\tilde{\Gamma}} x^{\chi_1}_{v_1}
    &= \sum_{(e,b)\in \dir{E}:\ver(e,b)=(v_0,v_1)} \sum_{g\in G}\frac{\chi_0(g)^*}{\sqrt{|G|}}\cdot\frac{\chi_1(g+(-1)^b\liftlab(e))}{\sqrt{|G|}} \\
    &= \1_{\chi_0=\chi_1}\cdot\sum_{(e,b)\in \dir{E}:\ver(e,b)=(v_0,v_1)}\chi_0((-1)^b\liftlab(e)).
  \end{align*}
  That is, when expressed in the orthonormal basis $\{x_v^\chi\}_{v,\chi}$ of $\bC^{\tilde{V}}$, the matrix $A_{\tilde{\Gamma}}$ is block diagonal, with $|G|$ blocks given by the matrices $A_{\Gamma,\chi}$ in Definition~\ref{def:signedadjmat} for all $|G|$ characters $\chi:G\rightarrow\bC$. Thus~(\ref{eq:specuniongen}) holds.
\end{proof}

Lemma~\ref{lem:tracebound} below provides a well-known means for bounding the spectrum of $A_{\Gamma,\chi}$ in terms of walks on $\Gamma$; we provide a brief proof for completeness. We first need the following definition.

\begin{definition}
  \label{def:walk}
  Let $\Gamma=(V,E,\ver)$ be an undirected graph. We say that a sequence of $2k$ directed edges $((e_1,b_1),\dots,(e_{2k},b_{2k}))\in(\dir{E})^{2k}$ is a \textbf{length-$2k$ walk on $\Gamma$} if it holds for every $i\in[2k]$ that $\dir{\ver}_1(e_i,b_i)=\dir{\ver}_0(e_{i+1},b_{i+1})$, where we let $(e_{2k+1},b_{2k+1})=(e_1,b_1)$.
\end{definition}

\begin{lemma}[Well known]
  \label{lem:tracebound}
  Let $\Gamma=(V,E,\ver)$ be an undirected graph with labeling $\liftlab:E\rightarrow G$. Then for every character $\chi:G\rightarrow\bC$,
  \begin{equation*}
    \lambda(A_{\Gamma,\chi})^{2k} \leq \tr(A_{\Gamma,\chi}^{2k}) = \sum_{\text{walks }(e_1,b_1),\dots,(e_{2k},b_{2k})}\prod_{i=1}^{2k}\chi((-1)^{b_i}\liftlab(e_i)).
  \end{equation*}
\end{lemma}
\begin{proof}
  Lemma~\ref{lem:signedadjherm} implies that all eigenvalues of $A_{\Gamma,\chi}$ are real, so all eigenvalues of $A_{\Gamma,\chi}^{2k}$ must be nonnegative real, one of which by definition equals $\lambda(A_{\Gamma,\chi})^{2k}$. Therefore $\lambda(A_{\Gamma,\chi})^{2k} \leq \tr(A_{\Gamma,\chi}^{2k})$.

  Now by definition,
  \begin{align*}
    \tr(A_{\Gamma,\chi}^{2k})
    &= \sum_{(v_1,\dots,v_{2k})\in V^{2k}}\prod_{i=1}^{2k}(A_{\Gamma,\chi})_{v_i,v_{i+1}},
  \end{align*}
  where we let $v_{2k+1}=v_1$. Expanding each matrix element in the product above using~(\ref{eq:signedadjmat}) then gives the second inequality in the lemma statement.
\end{proof}

\subsection{Pseudorandomness}
\label{sec:pseudorandom}
We will consider abelian lifts with labelings chosen according to some random distribution. To make these lifts explicit, we want to ensure that the distribution has at most polynomially large support in the size of the lifted graph, so that we can deterministically loop through all of elements in the support. For this purpose, we will use the following notion of low-bias distributions.

\begin{definition}
  Let $G$ be an abelian group. The \textbf{bias} of a probability distribution $\cD$ over $G$ is the maximum value of $|\bE_{g\sim G}[\chi(g)]|$ over all nontrivial characters $\chi:G\rightarrow\bC$.
\end{definition}

By definition the uniform distribution is $0$-biased, but has support size $G$. We follow \cite{jeronimo_explicit_2022} and use the following result of \cite{jalan_near-optimal_2021}.

\begin{theorem}[\cite{jalan_near-optimal_2021}]
  \label{thm:lowbias}
  Given $\beta>0$, $t\in\bN$, and a generating set of an abelian group $G$, there is a $\poly(|G|^t,1/\beta)$-time algorithm that outputs a multiset $S\subseteq G^t$ of size
  \begin{equation*}
    |S| \leq O\left(\frac{t\log(|G|)^{O(1)}}{\beta^{2+o(1)}}\right)
  \end{equation*}
  for which $\Unif(S)$ has bias $\leq\beta$. Here the $O$'s hide absolute constants, and $o(1)\rightarrow 0$ as $\beta\rightarrow 0$.
\end{theorem}

\subsection{Classical Codes}
In this section we describe relevant basic notions from classical coding theory.

\begin{definition}
  A \textbf{classical linear code of length $n$ and dimension $k$ over alphabet $\bF_q$} is a $k$-dimensional linear subspace $C\subseteq\bF_q^n$. The \textbf{distance} of $C$ is $d=\min_{c\in C\setminus\{0\}}|c|$. We summarize these parameters by saying that $C$ is an $[n,k,d]_q$ code. The \textbf{dual code} of $C$ is $C^\perp=\{x\in\bF_q^n:x\cdot c=0\forall c\in C\}$.
\end{definition}

In this paper we restrict attention to classical codes that are linear, and call them ``classical codes'' or simply just ``codes'' when clear from context.

We will make use of classical codes given by evaluations of polynomials. To begin, we present some basic notation for polynomials and their evaluations.

\begin{definition}
  For $t\in\bN$, $L\subseteq\bZ_{\geq 0}^t$, and $Z\subseteq\bF_q^t$, we let $\bF_q[X_1,\dots,X_t]_Z^{L}$ denote the space of $t$-variate polynomials over $\bF_q$ that lie in the span of the monomials of the form $X_1^{i_1}\cdots X_t^{i_t}$ for $(i_1,\dots,i_t)\in L$, and that have roots at all points in $Z$. If $L=\{(i_1,\dots,i_t)\in\bZ_{\geq 0}^t:i_1+\cdots+i_t<m\}$, we write $\bF_q[X_1,\dots,X_t]_Z^{L}=\bF_q[X_1,\dots,X_t]_Z^{<m}$. If $L=\bZ_{\geq 0}^t$, we may omit $L$ from the notation, and if $Z=\emptyset$ we may omit $Z$. Thus $\bF_q[X_1,\dots,X_t]$ denotes the space of all $t$-variate polynomials over $\bF_q$. For $E\subseteq\bF_q^t$ and $\alpha\in(\bF_q^*)^E$, we define $\evl_{E,\alpha}:\bF_q[X_1,\dots,X_t]\rightarrow\bF_q^E$ by $\evl_{E,\alpha}(f)=(\alpha_ef(e))_{e\in E}$. If $\alpha=\vec{1}$ is the all-1s vector, we may omit $\alpha$ from this notation.
\end{definition}

In this paper, we refer to codes of the form $\evl_{E,\alpha}(\bF_q[X]^{<m}_Z)$ as \textbf{Reed-Solomon} codes. Note that sometimes in the literature, the bare term ``Reed-Solomon'' is reserved for the case where $E=\bF_q$, $\alpha=\vec{1}$, and $Z=\emptyset$. The term ``punctured Reed-Solomon'' codes is often used for the case where $E\subsetneq\bF_q$, and the term ``generalized Reed-Solomon'' is often used when $\alpha\neq\vec{1}$. \cite{krishna_towards_2019} use refer to codes with $Z\neq\emptyset$ as ``shortened Reed-Solomon codes.'' However, in this paper we will simply refer to all of these codes as ``Reed-Solomon.''

In the multivariate case, we refer to codes of the form $\evl_{E,\alpha}(\bF_q[X_1,\dots,X_t]^{<m})$ as \textbf{Reed-Muller} codes.

We now describe some basic properties of Reed-Solomon codes and their duals.

\begin{lemma}[Well known]
  \label{lem:RSparameters}
  For every $E\subseteq\bF_q$, $\alpha\in(\bF_q^*)^E$, $0\leq m\leq|E|$, then $\evl_{E,\alpha}(\bF_q[X]^{<m})$ is an $[|E|,m,|E|-m+1]_q$ code. Furthermore, there exists $\beta\in(\bF_q^*)^E$ such that $\evl_{E,\alpha}(\bF_q[X]^{<m})^\perp=\evl_{E,\beta}(\bF_q[X]^{<|E|-m})$.
\end{lemma}

Note that in Lemma~\ref{lem:RSparameters}, $\beta$ is chosen as some nonzero element of the 1-dimensional space $\evl_{E,\alpha}(\bF_q[X]^{<|E|-1})^\perp$.

Lemma~\ref{lem:RSparameters} does not characterize the properties of the code $\evl_{E,\alpha}(\bF_q[X]^{<m}_Z)$ for $Z\neq\emptyset$. For this purpose, we use the following basic fact, a proof of which can for instance be found in \cite[Proof of Theorem~3.1]{golowich_asymptotically_2024}.

\begin{lemma}[Well known]
  \label{lem:shortdual}
  For a classical code $C\subseteq\bF_q^n$ and a subset $A\subseteq[n]$, letting $A^c=[n]\setminus A$, then it holds that
  \begin{equation*}
    ((C\cap\{0\}^A\times\bF_q^{A^c})|_{A^c})^\perp = C^\perp|_{A^c}.
  \end{equation*}
\end{lemma}

It follows from the above lemmas that by choosing appropriate coefficients $\beta$, we can always assume that $Z=\emptyset$, as shown below.

\begin{corollary}
  \label{cor:SRStoRS}
  For every $E\subseteq\bF_q$, $\alpha\in(\bF_q^*)^E$, $Z\subseteq\bF_q\setminus E$, $|Z|\leq m\leq|E|$, there exists $\beta\in(\bF_q^*)^E$ such that $\evl_{E,\alpha}(\bF_q[X]^{<m}_Z)=\evl_{E,\beta}(\bF_q[X]^{<m-|Z|})$.
\end{corollary}
\begin{proof}
  Let $\alpha'\in(\bF_q^*)^{E\cup Z}$ be any vector with $\alpha'|_E=\alpha$. Then by Lemma~\ref{lem:shortdual}, $\evl_{E,\alpha}(\bF_q[X]^{<m}_Z)^\perp = \evl_{E\sqcup Z,\alpha'}(\bF_q[X]^{<m})^\perp|_E$. By Lemma~\ref{lem:RSparameters}, there exists $\gamma'\in(\bF_q^*)^{E\sqcup Z}$ such that $\evl_{E\sqcup Z,\alpha'}(\bF_q[X]^{<m})^\perp=\evl_{E\sqcup Z,\gamma'}(\bF_q[X]^{<|E|+|Z|-m})$. Combining these equalities and letting $\gamma=\gamma'|_E\in(\bF_q^*)^E$ gives that $\evl_{E,\alpha}(\bF_q[X]^{<m}_Z)^\perp=\evl_{E,\gamma}(\bF_q[X]^{<|E|+|Z|-m})$. Now again applying Lemma~\ref{lem:RSparameters}, there exists $\beta\in(\bF_q^*)^E$ such that $\evl_{E,\gamma}(\bF_q[X]^{<|E|+|Z|-m})^\perp=\evl_{E,\beta}(\bF_q[X]^{<m-|Z|})$. Thus $\evl_{E,\alpha}(\bF_q[X]^{<m}_Z)=\evl_{E,\beta}(\bF_q[X]^{<m-|Z|})$, as desired.
\end{proof}

Reed-Solomon codes by definition satisfy the following multiplication property. Recall from Section~\ref{sec:notation} that for $c,c'\in\bF_q^n$, then $c*c'=(c_ic'_i)_{i\in[n]}\in\bF_q^n$ denotes the component-wise product.

\begin{lemma}
  \label{lem:RSmult}
  It holds for every prime power $q$, every set $E\subseteq\bF_q$, and every $m,r\in\bN$ that $\evl_E(\bF_q[X]^{\leq m})^{*r}=\evl_E(\bF_q[X]^{\leq rm})$.
\end{lemma}
\begin{proof}
  For $f,g\in\bF_q[X]$, by definition $\evl_E(fg)=\evl_E(f)*\evl_E(g)$. Therefore $\evl_E(\bF_q[X]^{\leq m})^{*r}\subseteq\evl_E(\bF_q[X]^{\leq rm})$ because $\deg(fg)=\deg(f)+\deg(g)$, and $\evl_E(\bF_q[X]^{\leq m})^{*r}\supseteq\evl_E(\bF_q[X]^{\leq rm})$ beacuse every polynomial of degree $\leq rm$ can be expressed as a linear combination of monomials $X^i$ for $i\leq rm$, and each such monomial is by definition a product of $r$ monomials of degree $\leq m$.
\end{proof}

\subsection{Quantum Codes}
In this section we describe relevant notions from quantum coding theory.

\begin{definition}
  \label{def:css}
  A \textbf{quantum CSS code of length $n$ and dimension $k$ over alphabet $\bF_q$} is a pair $Q=(Q_X,Q_Z)$ of subspaces $Q_X,Q_Z\subseteq\bF_q^n$ such that $Q_X^\perp\subseteq Q_Z$ (so that also $Q_Z^\perp\subseteq Q_X$) such that $k=\dim(Q_Z)-\dim(Q_X^\perp)$. The \textbf{distance} of $C$ is
  \begin{equation*}
    d = \min_{c\in(Q_X\setminus Q_Z^\perp)\cup(Q_Z\setminus Q_X^\perp)}|c|.
  \end{equation*}
  We summarize these parameters by saying that $Q$ is an $[[n,k,d]]_q$ code.
\end{definition}

In this paper we restrict attention to quantum codes given by the CSS framework described above, and call them ``quantum codes'' or simply just ``codes'' when clear from context.

\subsection{Chain Complexes}
\label{sec:ccdef}
We will use the language of chain complexes with systems of local coefficients similarly to \cite{panteleev_asymptotically_2022} (also called ``sheaves'' in \cite{first_couboundary_2022,dinur_expansion_2024}) to describe the codes we study. To begin, we define chain complexes below. Note that here we only provide definitions that are sufficiently general for our purpose.

\begin{definition}
  For $r\in\bN$, an \textbf{$r$-dimensional chain complex $\cC_*$ over a ring $R$} is a free $R$-module $\cC_*=\bigoplus_{i=0}^r\cC_i$ with an $R$-module homomorphism $\partial^{\cC}:\cC_*\rightarrow\cC_*$ satisfying $(\partial^{\cC})^2=0$ and $\partial^{\cC}\cC_i\subseteq\cC_{i-1}$. We call $\cC_i$ the space of \textbf{$i$-chains}, $\partial^{\cC}$ the \textbf{boundary map}, and $\partial^{\cC}_i=\partial^{\cC}|_{\cC_i}$ the \textbf{$i$-boundary map}. For $i\notin\{0,\dots,r\}$, we write $\cC_i=0$ and $\partial^{\cC}_i=0$. We also define
  \begin{align*}
    Z_i(\cC) &= \ker(\partial^{\cC}_i) \text{ the space of \bf $i$-cycles} \\
    B_i(\cC) &= \im(\partial^{\cC}_{i+1}) \text{ the space of \bf $i$-boundaries} \\
    H_i(\cC) &= Z_i(\cC)/B_i(\cC) \text{ the \bf $i$-homology group}.
  \end{align*}

  The \textbf{cochain complex $\cC^*$} associated to $\cC_*$ is the chain complex with free $R$-module $\cC^*=\bigoplus_{i=0}^r\cC^i$ for $\cC^i=\Hom_R(\cC_i,R)$, with boundary map $(\partial^{\cC})^\top:\cC^*\rightarrow\cC^*$. We call $\cC^i$ the space of \textbf{$i$-cochains}, $\delta^{\cC}=(\partial^{\cC})^\top$ the \textbf{coboundary map} and $\delta^{\cC}_i=(\partial^{\cC}_{i+1})^\top:\cC^i\rightarrow\cC^{i+1}$ the \textbf{$i$-coboundary map}. We also define
  \begin{align*}
    Z^i(\cC) &= \ker(\delta^{\cC}_i) \text{ the space of \bf $i$-cocycles} \\
    B^i(\cC) &= \im(\delta^{\cC}_{i-1}) \text{ the space of \bf $i$-coboundaries} \\
    H^i(\cC) &= Z^i(\cC)/B^i(\cC) \text{ the \bf $i$-cohomology group}.
  \end{align*}

  When the (co)chain complex is clear from context, we write $\partial=\partial^{\cC}$ and $\delta=\delta^{\cC}$. We restrict attention to chain complexes over $R=\bZ$ and $R=\bF_q$ in which $\cC_*$ is finite-dimensional, and we assume each $\cC_i$ has some fixed basis $C(i)$. Therefore $\cC_i\cong R^{C(i)}\cong\cC^i$ and hence we can write $C=\bigsqcup_{i=0}^rC(i)$ and $\cC_*\cong\cC^*=\cC:=R^C$.
  
  The \textbf{locality $w^{\cC}$} of $\cC_*$ refers to the locality of the matrix $\partial\in R^{C\times C}$ (see Section~\ref{sec:notation}).
\end{definition}

We will consider chain complexes obtained by imposing some local structures over $\bF_q$ on top of an incidence chain complex over $\bZ$, as defined below.

\begin{definition}
  \label{def:inccomplex}
  An \textbf{incidence chain complex} is a chain complex $\cX$ over $\bZ$ such that every entry of the coboundary map matrix $\delta^{\cX}\in\bZ^{\cX\times\cX}$ lies in $\{-1,0,1\}$. For $x\in X(i):$ and $x'\in X(i+1)$, the matrix entry $\delta^{\cX}_{x',x}\in\{-1,0,1\}$ is called the \textbf{incidence number} of $x',x'$. We write $x'\triangleright x$ if $\delta^{\cX}_{x',x}\neq 0$. Then for $x_{i'}\in X(i')$ and $x_i\in X(i)$ with $i'>i$, we write $x_{i'}\succ x_i$ if there exist $x_j\in X(j)$ for $i'>j>i$ such that $x_{j+1}\succ x_j$ for each $i'>j\geq i$. We also write $X_{\preceq x}=\{x'\in X:x'\preceq x\}$.
\end{definition}

\begin{remark}
  The precedence relation in Definition~\ref{def:inccomplex} endows $X=\bigsqcup_{i=1}^rX(i)$ the structure of a \textbf{rank-$r$ graded poset}. Specifically, define $\rank:X\rightarrow\bZ$ by $\rank(x)=i$ for every $x\in X(i)$. Then the following conditions are by definition satisfied:
  \begin{enumerate}
  \item If $x\prec x'$ then $\rank(x)<\rank(x')$.
  \item If $x\prec x'$ and there is no $x''$ with $x\prec x''\prec x'$, then $\rank(x)+1=\rank(x')$.
  \end{enumerate}
  For such posets, we often say ``dimension'' instead of ``rank,'' as in the constructions considered in this paper, rank-$i$ elements will have a natural interpretation as $i$-dimensional hypercubes.
\end{remark}

Graphs provide a particularly simply class of incidence complexes:

\begin{definition}
  \label{def:graphinc}
  For an undirected graph $\Gamma=(V,E,\ver)$, the associated 1-dimensional incidence cochain complex $\cX^*$ is given by $X(0)=V$, $X(1)=E$, and
  \begin{equation*}
    \delta^{\cX}_{e,v} = \begin{cases}
      1,&v=\ver_0(e)\\
      -1,&v=\ver_1(e)\\
      0,&v\notin\ver(e).
    \end{cases}
  \end{equation*}
\end{definition}

We will later obtain higher-dimensional incidence complexes by taking products of graphs.

\begin{definition}
  A \textbf{system of local coefficients $\cF$ over alphabet $\bF_q$} for an $r$-dimensional incidence complex $\cX$ consists of an assignment of an $\bF_q$-vector space $\cF_x$ to each $x\in X$, along with a linear map $\cF_{x'\leftarrow x}:\cF_x\rightarrow\cF_{x'}$ to each $x'\succ x$ with the property that $\cF_{x''\leftarrow x'}\circ\cF_{x'\leftarrow x}=\cF_{x''\leftarrow x}$ for every $x''\succ x'\succ x'$.

  We then define an associated $r$-dimensional cochain complex $\cF_{\cX}^*$ with $i$-cochain space $\cF_{\cX}^i=\bigoplus_{x\in X(i)}\cF_x$, and with $i$-coboundary map given by $(\delta^{\cF_{\cX}}_i(f))_{x'}=\sum_{x\triangleleft x'}\delta^{\cX}_{x',x}\cF_{x'\leftarrow x}f_x$.

  We also define an \textbf{$X$-Hamming norm} $|\cdot|_X:\cF_{\cX}^*\rightarrow\bZ_{\geq 0}$ by $|f|_X=|\{x\in X:f_x\neq 0\}|$.
\end{definition}

We typically assume each $\cF_x$ has some fixed basis with respect to which we can define Hamming norm, but the choice of basis will not affect our results.

We now describe how classical and quantum codes are naturally associated to (co)chain complexes of dimension $\geq 1$ and $\geq 2$, respectively.

\begin{definition}
  \label{def:cctocode}
  Let $\cC^*$ be a cochain complex. The \textbf{classical code associated to level $i$ of $\cC^*$} is $Z^i(\cC)\subseteq\cC^i$, and the \textbf{classical code associated to level $i$ of $\cC_*$} is $Z_i(\cC)\subseteq\cC_i$. These two respective codes are said to be \textbf{low-density parity-check (LDPC) of locality $w$} if the maps $\delta_i$ and $\partial_i$ respectively have locality $w$ (see Section~\ref{sec:notation}).

  Similarly,the \textbf{quantum code associated to level $i$ of $\cC^*$} is $(Z_i(\cC),Z^i(\cC))$, and the \textbf{quantum code associated to level $i$ of $\cC_*$} is $(Z^i(\cC),Z_i(\cC))$. These codes are said to be \textbf{LDPC of locality $w$} if both maps $\delta_i$ and $\partial_i$ have locality $\leq w$.
\end{definition}

The quantum CSS codes above are well defined because the coboundary map condition $\delta_i\delta_{i-1}=0$ implies (and in fact, is equivalent to) the CSS orthogonality condition $\ker(\partial_i)^\perp=\im(\delta_{i-1})\subseteq\ker(\delta_i)$. The quantum codes associated to level $i$ of $\cC^*$ and $\cC_*$ by definition have dimension equal to $\dim H^i(\cC)=\dim H_i(\cC)$.

Naturally, if we are given codes with chosen parity-check matrices, we can also recover chain complexes:

\begin{definition}
  \label{def:codetocc}
  Let $C$ be a length-$n$ classical code with chosen parity-check matrix $H\in\bF_q^{m\times n}$, so that $C=\ker(H)$. Then the \textbf{associated 1-dimensional cochain complex} is given by $\cC^0=\bF_q^n,\; \cC^1=\bF_q^m,\; \delta_0=H$.
  
  Let $Q=(Q_X,Q_Z)$ be a length-$n$ quantum CSS code with chosen parity check matrices $H_X\in\bF_q^{m_X\times n}$, $H_Z\in\bF_q^{m_Z\times n}$, so that $Q_X=\ker(H_X)$, $Q_Z=\ker(H_Z)$. The \textbf{associated 2-dimensional cochain complex} is given by $\cC^0=\bF_q^{m_X},\; \cC^1=\bF_q^n,\; \cC^2=\bF_q^{m_Z},\; \delta_0=H_X^\top,\; \delta_1=H_Z$.
\end{definition}

We now define the analogue of code distance for chain complexes, along with a notion of testability.

\begin{definition}
  Let $\cC$ be a chain complex. The \textbf{$i$-systolic distance $d_i(\cC)$} and \textbf{$i$-cosystolic distance $d^i(\cC)$} are defined as
  \begin{align*}
    d_i(\cC) &= \min_{c\in Z_i(\cC)\setminus B_i(\cC)}|c| \\
    d^i(\cC) &= \min_{c\in Z^i(\cC)\setminus B^i(\cC)}|c|.
  \end{align*}
  The \textbf{$i$-cycle expansion} $\rho_i(\cC)$ and \textbf{$i$-cocycle expansion $\rho^i(\cC)$} are defined as
  \begin{align*}
    \rho_i(\cC) &= \min_{c\in\cC_i\setminus Z_i(\cC)}\max_{c'\in Z_i(\cC)}\frac{|\partial_i(c)|}{|c-c'|} \\
    \rho^i(\cC) &= \min_{c\in\cC^i\setminus Z^i(\cC)}\max_{c'\in Z^i(\cC)}\frac{|\delta_i(c)|}{|c-c'|}
  \end{align*}
  If $\cC=\cF_{\cX}$ is the chain complex associated to a system of local coefficients over an incidence complex, we let $d_i(\cF_{\cX};X),d^i(\cF_{\cX};X),\rho_i(\cF_{\cX};X),\rho^i(\cF_{\cX};X)$ denote the expressions above with the standard Hamming norm $|\cdot|$ replaced by the $X$-Hamming norm $|\cdot|_X$.
\end{definition}

If $\cC$ is $r$-dimensional, then the distance of the classical code associated to level $0$ of $\cC^*$ is $d^0(\cC)$, and the distance of the classical code associated to level $r$ of $\cC_*$ is $d_r(\cC)$. Similarly, the distance of the quantum code associated to level $i$ of $\cC^*$ or $\cC_*$ is by definition $\min\{d_i(\cC),d^i(\cC)\}$. The following definition shows that this analogy extends to a property called local testability, which quantifies how well local queries can measure distance to the code.

\begin{definition}
  Let $\cC$ be a chain complex. We say that the classical code associated to level $i$ of $\cC^*$ (resp.~$\cC_*$) is \textbf{locally testable with soundness $\rho^i(\cC)$ (resp.~$\rho_i(\cC)$)}. Similarly, we say that the quantum code associated to level $i$ of $\cC^*$ (and of $\cC_*$) is \textbf{locally testable with soundness $\min\{\rho_i(\cC),\rho^i(\cC)\}$}.
\end{definition}

The following definition shows how the Sipser-Spielman classical LDPC codes \cite{sipser_expander_1996} can be described as a 1-dimensional chain complex obtained by imposing local coefficients from a classical code on the incidence complex of an expander graph.

\begin{definition}
  \label{def:sscode}
  Let $\Gamma=(V,E,\ver)$ be an undirected $\Delta$-regular graph with associated incidence complex $\cX$ (see Definition~\ref{def:graphinc}). Given $m\in[\Delta]$ and an assignment $(h_v\in\bF_q^{m\times E(v)})_{v\in V}$ of $m\times\Delta$ full-rank matrices to the vertices in $V$, we define a system of local coefficients $\cF$ on $\cX$ as follows: For every $v\in V$ and $e\in E$, let $\cF_v=\bF_q^m$, $\cF_e=\bF_q$ and $\cF_{e\leftarrow v}=\ind{e}^\top h_v^\top$.
\end{definition}

We typically consider the 1-dimensional chain complex $\cF_{\cX}$ from Definition~\ref{def:sscode} in the regime where $\Delta=|E(v)|\ll|V|$. In this case, \cite{sipser_expander_1996} showed that the classical code $Z_1(\cF_{\cX})\subseteq\bF_q^E$ associated to level $1$ of ${\cF_{\cX}}_*$ has good distance if $\Gamma$ is a good spectral expander, and the matrices $h_v$ are such that the ``local codes'' $\ker(h_v)\subseteq\bF_q^{E(v)}$ have good distance. A similar proof also implies that the classical code $Z^0(\cF_{\cX})$ associated to level $0$ of $\cF_{\cX}^*$ has good distance if $\Gamma$ is a good spectral expander and the dual local codes $\im(h_v^\top)=\ker(h_v)^\perp\subseteq\bF_q^{E(v)}$ have good distance. Below we formally state and prove this result for completeness.

\begin{lemma}[\cite{sipser_expander_1996}]
  \label{lem:ssdis}
  Define $\Gamma,\cX,\cF$ as in Definition~\ref{def:sscode}. Let $d$ (resp.~$d^\perp$) denote the minimum over all $v\in V$ of the distance of $\ker(h_v)$ (resp.~of $\im(h_v^\top)=\ker(h_v)^\perp$) Then
  \begin{align}
    \label{eq:ssdis1} d_1(\cF_{\cX}) = d_1(\cF_{\cX};X) &\geq (d-\lambda_2(\Gamma))\cdot\frac{d}{2\Delta}\cdot|V| \\
    \label{eq:ssdis0} d^0(\cF_{\cX}) \geq d^0(\cF_{\cX};X) &\geq (d^\perp-\lambda_2(\Gamma))\cdot\frac{1}{\Delta}\cdot|V|.
  \end{align}
\end{lemma}
\begin{proof}
  First observe that $d_1(\cF_{\cX})=d_1(\cF_{\cX};X)$ because $|x|=|x|_X$ for $x\in X(1)=E$, and $d^0(\cF_{\cX})\geq d^0(\cF_{\cX};X)$ beacuse $|x|\geq|x|_X$ for $x\in X(0)=V$.

  We first prove~(\ref{eq:ssdis1}). Consider an arbitrary nonzero $f\in Z_1(\cF_{\cX})\subseteq\bF_q^E$. Let $S\subseteq V$ denote the set of vertices incident to an edge in the support of $f$. Then
  \begin{equation}
    \label{eq:applyemdis}
    d|S| \leq |\dir{E}(S,S)| \leq |S|\cdot\left(\lambda_2(\Gamma)+\Delta\cdot\frac{|S|}{|V|}\right),
  \end{equation}
  where the first inequality above holds because every $v\in S$ must have $(f_e)_{e\in E(v)}\in\ker h_v$ of weight $\geq d$, and therefore $|\dir{E}(v,S)|\geq d$; the second inequality in~(\ref{eq:applyemdis}) above holds by Lemma~\ref{lem:expmix}. By definition we have $|S|\leq 2|f|/d$; applying this fact in~(\ref{eq:applyemdis}) and rearranging terms then gives that $|f|$ is bounded below by the RHS of~(\ref{eq:ssdis1}), as desired.

  We now prove~(\ref{eq:ssdis0}). Consider an arbitrary nonzero $f\in Z^0(\cF_{\cX})\subseteq(\bF_q^m)^V$. Now let $S=\supp(f)\subseteq V$. Then~(\ref{eq:applyemdis}) again holds, as now the first inequality in~(\ref{eq:applyemdis}) holds because every $v\in S$ must have $h_v^\top f_v\in\im h_v^\top$ of weight $\geq d^\perp$, and all $\geq d^\perp$ edges in the support of $h_v^\top f_v$ must have their other vertex also in $S$ in order for $\delta_0^{\cF_{\cX}}f=0$ to be satisfied, which implies $|\dir{E}(v,S)|\geq d^\perp$. Now rearranging terms in~(\ref{eq:applyemdis}) gives that $|f|_X=|S|$ is bounded below by the RHS of~(\ref{eq:ssdis0}), as desired.
\end{proof}

We will construct higher-dimensional chain complexes as tensor or balanced products of lower-dimensional chain complexes. We now define balanced products, of which tensor products are a special case. Balanced products are taken with respect to a group action on the complexes' chain spaces. In this paper, we restrict attention to abelian groups that act freely on basis elements. In this setting, balanced products \cite{breuckmann_balanced_2021} are also called ``lifted products'' \cite{panteleev_quantum_2022,panteleev_asymptotically_2022}. Note that tensor products of chain complexes are also sometimes called ``homological products'' or ``hypergraph products.''

\begin{definition}
  Let $\cC$ be a chain complex over $\bF_q$. We say that a group homomorphism $\sigma:G\rightarrow \cS_C$ from $G$ to the symmetric group on the set of basis elements $C$ of $\cC$ is a \textbf{free action of $G$ on $\cC$} if the following conditions are satisfied:
  \begin{enumerate}
  \item (Preserves dimension) For every $g\in G$ and $i\in\bZ$, we have $\sigma(g)C(i)=C(i)$. Hence we can extend $\sigma(g)$ to acting on $\cC_i=\bF_q^{C(i)}$ by permuting the basis elements.
  \item (Free) For every $g\in G\setminus\{0\}$ and $c\in C$, we have $\sigma(g)c\neq c$.
  \item (Respects boundary maps) For every $g\in G$, we have $\sigma(g)\partial^{\cC}=\partial^{\cC}\sigma(g)$.
  \end{enumerate}
\end{definition}

\begin{definition}
  For an abelian group $G$ and a ring $R$, we let $R[G]$ denote the \textbf{group algebra}, whose elements are formal combinations $\sum_{g\in G}r_gg$ with each $r_g\in R$, and where multiplication is given by $(\sum_{g\in G}r_gg)\cdot(\sum_{g\in G}r_g'g)=\sum_{g\in G}\sum_{g'\in G}r_{g'}r_{g-g'}g$.\footnote{An unfortunate consequence of our notation here is the in the group algebra $R[G]$, we have $(1g_1)\cdot(1g_2)=1(g_1+g_2)$; that is, multiplication in the group algebra performs the group operation, which we denote by addition as we work with abelian groups. We will ultimately set $G$ to be the additive group $\bF_q^t$, and we will separately use the scalar multiplicative structure in this vector space, so the notational clash here is difficult to avoid. For this reason, we will almost never explicitly write expressions involving addition and multiplication in the group algebra outside of this section.}
  
  Suppose we are given actions $\sigma^A:G\rightarrow \cS_A$ and $\sigma^B:G\rightarrow \cS_B$ of an abelian group $G$ on sets $A,B$, which then naturally endow $R^A$ and $R^B$ the structure of $R[G]$-modules. We let $R^A\otimes_G R^B$ denote the tensor product over $R[G]$ of $R^A$ and $R^B$. Equivalently, $R^A\otimes_G R^B = R^{A\times_G B}$, where
  \begin{equation}
    \label{eq:bpset}
    A\times_GB = A\times B/\sim \hspace{1em} \text{ with } \hspace{1em} (a,b)\sim(\sigma^A(g)a,\sigma^A(-g)b)\;\forall a\in A,\; b\in B,\; g\in G.
  \end{equation}
  denotes the \textbf{balanced product of sets $A,B$}. We let $a\times_Gb\in A\times_GB$ denote the equivalence class of $(a,b)$ in~(\ref{eq:bpset}).
\end{definition}

Note that the choice of $\sigma^A,\sigma^B$ is implicit in the notation $\otimes_B$ and $\times_B$, and will always be made clear from context.

\begin{definition}
  \label{def:product}
  Let $\cA^*$ and $\cB^*$ be cochain complexes of dimension $r^{\cA}$ and $r^{\cB}$ respectively over $\bF_q$, each of which respects a free action of an abelian group $G$. The \textbf{balanced product of cochain complexes $\cA,\cB$} is the cochain complex $\cC^*=\cA^*\otimes_G\cB^*$ of dimension $r^{\cC}=r^{\cA}+r^{\cB}$ given by
  \begin{equation*}
    \cC^i = \bigoplus_{j\in\bZ}\cA^j\otimes_G\cB^{i-j}
  \end{equation*}
  and
  \begin{equation*}
    \delta^{\cC}_i = \bigoplus_{j\in\bZ}(\delta^{\cA}_j\otimes_G I+(-1)^jI\otimes_G\delta^{\cB}_{i-j}).
  \end{equation*}
  If $G=\{0\}$, the balanced product is called a \textbf{tensor product}, and is denoted $\cA^*\otimes\cB^*$.
\end{definition}

Note that when $G=\{0\}$, the group action is trivial and $\otimes_G=\otimes$ is simply the ordinary tensor product over the ring at hand (which in Definition~\ref{def:product} is $R=\bF_q$). Hence tensor products of chain complexes are well defined without a group action.

Tensor products obey the following well-known formula.

\begin{lemma}[K\"{u}nneth formula]
  \label{lem:kunneth}
  Let $\cA^*$ and $\cB^*$ be cochain complexes over $\bF_q$, and let $\cC^*=\cA^*\otimes\cB^*$. Then for $i\in\bZ$,
  \begin{equation*}
    H^i(\cC) \cong \bigoplus_{j\in\bZ}H^j(\cA)\otimes H^{i-j}(\cB).
  \end{equation*}
  Furthermore, the isomorphism above is given by
  \begin{equation*}
    a\otimes b+B^i(\cC) \mapsfrom (a+B^j(\cA))\otimes(b+B^{i-j}(\cB))
  \end{equation*}
  for $a\in Z^j(\cA)$, $b\in Z^{i-j}(\cB)$.
\end{lemma}

The K\"{u}nneth formula shows how to bound the cohomology dimension of a product of cochain complexes. In a similar vein, the following known result shows how to bound the cosystolic distance of a product of cochain complexes. \cite{tillich_quantum_2014} proved the result below for the product of two 1-dimensional complexes. For completeness, we provide a proof of the general case, which is based on a proof given in \cite{bravyi_homological_2014} of a similar result.

\begin{lemma}[Similar to \cite{tillich_quantum_2014,bravyi_homological_2014}]
  \label{lem:proddis}
  Let $\cA^*$ and $\cB^*$ be cochain complexes over $\bF_q$, and let $\cC^*=\cA^*\otimes\cB^*$. Then for $i\in\bZ$,
  \begin{equation*}
    d^i(\cC) \geq \min_{j\in\bZ}\max\{d^j(\cA),d^{i-j}(\cB)\}.
  \end{equation*}
\end{lemma}
\begin{proof}
  Assume for a contradiction that there exists some $c\in Z^i(\cC)\setminus B^i(\cC)$ such that $|c|<\min_{j\in\bZ}\max\{d^j(\cA),d^{i-j}(\cB)\}$. Recalling that $\cC^i=\bigoplus_{j\in\bZ}\cA^j\otimes\cB^{i-j}$, we write $c=(c^j)_{j\in\bZ}$, where each $c^j\in\cA^j\otimes\cB^{i-j}$. By Lemma~\ref{lem:kunneth}, $H_i(\cC)$ is spanned by elements of the form $a\otimes b+B_i(\cC)$ for pairs $(a\in Z_j(\cA)\setminus B_j(\cA),\;b\in Z_{i-j}(\cB)\setminus B_{i-j}(\cB))$ for $j\in\bZ$. Because the natural bilinear form $\langle\cdot,\cdot\rangle:H_i(\cC)\times H^i(\cC)\rightarrow\bF_q$ given by $\langle c'+B_i(\cC),c+B^i(\cC)\rangle=c'\cdot c$ is nondegenerate, it follows that there exists some $j\in\bZ$ and some $(a\in Z_j(\cA)\setminus B_j(\cA),\;b\in Z_{i-j}(\cB)\setminus B_{i-j}(\cB))$ such that $(a\otimes b)\cdot c^j=(a^\top\otimes b^\top)c^j\neq 0$.

  By assumption, either $|c|<d^j(\cA)$ or $|c|<d^{i-j}(\cB)$; assume the former, as the proof for the latter is analogous. Then $a\cdot((I\otimes b^\top)c^j)=(a^\top\otimes b^\top)c^j\neq 0$, where $(I\otimes b^\top)c^j\in\cA^j$ denotes the vector whose entries are the dot products of the rows of the matrix $c^j$ with the vector $b$. In fact, because $c\in Z^i(\cC)$, we have $(\delta_j^{\cA}\otimes I)c^j+(-1)^{j+1}(I\otimes\delta_{i-j-1}^{\cB})c^{j+1}=0$, so $\delta^{\cA}(I\otimes b^\top)c^j=(-1)^j(I\otimes b^\top\delta_{i-j-1}^{\cB}c^{j+1}=0$, as $b^\top\delta_{i-j-1}^{\cB}=(\partial_{i-j}^{\cB}b)^\top=0$ beacuse $b\in Z_{i-j}(\cB)$ by assumption. Thus we have shown that $(I\otimes b^\top)c^j\in Z^j(\cA)$. Then because $a\cdot(I\otimes b^\top)c^j\neq 0$ as shown above, and the natural bilinear form $\langle\cdot,\cdot\rangle:H_j(\cA)\times H^j(\cB)\rightarrow\bF_q$ is well defined, we must have that $(I\otimes b^\top)c^j+B^j(\cA)$ is cohomologically nontrivial, meaning that $(I\otimes b^\top)c^j\in Z^j(\cA)\setminus B^j(\cA)$. Thus by definition $|(I\otimes b^\top)c^j|\geq d^j(\cA)$, so at least $d^j(\cA)$ rows of the matrix $c^j$ are nonzero, which implies that $|c|\geq|c^j|\geq d^j(\cA)$, contradicting the assumption above that $|c|<d^j(\cA)$. Thus we indeed must have $|c|\geq\min_{j\in\bZ}\max\{d^j(\cA),d^{i-j}(\cB)\}$, as desired.
\end{proof}

By dualizing the complexes involved, we see that Lemma~\ref{lem:kunneth} and Lemma~\ref{lem:proddis} apply to chain complexes analogously to cochain complexes.

In this paper, we take tensor and balanced products of the chain complexes from Definition~\ref{def:sscode} arising from placing systems of local coefficients $\cF$ on incidence complexes $\cI$ of graphs. In this case, if the graph respects a group action, and the local coefficient spaces and maps are the same within each orbit of the group action, then the associated chain complex $\cF_{\cX}$ also respects the group action. Hence we will be able to take balanced products of such chain complexes. The following example provides more details regarding such products.

\begin{example}
  \label{example:graphproduct}
  Fix some $r\in\bN$ and some abelian group $G$. For $h\in[r]$, let $\Gamma^{(h)}=(V^{(h)}=V^{(h)}_0\sqcup V^{(h)}_1,E^{(h)},\ver^{(h)})$ be an undirected $\Delta$-regular bipartite graph with associated (1-dimensional) incidence complex $\cX^{(h)}$ that respects a free action $\sigma^{(h)}$ of $G$. Let $\cF^{(h)}$ by a system of local coefficients on $\cX^{(h)}$ as given by Definition~\ref{def:sscode} such that the local matrices $h_v$ for $v\in V^{(h)}$ satisfy\footnote{There is a slight clash of notation here: $h_v$ refers to local parity-check matrices as described in Definition~\ref{def:sscode}, while $h$ is an index in $[r]$.} $h_v^{(h)}=h_{\sigma^{(h)}(g)v}^{(h)}$ for every $g\in G$. Recall here that $h_v^{(h)}\in\bF_q^{m\times E^{(h)}(v)}$ and $h_{\sigma^{(h)}(g)v}^{(h)}\in\bF_q^{m\times E^{(h)}(\sigma^{(h)}(g)v)}$ for some $m\in[\Delta]$, where we use the isomorphism $E^{(h)}(v)\cong E^{(h)}(\sigma^{(h)}(g)v)$ given by $\sigma^{(h)}(g)$.

  Then the free action $\sigma^{(h)}$ of $G$ naturally extends to act on $\cF^{(h)}_{\cX^{(h)}}$. Therefore we have a well defined balanced product $\cF^{(1)}_{\cX^{(1)}}\otimes_G\cdots\otimes_G\cF^{(r)}_{\cX^{(r)}}$. This product is an $r$-dimensional chain complex, which can be described as the associated chain complex $\cG_{\cY}$ to a product system of local coefficients $\cG$ on a product incidence complex $\cY$. Formally, the product incidence complex is given by $\cY=\cX^{(1)}\otimes_G\cdots\otimes_G\cX^{(r)}$, and for $y=x_1\times_G\cdots\times_Gx_r\in Y$ with each $x_h\in X^{(h)}$, then we have local coefficient space $\cG_y=\cF_{x_1}\otimes\cdots\otimes\cF_{x_r}$. Furthermore, given some $h\in[r]$ and some $x_h'\in X^{(h)}$ with $x_h'\triangleright x_h$ (so that $x_h'$ is an edge in $\Gamma^{(h)}$ that contains vertex $x_h$), then letting $y'=x_1\times_G\cdots\times_G x_h'\times_G\cdots\times_G x_r$, we have $y'\triangleright y$, and we let the local coefficient map $\cF_{y'\leftarrow y}=I\otimes\cdots\otimes\cF_{x_h'\leftarrow x_h}\otimes\cdots\otimes I$ simply apply $\cF_{x_h'\leftarrow x_h}=(h_{x_h}^{(h)})^\top$ in the $h$th direction of the local $r$-tensor space $\cF_y$.

  The $r$-dimensional product incidence complex $\cY$ described above with associated graded poset $Y=\bigsqcup_{i=0}^rY(i)$ can be viewed as a cubical complex, in which elements of $Y(i)$ are $i$-dimensional cubes. Specifically, every $y=x_1\times_G\cdots\times_Gx_r\in Y(i)$ must have $x_h\in X^{(h)}(1)=E^{(h)}$ for exactly $i$ values $h\in[r]$, and we have $x_h\in X^{(h)}(0)=V^{(h)}$ for the other $r-i$ values $h\in[r]$. Therefore we can define the \textbf{type $T=T(y)\in\{0,1,*\}^r$} of $y$ to be the length-$r$ tuple $T$ given by
  \begin{equation*}
    T_h = \begin{cases}
      0,&x_h\in V^{(h)}_0 \\
      1,&x_h\in V^{(h)}_1 \\
      *,&x_h\in E^{(h)}.
    \end{cases}
  \end{equation*}
  As $T(y)\in Y(i)$ has exactly $i$ $*$'s in its type, $T(y)$ can be viewed as an $i$-dimensional face of the $r$-dimensional boolean hypercube. Hence the elements $y\in Y(i)$ can be interpreted as $i$-dimensional faces within the entire cubical incidence complex $\cY$. This perspective is indeed used in \cite{dinur_expansion_2024} to obtain quantum locally testable codes from balanced product complexes of the form $\cG_{\cY}$.

  To provide slightly more detail, define a ``local'' 1-dimensional incidence complex $\cX^{\loc}$, which has $X^{\loc}(0)=\{0,1\}$, $X^{\loc}(1)=\{*\}$, and has $\delta^{\loc}\in\bZ^{\{*\}\times\{0,1\}}$ defined by $\delta^{\loc}_{*,b}=(-1)^b$ for $b\in\{0,1\}$. Let $\cY^{\loc}=(\cX^{\loc})^{\otimes r}$. Then $\cY^{\loc}$ is the incidence complex of an $r$-dimensional boolean hypercube. The elements of the associated graded poset $Y^{\loc}$ are precisely the type vectors in $\{0,1,*\}^r$. Thus $\cY^{\loc}$ can be viewed as describing the local structure of $\cY$, or alternatively, as a projection of $\cY$ obtained by collapsing all global information except for types.
\end{example}

\subsection{Transversal $C^{r-1}Z$ Gates}
In this section, we present the necessary definitions and basic results for transversal $C^{r-1}Z$ gates on quantum codes. It will be convenient to present these definitions with the chain complex language, which is equivalent to the ordinary CSS formalism by Definition~\ref{def:cctocode} and Definition~\ref{def:codetocc}. Note that despite its name, the $C^{r-1}Z$ gate is symmetric with respect to the $r$ qudits that it acts on.

\begin{definition}
  \label{def:cobinv}
  Given $r\in\bZ_{\geq 2}$, cochain complexes $({\cC^{(h)}}^*)_{h\in[r]}$ over $\bF_q$, an integer $i$, and an $r$-tuple of $i$-cohomology subspaces $H'=({H^{(h)}}'\subseteq H^i(\cC^{(h)}))_{h\in[r]}$, we say a multilinear form $\zeta:(\cC^{(1)})^i\times\cdots\times(\cC^{(r)})^i\rightarrow\bF_q$ is \textbf{coboundary-invariant on $H'$} if it holds for every $(z^{(h)}+B^i(\cC^{(h)})\in {H^{(h)}}')_{h\in[r]}$ that
  \begin{equation*}
    \zeta({z^{(h)}}')_{h\in[r]}
  \end{equation*}
  has the same same value for every $({z^{(h)}}'\in z^{(h)}+B^i(\cC^{(h)}))_{h\in[r]}$. If $\zeta$ is coboundary-invariant on $H'$, it naturally induces a multilinear form on the cohomology subspaces $H'$, which we denote $\zeta_{H'}:{H^{(1)}}'\times\cdots\times {H^{(r)}}'\rightarrow\bF_q$. Formally, for every $(z^{(h)}+B^i(\cC^{(h)})\in H^i(\cC^{(h)}))_{h\in[r]}$, we let
  \begin{equation*}
    \zeta_{H'}(z^{(h)}+B^i(\cC^{(h)}))_{h\in[r]}=\zeta(z^{(h)})_{h\in[r]}.
  \end{equation*}
  Here the coboundary-invariance of $\zeta$ ensures that $\zeta_{H'}$ is well defined, regardless of the choice of cohomology class representatives $z^{(h)}$.

  We also define the \textbf{locality $w^\zeta$} of $\zeta$ to be the maximum number of nonzero entries in any axis-aligned $(r-1)$-dimensional hyperplane of the representation of $\zeta$ as an $r$-tensor with the bases $C^{(h)}(i)$ for $h\in[r]$. More formally, we can define locality using a \textbf{connectivity graph $\Gamma^\zeta=(V_0^\zeta\sqcup V_1^\zeta,E^\zeta,\ver^\zeta)$} for $\zeta$ defined as follows. Let $\Gamma^\zeta$ be the simple bipartite graph with left vertices $V^\zeta_0=C^{(1)}(i)\sqcup\cdots\sqcup C^{(r)}(i)$, right vertices $V^\zeta_1=C^{(1)}(i)\times\cdots\times C^{(r)}(i)$, and with an edge connecting left vertex $a^{(h)}\in C^{(h)}(i)$ with right vertex $(b^{(1)},\dots,b^{(r)})\in V^\zeta_1$ if $a^{(h)}=b^{(h)}$ and $\zeta(b^{(1)},\dots,b^{(r)})\neq 0$. Then the locality $w^\zeta$ of $\zeta$ equals the maximum degree in $\Gamma^\zeta$ of any left vertex $a^{(h)}\in V_0^\zeta$.
\end{definition}

\begin{definition}
  \label{def:subrank}
  For $\bF_q$-vector spaces $V^{(1)},\dots,V^{(r)}$, the \textbf{subrank} of a multilinear form $\zeta:V^{(1)}\times\cdots\times V^{(r)}\rightarrow\bF_q$ is the maximum $s\in\bZ_{\geq 0}$ such that there exist vectors $(v_j^{(h)})_{j\in[s]}^{h\in[r]}$ satisfying
  \begin{equation*}
    \zeta(v_{j_1}^{(1)},\dots,v_{j_r}^{(r)})=\1_{j_1=\cdots=j_r} \hspace{1em} \text{ for every } \hspace{1em} (j_1,\dots,j_r)\in[s]^r.
  \end{equation*}
\end{definition}

Below, we relate the above definitions to transversal $C^{r-1}Z$ gates in the traditional sense. Specifically, we show that if cochain complexes $({\cC^{(h)}}^*)_{h\in[r]}$ admit a multilinear form $\zeta:(\cC^{(1)})^i\times\cdots\times(\cC^{(r)})^i\rightarrow\bF_q$ that is coboundary-invariant on some $H'$, then there exist $\subrank(\zeta_{H'})$ $r$-tuples of logical qudits in the quantum codes associated to level $i$ of $({\cC^{(h)}}^*)_{h\in[r]}$ (where each tuple contains one logical qudit from each code) such that a logical $C^{r-1}Z$ gate on all tuples can be induced by a physical circuit consisting of $C^{r-1}Z$ gates, with each physical qudit involved in at most $w^\zeta$ of these physical gates.

We first formally define the $C^{r-1}Z$ gate.

\begin{definition}
  For a finite field $\bF_q$ of characteristic $p$ and an element $a\in\bF_q^*$, the gate
  \begin{equation*}
    C^{r-1}Z_q^a:(\bC^q)^{\otimes r}\rightarrow(\bC^q)^{\otimes r}
  \end{equation*}
  is a quantum gate (i.e.~a unitary operator) acting on $r$ qudits of local dimension $q$, such that for $(z_1,\dots,z_r)\in\bF_q^r$ then
  \begin{equation*}
    C^{r-1}Z_q^a\ket{z_1}\cdots\ket{z_r} = e^{2\pi i\tr_{\bF_q/\bF_p}(a\cdot z_1\cdots z_r)/p}\ket{z_1}\cdots\ket{z_r}.
  \end{equation*}
  When the field is clear from context, we often omit the $q$ subscript. When $a=1$, we write $C^{r-1}Z^1=C^{r-1}Z$.
\end{definition}

\begin{remark}
  \label{remark:removecoeff}
  For $a\in\bF_q^*$, the gate $C^{r-1}Z_q^a$ can be performed using a $C^{r-1}Z_q^1$ gate along with two Clifford gates. Specifically, for an input state $\ket{z_1}\cdots\ket{z_r}$, if we apply the Clifford gate $\ket{z}\mapsto\ket{az}$ to the first qudit $\ket{z_1}$, then apply $C^{r-1}Z_q^1$ to all $r$ qudits, and then apply the Clifford $\ket{z}\mapsto\ket{a^{-1}z}$ to the first qudit, the resulting state is precisely $C^{r-1}Z_q^a\ket{z_1}\cdots\ket{z_r}$.
\end{remark}

The lemma below formally relates our chain complex language to transversal $C^{r-1}Z$ gates, as described above. Below, for a set $S\subseteq\bF_q^n$, we let $\ket{S}=(1/\sqrt{|S|})\sum_{y\in S}\ket{y}\in(\bC^q)^{\otimes n}$ denote the uniform superposition over elements of $S$. Also recall that for a cochain complex $\cC^*$ over $\bF_q$, the quantum CSS code associated to level $i$ of $\cC^*$ corresponds to the subspace of the Hilbert space $(\bC^q)^{\otimes\dim(\cC^i)}$ given by $\spn\{\ket{z+B^i(\cC)}:z\in Z^i(\cC)\}$.

\begin{lemma}
  \label{lem:cctotrans}
  For some $r\in\bZ_{\geq 2}$, let $({\cC^{(h)}}^*)_{h\in[r]}$ be cochain complexes over a finite field $\bF_q$ of characteristic $p$, and for some integer $i$ let $\zeta:(\cC^{(1)})^i\times\cdots\times(\cC^{(r)})^i\rightarrow\bF_q$ be a multilinear form that is coboundary-invariant on some $H'=({H^{(h)}}'\subseteq H^i(\cC^{(h)}))_{h\in[r]}$. Let $s=\subrank(\zeta_{H'})$, and let $(z^{(h)}_j+B^i(\cC^{(h)})\in{H^{(h)}}')_{j\in[s]}^{h\in[r]}$ be cohomology classes satisfying
  \begin{equation*}
    \zeta_{H'}(z^{(1)}_{j_1}+B^i(\cC^{(1)}),\dots,z^{(r)}_{j_r}+B^i(\cC^{(r)})) = \1_{j_1=\cdots=j_r} \hspace{1em} \text{ for every } \hspace{1em} (j_1,\dots,j_r)\in[s]^r,
  \end{equation*}
  as given by Definition~\ref{def:subrank}. For $h\in[r]$, let $N^{(h)}=|C^{(h)}(i)|=\dim(\cC^{(h)})^i$ be the length of the quantum code at level $i$ of $\cC^{(h)}$, and define an encoding isometry
  \begin{align*}
    &\Enc^{(h)}:(\bC^q)^{\otimes s}\rightarrow(\bC^q)^{\otimes N^{(h)}} \\
    &\Enc^{(h)}\ket{y_1,\dots,y_s} = \ket{\sum_{j\in[s]}y_jz^{(h)}_j+B^i(\cC^{(h)})},
  \end{align*}
  where we index the $N^{(h)}$ physical qudits in $(\bC^q)^{\otimes N^{(h)}}$ by the basis elements in $C^{(h)}(i)$. Also define a unitary
  \begin{equation*}
    C^{r-1}Z^\zeta:\bigotimes_{h\in[r]}(\bC^q)^{\otimes N^{(h)}}\rightarrow\bigotimes_{h\in[r]}(\bC^q)^{\otimes N^{(h)}}
  \end{equation*}
  that applies the gate $C^{r-1}Z^a$ to every $r$-tuple of physical qudits $(u^{(1)},\dots,u^{(r)})\in C^{(1)}(i)\times\cdots\times C^{(r)}(i)$ for which $\zeta(u^{(1)},\dots,u^{(r)})=a$. Then it holds for every $(y^{(1)},\dots,y^{(r)})\in(\bF_q^s)^r$ that
  \begin{align*}
    \hspace{1em}&\hspace{-1em} C^{r-1}Z^\zeta(\Enc^{(1)}\otimes\cdots\otimes\Enc^{(r)})\left(\ket{y^{(1)}}\cdots\ket{y^{(s)}}\right) \\
                &= (\Enc^{(1)}\otimes\cdots\otimes\Enc^{(r)})C^{r-1}Z^{\otimes s}\left(\ket{y^{(1)}}\cdots\ket{y^{(s)}}\right),
  \end{align*}
  where $C^{r-1}Z^{\otimes s}\left(\ket{y^{(1)}}\cdots\ket{y^{(s)}}\right)$ above applies a $C^{r-1}Z$ gate to $\ket{y^{(1)}_j}\cdots\ket{y^{(s)}_j}$ for every $j\in[s]$.
\end{lemma}
\begin{proof}
  In this proof, to simplify notation we work with unnormalized quantum states, so that we write $\ket{x}=\beta\ket{x}$ for every positive real number $\beta$. Then by definition,
  \begin{align*}
    \hspace{1em}&\hspace{-1em} C^{r-1}Z^\zeta(\Enc^{(1)}\otimes\cdots\otimes\Enc^{(r)})\left(\ket{y^{(1)}}\cdots\ket{y^{(s)}}\right) \\
                &= C^{r-1}Z^\zeta\left(\ket{\sum_{j\in[s]}y^{(1)}_jz^{(1)}_j+B^i(\cC^{(1)})}\cdots\ket{\sum_{j\in[s]}y^{(r)}_jz^{(r)}_j+B^i(\cC^{(r)})}\right) \\
                &= \sum_{(b^{(1)},\dots,b^{(r)})\in B^i(\cC^{(1)})\times\cdots\times B^i(\cC^{(r)})} C^{r-1}Z^\zeta\ket{\sum_{j\in[s]}y_j^{(1)}z_j^{(1)}+b^{(1)},\dots,\sum_{j\in[s]}y_j^{(r)}z_j^{(r)}+b^{(r)}} \\
                &= \sum_{(b^{(1)},\dots,b^{(r)})\in B^i(\cC^{(1)})\times\cdots\times B^i(\cC^{(r)})} e^{(2\pi i/p)\tr_{\bF_q/\bF_p}(\zeta(\sum_{j\in[s]}y_j^{(1)}z_j^{(1)}+b^{(1)},\dots,\sum_{j\in[s]}y_j^{(r)}z_j^{(r)}+b^{(r)}))} \\
                &\hspace{15em} \cdot\ket{\sum_{j\in[s]}y_j^{(1)}z_j^{(1)}+b^{(1)},\dots,\sum_{j\in[s]}y_j^{(r)}z_j^{(r)}+b^{(r)}} \\
                &= \sum_{(b^{(1)},\dots,b^{(r)})\in B^i(\cC^{(1)})\times\cdots\times B^i(\cC^{(r)})} e^{(2\pi i/p)\tr_{\bF_q/\bF_p}(\sum_{j\in[s]}y^{(1)}_j\cdots y^{(r)}_j)} \\
                &\hspace{15em} \cdot\ket{\sum_{j\in[s]}y_j^{(1)}z_j^{(1)}+b^{(1)},\dots,\sum_{j\in[s]}y_j^{(r)}z_j^{(r)}+b^{(r)}} \\
                &= (\Enc^{(1)}\otimes\cdots\otimes\Enc^{(r)}) e^{(2\pi i/p)\tr_{\bF_q/\bF_p}(\sum_{j\in[s]}y^{(1)}_j\cdots y^{(r)}_j)} \left(\ket{y^{(1)}}\cdots\ket{y^{(s)}}\right) \\
                &= (\Enc^{(1)}\otimes\cdots\otimes\Enc^{(r)}) C^{r-1}Z^{\otimes s} \left(\ket{y^{(1)}}\cdots\ket{y^{(s)}}\right),
  \end{align*}
  where the third equality above holds by the definition of $C^{r-1}Z^\zeta$, and the fourth equality above holds by the coboundary-invariance of $\zeta$ on $H'$.
\end{proof}

In light of Lemma~\ref{lem:cctotrans}, we make the following definition regarding transversal gates in the ordinary colloquial sense.

\begin{definition}
  \label{def:transversal}
  For some $r\in\bZ_{\geq 2}$, let $({\cC^{(h)}}^*)_{h\in[r]}$ be cochain complexes over $\bF_q$, and for some integer $i$ let $\zeta:(\cC^{(1)})^i\times\cdots\times(\cC^{(r)})^i\rightarrow\bF_q$ be a multilinear form that is coboundary-invariant on some $H'=({H^{(h)}}'\subseteq H^i(\cC^{(h)}))_{h\in[r]}$. If $\zeta$ has locality $w^\zeta=1$, then letting $s=\subrank(\zeta_{H'})$, we say that the quantum codes at level $i$ of $({\cC^{(h)}}^*)_{h\in[r]}$ support a \textbf{transversal $C^{r-1}Z$ gate inducing $s$ logical $C^{r-1}Z$ gates}.
\end{definition}

The condition that $\zeta$ has locality $1$ ensures that by permuting the physical code qudits, we can make the physical $C^{r-1}Z^\zeta$ circuit defined in Lemma~\ref{lem:cctotrans} simply perform the following: for each $j\in\min\{N^{(1)},\dots,N^{(r)}\}$, for some $a_j\in\bF_q$ we apply a $C^{r-1}Z^{a_j}$ gate to the $r$-tuple of $j$th physical qudits in the $r$ codes. Indeed, the term ``transversal gates'' is often used to refer to such a depth-1 physical circuit that applies at most one gate to each physical qudit. Also recall that by Remark~\ref{remark:removecoeff}, the $C^{r-1}Z^a$ gate is equivalent to $C^{r-1}Z=C^{r-1}Z^1$ up to Cliffords.

While the condition that $\zeta$ has locality $w^\zeta=1$ may seem restrictive, in Lemma~\ref{lem:locred} below we show that there is always a basic procedure that reduces the locality down to $1$, while preserving the other code properties up to a small loss in parameters.

However, we first present the following basic lemma, which shows that the coboundary-invariance, locality, and subrank of a multilinear form, as well as quantum code length, dimension, and distance, are preserved under passing to subfields, up to a small loss in parameters. As a consequence, we can for instance obtain chain complexes over $\bF_2$ from chain complexes over $\bF_{2^m}$, while preserving the properties of interest to us. We remark that a more sophisticated such alphabet reduction technique, with a smaller loss in parameters, is given in \cite{golowich_asymptotically_2024,nguyen_good_2024}, though the basic result below suffices for our purposes.

\begin{lemma}
  \label{lem:alphred}
  For a prime power $q$ and for $m,r,i\in\bZ$ with $m,r\geq 2$, let $(\cC^{(h)})_{h\in[r]}$ be cochain complexes over $\bF_{q^m}$ with a multilinear form $\zeta:(\cC^{(1)})^i\times\cdots\times(\cC^{(r)})^i\rightarrow\bF_{q^m}$ that is coboundary-invariant on some $H'=({H^{(h)}}'\subseteq H^i(\cC^{(h)}))_{h\in[r]}$.

  Let $(\tilde{\cC}^{(h)})_{h\in[r]},\tilde{H}'$ equal the respective objects $(\cC^{(h)})_{h\in[r]},H'$ viewed over the subfield $\bF_q\subseteq\bF_{q^m}$, so that each $\tilde{\cC}^{(h)}$ is a chain complex over $\bF_q$ with basis $\tilde{C}^{(h)}=C^{(h)}\times[m]$. Fix an arbitrary $\bF_q$-linear map $\phi:\bF_{q^m}\rightarrow\bF_q$ such that $\phi|_{\bF_q}=\mathrm{id}$, and let $\tilde{\zeta}:(\tilde{\cC}^{(1)})^i\times\cdots\times(\tilde{\cC}^{(r)})^i\rightarrow\bF_q$ be the multilinear form over $\bF_q$ defined by $\tilde{\zeta}=\phi\circ\zeta$. Then $\tilde{\zeta}$ is coboundary-invariant on $\tilde{H}'$, and we have
  \begin{align}
    \label{eq:subfield}
    \begin{split}
      \dim_{\bF_q}(\tilde{\cC}^{(h)})^i &= m\cdot\dim_{\bF_{q^m}}(\cC^{(h)})^i \\
      \dim_{\bF_q}H^i(\tilde{\cC}^{(h)}) &= m\cdot\dim_{\bF_{q^m}}{H^i(\cC^{(h)})} \\
      d^i(\tilde{\cC}^{(h)}) &\geq d^i(\cC^{(h)}) \\
      d_i(\tilde{\cC}^{(h)}) &\geq d_i(\cC^{(h)}) \\
      w^{\tilde{\cC}^{(h)}} &\leq m\cdot w^{\cC^{(h)}} \\
      w^{\tilde{\zeta}} &\leq m^{r-1}\cdot w^{\zeta} \\
      \subrank(\tilde{\zeta}_{\tilde{H}'}) &\geq \subrank(\zeta_{H'}).
    \end{split}
  \end{align}
\end{lemma}
\begin{proof}
  The coboundary-invariance of $\tilde{\zeta}$ on $\tilde{H}'$ follows directly from the coboundary-invariance of $\zeta$ on $H'$, as $B^i(\cC^{(h)})=B^i(\tilde{\cC}^{(h)})$ and $H'=\tilde{H}'$. The first two equalities in~(\ref{eq:subfield}) also follow by definition. The two distance bounds in~(\ref{eq:subfield}) hold because every $\tilde{c}\in Z^i(\tilde{\cC}^{(h)})\setminus B^i(\tilde{\cC}^{(h)})$ can also be viewed as an element $c\in Z^i(\cC^{(h)})\setminus B^i(\cC^{(h)})$, where $|c|\leq|\tilde{c}|\leq m\cdot|c|$ because $\tilde{C}^{(h)}=C^{(h)}\times[m]$, so that every nonzero component of $\tilde{c}$ corresponds to a nonzero $m$-tuple of components of $c$. The bound on $w^{\tilde{\cC}^{(h)}}$ in~(\ref{eq:subfield}) also holds beacuse $\tilde{C}^{(h)}=C^{(h)}\times[m]$, as for $(x,j),(x',j')\in C^{(h)}\times[m]$, the matrix element $(\delta_i^{\tilde{\cC}^{(h)}})_{(x,j),(x',j')}\in\bF_q$ can only be nonzero if $(\delta_i^{\cC^{(h)}})_{x,x'}\in\bF_{q^m}$ is nonzero. The bound on $w^{\tilde{\zeta}}$ in~(\ref{eq:subfield}) holds by similar reasoning, as for $((x^{(1)},j^{(1)}),\dots,(x^{(r)},j^{(r)}))\in(C^{(1)}(i)\times[m])\times\cdots\times(C^{(r)}(i)\times[m])$, we can only have $\tilde{\zeta}(\ind{(x^{(1)},j^{(1)})},\dots,\ind{(x^{(r)},j^{(r)})})\neq 0$ if $\zeta(\ind{x^{(1)}},\dots,\ind{x^{(r)}})\neq 0$. The subrank bound in~(\ref{eq:subfield}) holds because $0,1\in\bF_q\subseteq\bF_{q^m}$, so if there exist $(v_j^{(h)}\in(H^{(h)})'\cong(\tilde{H}^{(h)})')_{j\in[s]}^{h\in[r]}$ with $\zeta_{H'}(v_{j_1}^{(1)},\dots,v_{j_r}^{(r)})=\1_{j_1=\cdots=j_r}$, then by definition $\tilde{\zeta}_{\tilde{H}'}(v_{j_1}^{(1)},\dots,v_{j_r}^{(r)})=\phi(\1_{j_1=\cdots=j_r})=\1_{j_1=\cdots=j_r}$.
\end{proof}

The following lemma shows that the locality of a coboundary-invariant multilinear form can always be reduced to $1$, at the cost of increasing the associated quantum code length.

\begin{lemma}
  \label{lem:locred}
  For a prime power $q$ and an integer $r\geq 2$, let $(\cC^{(h)})_{h\in[r]}$ be 2-dimensional\footnote{For chain complexes of arbitrary dimension with a coboundary-invariant multilinear form on $i$-cochains, we can always truncate the complex to levels $i-1,i,i+1$ and then relabel these three levels $0,1,2$ in order to obtain a 2-dimensional complex wtih the multilinear form $\zeta$ acting on 1-cochains. Furthermore, the quantum code associated to level $i$ of the original complex is equal to the quantum code associated to level $1$ of the truncated complex.} cochain complexes over $\bF_q$ with a multilinear form $\zeta:(\cC^{(1)})^1\times\cdots\times(\cC^{(r)})^i\rightarrow\bF_q$ that is coboundary-invariant on some $H'=({H^{(h)}}'\subseteq H^1(\cC^{(h)}))_{h\in[r]}$.

  Then there exist 2-dimensional cochain complexes $(\tilde{\cC}^{(h)})_{h\in[r]}$ over $\bF_q$ with a multilinear form $\tilde{\zeta}:(\tilde{\cC}^{(1)})^1\times\cdots\times(\tilde{\cC}^{(r)})^i\rightarrow\bF_q$ that is coboundary-invariant on some $\tilde{H}'=((\tilde{H}^{(h)})'\subseteq H^1(\tilde{\cC}^{(h)}))_{h\in[r]}$, such that the following hold:
  \begin{align}
    \label{eq:decongest}
    \begin{split}
      \dim(\tilde{\cC}^{(h)})^1 &= w^\zeta\cdot\dim(\tilde{\cC}^{(h)})^1 \\
      \dim H^1(\tilde{\cC}^{(h)}) &= \dim H^1(\cC^{(h)}) \\
      d^1(\tilde{\cC}^{(h)}) &= w^\zeta\cdot d^1(\cC^{(h)}) \\
      d_1(\tilde{\cC}^{(h)}) &\geq d_1(\cC^{(h)}) \\
      w^{\tilde{\cC}^{(h)}} &\leq \max\{w^\zeta\cdot w^{\cC^{(h)}},\; w^{\cC^{(h)}}+2\} \\
      w^{\tilde{\zeta}} &= 1 \\
      \subrank(\tilde{\zeta}_{\tilde{H}'}) &= \subrank(\zeta_{H'}).
    \end{split}
  \end{align}
\end{lemma}
\begin{proof}
  At a high level, we let each $\tilde{\cC}^{(h)}$ be the cochain complex associated to the concatenation of a classical repetition code with the quantum code associated to $\cC^{(h)}$. This concatenation with a classical repetition code ensures that each 1-cochain of $\cC^{(h)}$ has every value repeated many times, so we can obtain $\tilde{\zeta}$ by ``decongesting'' the terms of the multilinear form $\zeta$, meaning that we use different copies of a given value in different terms. The formal details are presented below.

  Let $\vec{1}\in\bF_q^{w^\zeta}$ denote the all-1s vector of length $w^\zeta$, and let $\delta^{\mathrm{rep}}\in\bF_q^{(w^\zeta-1)\times w^\zeta}$ be the matrix whose $i$th row equals $\ind{i}-\ind{i+1}\in\bF_q^{w^\zeta}$, so that $\ker(\delta^{\mathrm{rep}})=\vec{1}$. Let $\partial^{\mathrm{rep}}=(\delta^{\mathrm{rep}})^\top$, so that $\im(\partial^{\mathrm{rep}})=\vec{1}^\perp$ is the set of all vectors in $\bF_q^{w^\zeta}$ whose entries sum to $0$.

  Then for $h\in[r]$, we define $\tilde{\cC}^{(h)}$ by
  \begin{align*}
    (\tilde{\cC}^{(h)})^0 &= (\cC^{(h)})^0 \\
    (\tilde{\cC}^{(h)})^1 &= (\cC^{(h)})^1\otimes\bF_q^{w^\zeta} \\
    (\tilde{\cC}^{(h)})^2 &= (\cC^{(h)})^2 \oplus ((\cC^{(h)})^1\otimes\bF_q^{w^\zeta-1}) \\
    \delta_0^{\tilde{\cC}^{(h)}} &= \delta_0^{\cC^{(h)}}\otimes\vec{1} \\
    \delta_1^{\tilde{\cC}^{(h)}} &= (\delta_1^{\cC^{(h)}}\otimes\ind{1}^\top) \oplus (I\otimes\delta^{\mathrm{rep}}).
  \end{align*}
  Then $\tilde{\cC}^{(h)}$ is a well-defined cochain complex because
  \begin{align*}
    \delta_1^{\tilde{\cC}^{(h)}}\delta_0^{\tilde{\cC}^{(h)}}
    &= ((\delta_1^{\cC^{(h)}}\otimes\ind{1}^\top) \oplus (I\otimes\delta^{\mathrm{rep}}))(\delta_0^{\cC^{(h)}}\otimes\vec{1}) \\
    &= (\delta_1^{\cC^{(h)}}\delta_0^{\cC^{(h)}})\oplus(\delta_0^{\cC^{(h)}}\otimes\delta^{\mathrm{rep}}\vec{1}) \\
    &= 0\oplus 0 \\
    &= 0,
  \end{align*}
  where the third equality above holds because $\cC^{(h)}$ is a cochain complex and $\vec{1}\in\ker(\delta^{\mathrm{rep}})$.

  Now the expression for $\dim(\tilde{\cC}^{(h)})^1$ in~(\ref{eq:decongest}) holds by the definition of $(\tilde{\cC}^{(h)})^1$. By definition, $Z^1(\tilde{\cC}^{(h)})$ is the space of all $\dim(\cC^{(h)})^1\times w^\zeta$ matrices for which the first column lies in $Z^1(\cC^{(h)})$, and all values within each row are equal. Similarly, $B^1(\tilde{\cC}^{(h)})$ is the space of all $\dim(\cC^{(h)})^1\times w^\zeta$ matrices for which the first column lies in $B^1(\cC^{(h)})$, and all values within each row are equal. Thus the mapping
  \begin{align}
    \label{eq:copyiso}
    \begin{split}
      &\phi:(\cC^{(h)})^1\rightarrow(\tilde{\cC}^{(h)})^1 \\
      &\phi(c) = c\otimes\vec{1},
    \end{split}
  \end{align}
  which simply copies each component if its input $w^\zeta$ times, induces isomorphisms $B^1(\cC^{(h)})\cong B^1(\tilde{\cC}^{(h)})$ and $Z^1(\cC^{(h)})\cong Z^1(\tilde{\cC}^{(h)})$. It follows that $\dim H^1(\tilde{\cC}^{(h)})=\dim H^1(\cC^{(h)})$ and that $d^1(\tilde{\cC}^{(h)})=w^\zeta\cdot d^1(\cC^{(h)})$, which are precisely the equalities in~(\ref{eq:decongest}).

  Consider some $\tilde{c}\in Z_1(\tilde{\cC}^{(h)})\setminus B_1(\tilde{\cC}^{(h)})$. Then because $\partial_1^{\tilde{\cC}^{(h)}}=\partial_1^{\cC^{(h)}}\otimes\vec{1}^\top$ by definition, it follows that $(I\otimes\vec{1}^\top)\tilde{c}\in Z_1(\cC^{(h)})$. Furthermore, if $(I\otimes\vec{1}^\top)\tilde{c}=\partial_2^{\cC^{(h)}}c$ for some $c\in(\cC^{(h)})^2$, then $\tilde{c}-(\partial_2^{\cC^{(h)}}\otimes\ind{1})c$ is a $\dim(\cC^{(h)})^1\times w^\zeta$ matrix in which each row sums to $0$, which means that this matrix lies in $\im(I\otimes\partial^{\mathrm{rep}})$. Hence there exists some $b\in(\cC^{(h)})^1\otimes\bF_q^{w^\zeta-1}$ with $(I\otimes\partial^{\mathrm{rep}})b=\tilde{c}-(\partial_2^{\cC^{(h)}}\otimes\ind{1})c$, which implies that $\tilde{c}=\partial_2^{\tilde{\cC}^{(h)}}(c,b)$ by the definition of $\partial_2^{\tilde{\cC}^{(h)}}$. But this equality contradicts the assumption that $\tilde{c}\notin B_1(\tilde{\cC}^{(h)})$, so the assumption that $(I\otimes\vec{1}^\top)\tilde{c}=\partial_2^{\cC^{(h)}}c$ for some $c$ was false. Thus we have shown that $(I\otimes\vec{1}^\top)\tilde{c}\in Z_1(\cC^{(h)})\setminus B_1(\cC^{(h)})$, so $|(I\otimes\vec{1}^\top)\tilde{c}|\geq d_1(\cC^{(h)})$. Then as $(I\otimes\vec{1}^\top)\tilde{c}$ is simply the vector whose entries are the sums of the rows of $\tilde{c}$, it follows that $|\tilde{c}|\geq|(I\otimes\vec{1}^\top)c|\geq d_1(\cC^{(h)})$. Thus $d_1(\tilde{\cC}^{(h)})\geq d_1(\cC^{(h)})$, which is precisely the bound in~(\ref{eq:decongest}).

  By definition the locality of the matrix $\delta_0^{\tilde{\cC}^{(h)}}=\delta_0^{\cC^{(h)}}\otimes\vec{1}$ is at most $w^\zeta$ times the locality of the matrix $\delta_0^{\cC^{(h)}}$, which in turn is at most $w^{\cC^{(h)}}$. Thus $\delta_0^{\tilde{\cC}^{(h)}}$ has locality $\leq w^\zeta\cdot w^{\cC^{(h)}}$. Meanwhile, the locality of the matrix $\delta_1^{\cC^{(h)}}\otimes\ind{1}^\top$ equals the locality of $\delta_1^{\cC^{(h)}}$, which is at most $w^{\cC^{(h)}}$, and the locality of the matrix $I\otimes\delta^{\mathrm{rep}}$ equals the locality of $\delta^{\mathrm{rep}}$, which is $2$. Thus $\delta_1^{\tilde{\cC}^{(h)}} = (\delta_1^{\cC^{(h)}}\otimes\ind{1}^\top) \oplus (I\otimes\delta^{\mathrm{rep}})$ has locality at most $w^{\cC^{(h)}}+2$. Thus we have shown that $\cC^{(h)}$ has locality at most $\max\{w^\zeta\cdot w^{\cC^{(h)}},\; w^{\cC^{(h)}}+2\}$, as stated in~(\ref{eq:decongest}).

  We now define the multilinear form $\tilde{\zeta}$. Let $\Gamma^\zeta=(V_0^\zeta\sqcup V_1^\zeta,E^\zeta,\ver^\zeta)$ be the connectivity graph of $\zeta$ from Definition~\ref{def:cobinv}. Let $U\subseteq V_1^\zeta$ be the subset of right vertices of strictly positive degree. Then by definition
  \begin{align}
    \label{eq:zetaexp}
    \zeta(c^{(1)},\dots,c^{(r)})
    &= \sum_{u=(u^{(1)},\dots,u^{(r)})\in U}\zeta(u)\cdot c^{(1)}_{u^{(1)}}\cdots c^{(r)}_{u^{(r)}}.
  \end{align}
  Note that above we let $\zeta(u)$ denote the evaluation of $\zeta$ on the basis elements $(u^{(1)},\dots,u^{(r)})\in C^{(1)}(1)\times\cdots\times C^{(r)}(1)$, while we let $c^{(h)}_{u^{(h)}}$ denote the $u^{(h)}$-component of $c^{(h)}\in(\cC^{(h)})^1=\bF_q^{C^{(h)}(1)}$ when expressed in the basis $C^{(h)}(1)$.

  By the definition of $w^\zeta$, each left vertex $v_0\in V_0^\zeta$ has degree $\leq w^\zeta$, so we may construct a labeling $\labE:E^\zeta\rightarrow[w^\zeta]$ such that all $\leq w^\zeta$ edges incident to each $v_0\in V_0^\zeta$ have distinct labels. For a vertex $u=(u^{(1)},\dots,u^{(r)})\in U$, for every $h\in[r]$ there is a unique edge from vertex $u^{(h)}\in V_0^\zeta$ to vertex $u\in U\subseteq V_1^\zeta$, so we let $\labE(u^{(h)},u)\in[w^\zeta]$ denote the label of this edge. Then we define
  \begin{align}
    \label{eq:tzetadef}
    \tilde{\zeta}(\tilde{c}^{(1)},\dots,\tilde{c}^{(r)})
    &= \sum_{u=(u^{(1)},\dots,u^{(r)})\in U}\zeta(u)\cdot \tilde{c}^{(1)}_{(u^{(1)},\labE(u^{(1)},u))}\cdots\tilde{c}^{(r)}_{(u^{(r)},\labE(u^{(r)},u))},
  \end{align}
  where similarly as above we let $\tilde{c}^{(h)}_{(u^{(h)},\labE(u^{(h)},u))}$ denote the $(u^{(h)},\labE(u^{(h)},u))$-component of $\tilde{c}^{(h)}\in(\tilde{\cC}^{(h)})^1=\bF_q^{C^{(h)}(1)\times[w^\zeta]}$ expressed in the basis $C^{(h)}(1)\times[w^\zeta]$.

  By definition, for every $h\in[r]$ and every $(u^{(h)},j^{(h)})\in C^{(h)}(1)\times[w^\zeta]$, there is at monst one vertex $u=(u^{(1)},\dots,u^{(r)})\in U$ incident to $u^{(h)}$ with edge label $\labE(u^{(h)},u)=j^{(h)}$. Therefore the $(u^{(h)},j^{(h)})$-component $\tilde{c}^{(h)}_{(u^{(h)},j^{(h)})\in C^{(h)}(1)\times[w^\zeta]}$ of $\tilde{c}^{(h)}$ appears in at most one term of the sum in~(\ref{eq:tzetadef}). It follows that $\tilde{\zeta}$ has locality $w^{\tilde{\zeta}}=1$, as desired in~(\ref{eq:decongest}).

  Recalling that the map $\phi$ defined in~(\ref{eq:copyiso}) induces isomorphisms $B^1(\cC^{(h)})\cong B^1(\tilde{\cC}^{(h)})$ and $Z^1(\cC^{(h)})\cong Z^1(\tilde{\cC}^{(h)})$, we may define $\tilde{H'}$ by letting $(\tilde{H}^{(h)})'=\phi((H^{(h)})')$ for each $h\in[r]$. By definition, for every $(\tilde{c}^{(1)},\dots,\tilde{c}^{(r)})\in Z^{(1)}(\tilde{\cC}^{(1)})\times\cdots\times Z^{(r)}(\tilde{\cC}^{(r)})$, if we let $c^{(h)}=\phi|_{Z^1(\cC^{(h)})}^{-1}(\tilde{c}^{(h)})$ for $h\in[r]$, then we have
  \begin{align*}
    \tilde{\zeta}(\tilde{c}^{(1)},\dots,\tilde{c}^{(r)})
    &= \sum_{u=(u^{(1)},\dots,u^{(r)})\in U}\zeta(u)\cdot \tilde{c}^{(1)}_{(u^{(1)},\labE(u^{(1)},u))}\cdots\tilde{c}^{(r)}_{(u^{(r)},\labE(u^{(r)},u))} \\
    &= \sum_{u=(u^{(1)},\dots,u^{(r)})\in U}\zeta(u)\cdot c^{(1)}_{u^{(1)}}\cdots c^{(r)}_{u^{(r)}} \\
    &= \zeta(c^{(1)},\dots,c^{(r)}),
  \end{align*}
  where the second equality above follows by the definition of $\phi$, and the third equality follows from~(\ref{eq:zetaexp}). Thus under the isomorphism $\phi|_{Z^1(\cC^{(h)})}$, the multilinear forms $\tilde{\zeta}$ and $\zeta$ act identically on 1-cocycles (and as a consequence, also on 1-coboundaries) in their respective cochain complexes. Thus because $\zeta$ is coboundary-invariant on $H'$, it follows that $\tilde{\zeta}$ is coboundary-invariant on $\tilde{H}'$, and that $\subrank(\tilde{\zeta}_{\tilde{H}'})=\subrank(\zeta_{H'})$, as desired.
\end{proof}

\section{Expander Construction}
\label{sec:expander}
In this section, we present a family of spectral expander graphs that embed nicely into a high-dimensional vector space. Our expanders are lifts of a small expander by the abelian (additive) group $\bF_q^t$, where we will typically choose $t\approx q^\tau$ for a small $\tau>0$.

\subsection{Construction}
\label{sec:expconstruct}
Here we present the details of the expander construction. The construction takes as input a prime power $q$, positive integers $t,\Delta\in\bN$, and an undirected $\Delta$-regular bipartite graph $\Gamma=(V=V_0\sqcup V_1,E,\ver)$ with an $\bF_q^t$-valued vertex labeling $\labV:V_0\sqcup V_1\rightarrow\bF_q^t$, where the edge set $E\subseteq\bF_q$. This latter condition can be viewed as giving an injective $\bF_q$-valued edge labeling $E\hookrightarrow\bF_q$ (which, unlike the edge labelings $\liftlab$ in Section~\ref{sec:ablift}, will \textit{not} be used to directly take an $\bF_q$-lift). We refer to $\Gamma$ as the base graph, and for simplicity we will assume that $\Gamma$ is balanced. Note that we allow multiple edges in $E$ to connect the same two vertices, but these edges must all have distinct values in $\bF_q$.

We then construct our desired graph $\bar{\Gamma}=(\bar{V}=\bar{V}_0\sqcup\bar{V}_1,\bar{E},\bar{\ver})$, which we refer to as the lifted graph (see Lemma~\ref{lem:islift} below), as follows. For $i=0,1$, we define the vertex set
\begin{align*}
  \bar{V}_i &= \bigsqcup_{v\in V_i}\bF_q^{t+1}/\spn\{(1,\labV(v))\},
\end{align*}
where $(1,\labV(v))\in\bF_q^{t+1}$ denotes the vector whose first component is $1$ and whose remaining $t$ components are given by $\labV(v)\in\bF_q^t$. We then define the edge set $\bar{E}=E\times\bF_q^t$, and for every $e\in E$, $x\in\bF_q^t$, and $b\in\{0,1\}$, we let
\begin{align*}
  \bar{\ver}_b(e,x) &= (e,x)+\spn\{(1,\labV(\ver_b(e)))\}.
\end{align*}

In words, the vertices of $\bar{\Gamma}$ correspond to affine lines in $\bF_q^t$ pointing in the directions $(1,\labV(v_0))$ for $v_0\in V_0$ and $(1,\labV(v_1))$ for $v_1\in V_1$. The edges of $\bar{\Gamma}$ correspond to pairs of such affine lines that intersect, such that the first component of the intersection point equals the value $e\in\bF_q$ of the associated edge $e$ in the base graph $\Gamma$, so that $\ver(e)=(v_0,v_1)$. Therefore for every $e\in E$, each point in $\{e\}\times\bF_q^t$ is associated to a unique edge in $\bar{E}$ connecting a vertex in $\bF_q^{t+1}/\spn\{(1,\labV(\ver_0(e)))\}$ to a vertex in $\bF_q^{t+1}/\spn\{(1,\labV(\ver_1(e)))\}$.

The following lemma shows that $\bar{\Gamma}$ is an $\bF_q^t$-lift of $\Gamma$ in the sense of Definition~\ref{def:graphlift}.

\begin{lemma}
  \label{lem:islift}
  The graph $\bar{\Gamma}=(\bar{V},\bar{E},\bar{\ver})$ is isomorphic to the $\bF_q^t$-lift $\tilde{\Gamma}=(\tilde{V},\tilde{E},\tilde{\ver})$ of $\Gamma$ with labeling $\liftlab:E\rightarrow\bF_q^t$ given by
  \begin{equation*}
    \liftlab(e) = e\cdot(\labV(\ver_0(e))-\labV(\ver_1(e))).
  \end{equation*}
\end{lemma}
\begin{proof}
  By Definition~\ref{def:graphlift}, $\tilde{V}=V\times\bF_q^t$ and $\tilde{E}=E\times\bF_q^t$.

  Define set isomorphisms $\phi_V:\tilde{V}\xrightarrow{\sim}\bar{V}$ and $\phi_E:\tilde{E}\xrightarrow{\sim}\bar{E}$ by
  \begin{align*}
    \phi_V(v,x) &= (0,x)+\spn\{(1,\labV(v))\} \\
    \phi_E(e,x) &= (e,x+e\cdot\labV(\ver_0(e))).
  \end{align*}
  Then by definition for every $(e,x)\in\tilde{E}$,
  \begin{align*}
    \phi_V(\tilde{\ver}_0(e,x))
    &= \phi_V(\ver_0(e),x) \\
    &= (0,x)+\spn\{(1,\labV(\ver_0(e)))\} \\
    &= (e,x+e\cdot\labV(\ver_0(e)))+\spn\{(1,\labV(\ver_0(e)))\} \\
    &= \bar{\ver}_0(e,x+e\cdot\labV(\ver_0(e))) \\
    &= \bar{\ver}_0(\phi_E(e,x)).
  \end{align*}
  and
  \begin{align*}
    \phi_V(\tilde{\ver}_1(e,x))
    &= \phi_V(\ver_0(e),x+\liftlab(e)) \\
    &= (0,x+e\cdot(\labV(\ver_0(e))-\labV(\ver_1(e))))+\spn\{(1,\labV(\ver_1(e)))\} \\
    &= (e,x+e\cdot\labV(\ver_0(e)))+\spn\{(1,\labV(\ver_1(e)))\} \\
    &= \bar{\ver}_1(e,x+e\cdot\labV(\ver_0(e))) \\
    &= \bar{\ver}_1(\phi_E(e,x)).
  \end{align*}
  Thus the images under $\phi_V$ of the vertices of every $(e,x)\in\tilde{E}$ equal the vertices of $\phi_E(e,x)\in\bar{E}$, so $\phi_V,\phi_E$ provide an isomorphism between $\tilde{\Gamma}$ and $\bar{\Gamma}$.
\end{proof}

It also follows by definition that $\bar{\Gamma}\cong\tilde{\Gamma}$ naturally respects a free action of the additive group $\bF_q^t$:

\begin{lemma}
  \label{lem:expaction}
  The 1-dimensional incidence complex associated to $\bar{\Gamma}$ (see Definition~\ref{def:graphinc}) respects a free action $\sigma$ of the additive group $\bF_q^t$, defined so that for $y\in\bF_q^t$, $\bar{v}\in\bar{V}$, and $\bar{e}\in\bar{E}$ we have $\sigma(y)\bar{v}=(0,y)+\bar{v}$ and $\sigma(y)\bar{e}=(0,y)+\bar{e}$.
\end{lemma}
\begin{proof}
  By the definition of $\bar{\Gamma}$, for $y\in\bF_q^t$, $\bar{e}\in\bar{E}$, and $b\in\{0,1\}$, we have $\bar{\ver}_b((0,y)+\bar{e})=(0,y)+\bar{\ver}_b(\bar{e})$, so the result follows.
\end{proof}

\subsection{Spectral Expansion Bound}
We now show that appropriate instantiations of the construction $\bar{\Gamma}$ in Section~\ref{sec:expconstruct} have good spectral expansion. In particular, we will choose the base graph $\Gamma$ to be a complete bipartite (multi)graph with multiple edges between every pair of vertices, and then we will choose the vertex labeling $\labV:V=V_0\sqcup V_1\rightarrow\bF_q^t$ at random from a low-bias distribution. Our main expansion bound is stated below.

\begin{theorem}
  \label{thm:expproof}
  Define all variables as in Section~\ref{sec:expconstruct}. Let $\Gamma=(V=V_0\sqcup V_1,E,\ver)$ be a $\Delta$-regular complete balanced bipartite (multi)graph, for some $\Delta\in[|V_0|,q/|V_0|]$ that is a multiple of $|V_0|=|V_1|$. Specifically, let $E\subseteq\bF_q$ be an arbitrary subset of size $|E|=|V_0|\cdot\Delta$, and choose $\ver:E\rightarrow V_0\times V_1$ such that for every $(v_0,v_1)\in V_0\times V_1$, there are precisely $\Delta/|V_0|$ edges $e\in E$ with $\ver(e)=(v_0,v_1)$.

  Also for some $k\in\bN$, let $\cD$ be a $\beta=1/|V|^{2k+1}$-biased distribution over $(\bF_q^t)^V$. Then for every $\eta\in(0,1)$,
  \begin{align*}
    \Pr_{\labV\sim\cD}[\lambda_2(\bar{\Gamma})\geq\eta\Delta]
    &< q^t \cdot \left(\frac{4k}{\eta\cdot|V|^{1/4}}\right)^{2k}.
  \end{align*}
\end{theorem}

We will prove Theorem~\ref{thm:expproof} in Section~\ref{sec:expboundproof} below. However, we first present the following corollary of Theorem~\ref{thm:expproof}, which shows how to instantiate parameters to obtain the explicit expanders that we will ultimately use to construct codes.

\begin{corollary}
  \label{cor:expinst}
  For every $\nu\in(0,1/2)$, $\delta\in(\nu,1-\nu]$, and $\eta\in(0,1)$, there exists $\tau=\tau(\nu)\in(0,\nu/32)$ and $q_0=q_0(\nu,\delta,\eta)\in\bN$ such that the following holds for every prime power $q\geq q_0$.

  Let $\Delta=\lfloor q^\nu\rfloor\cdot\lfloor q^{\delta-\nu}\rfloor$, and let $\Gamma=(V=V_0\sqcup V_1,E,\ver)$ be a $\Delta$-regular complete bipartite (multi)graph with $|V_0|=|V_1|=\lfloor q^\nu\rfloor$ and with $E\subseteq\bF_q$, as defined in Theorem~\ref{thm:expproof}.

  Then there exists a deterministic $|\bar{E}|^{O(1)}$-time algorithm that computes a labeling $\labV:V\rightarrow\bF_q^t$ for some integer $t\in(q^\tau,q^{\nu/32})$ such that the lifted graph $\bar{\Gamma}=(\bar{V},\bar{E},\bar{\ver})$ defined in Section~\ref{sec:expconstruct} with base graph $\Gamma$ has $\lambda_2(\bar{\Gamma})\leq\eta\Delta$.
\end{corollary}

To obtain an explicit (i.e.~polynomial-time computable) labeling $\labV$ in Corollary~\ref{cor:expinst}, we instantiate the distribution $\cD$ in Theorem~\ref{thm:expproof} with the explicit low-bias distributions of \cite{jalan_near-optimal_2021} described in Theorem~\ref{thm:lowbias}, and then perform a brute force serach over all $\labV$ in the support of $\cD$. This technique of obtaining explicit abelian lifts from such low-bias distributions is adapted from \cite{jeronimo_explicit_2022}. However, whereas \cite{jeronimo_explicit_2022} applied this idea to low-degree base graphs with random edge labels, we consider complete bipartite base graphs $\Gamma$ with random \textit{vertex} labels (see the discussion in Section~\ref{sec:expboundproof} below).

\begin{proof}[Proof of Corollary~\ref{cor:expinst}]
  Let $k=\lfloor q^\kappa\rfloor$ for $\kappa=\nu/32$, and let $\tau=\nu/64$. Let
  \begin{equation*}
    t = \left\lfloor\log_q\left(\frac12\left(\frac{\eta\cdot|V|^{1/4}}{4k}\right)^{2k}\right)\right\rfloor.
  \end{equation*}
  Then by definition
  \begin{align}
    \label{eq:tbound}
    t
    &= \frac{2k\log\left(\frac{\eta\cdot|V|^{1/4}}{4k}\right)}{\log q}+o(1) = 2q^\kappa\left(\frac{\nu}{4}-\kappa+o_\eta(1)\right)
  \end{align}
  where here $o(1)$ denotes a positive or negative function that approaches $0$ as $q\rightarrow\infty$. Because $\nu\in(0,1/2)$, it follows that indeed for all sufficiently large $q$ we have $t<q^\kappa=q^{\nu/32}$, and $t>q^{\kappa/2}=q^\tau$. Thus it remains to give an efficient algorithm to construct some $\labV$ for which $\lambda_2(\bar{\Gamma})\leq\eta\Delta$.
  
  For this purpose, by Theorem~\ref{thm:expproof} along with the definition of $t$, we have
  \begin{equation*}
    \Pr_{\labV\sim\cD}[\lambda_2(\bar{\Gamma})\geq\eta\Delta] < \frac12.
  \end{equation*}
  Therefore for every $\beta$-biased distribution $\cD$ over $\bF_q^t$, there exists some $\labV\in\supp(\cD)$ for which $\lambda_2(\bar{\Gamma})<\eta\Delta$. By Theorem~\ref{thm:lowbias}, there exists a $\poly(q^t,1/\beta)$-time algorithm that outputs a subset of $\bF_q^t$ of size $O(tq^{O(1)}/\beta^{2+o(1)})$ on which the uniform distribution has bias $\leq\beta$; we will set $\cD$ to be this distribution. Note that because $|\bar{E}|\in[q^t,q^{t+1}]$, it follows from~(\ref{eq:tbound}) that
  \begin{align*}
    1/\beta
    &= |V|^{2k+1} = q^{\Theta(\nu q^\kappa)} = q^{\Theta(t)} = |\bar{E}|^{\Theta(1)}.
  \end{align*}

  Thus our algorithm to construct the desired $\labV$ will simply loop through all $O(tq^{O(1)}/\beta^{2+o(1)})=|\bar{E}|^{O(1)}$ elements $\labV$ of the $\beta$-biased set $\supp(\cD)$ from Theorem~\ref{thm:lowbias}, construct $\bar{\Gamma}$ to compute $\lambda_2(\bar{\Gamma})$ for each such $\labV$, and then output whichever $\labV$ yields the smallest $\lambda_2(\bar{\Gamma})$.

  It is well known 
  that the spectral expansion of a graph (i.e.~the 2nd largest eigenvalue of the adjacency matrix) can be computed in polynomial time with respect to the number of edges. By Theorem~\ref{thm:lowbias}, each of the $|\bar{E}|^{O(1)}$ elements $\labV$ in the support of $\cD$ can be computed in time $\poly(|\bar{E}|,1/\beta)=|\bar{E}|^{O(1)}$. Thus the algorithm's overall running time is $|\bar{E}|^{O(1)}$, and for all sufficiently large $q$ we have shown that it must output some $\labV$ for which $\lambda_2(\bar{\Gamma})<\eta\Delta$, as desired.
\end{proof}

\subsection{Proof of Spectral Expansion Bound}
\label{sec:expboundproof}
In this section, we prove Theorem~\ref{thm:expproof}. Our proof is based on the trace power method, which bounds the spectrum of a matrix by bounding the trace of a high power of the matrix. For a signed adjacency matrix with random edge labels, this trace can in turn be bounded by counting walks in the underlying graph. Such techniques have been used to bound expansion of graph lifts in the past (see e.g.~\cite{mohanty_explicit_2021,jeronimo_explicit_2022,paredes_expansion_2022}).

Our analysis is similar to arguments in such prior works. However, such prior results bounding the expansion of abelian lifts typically considered a random edge labeling $\liftlab:E\rightarrow G$ (with $G=\bF_q^t$ in our setting). In contrast, we obtain an edge labeling $\liftlab$ from linear combinations of random \textit{vertex} labels (see Lemma~\ref{lem:islift}). Fortunately, we find that the trace power method applies naturally in our setting as well.

We remark that \cite{jeronimo_explicit_2022} also gave an alternative method for showing expansion of explicit abelian lifts (Section~6 of their paper), based on applying an expander-Chernoff-like bound. It is an interesting question whether this technique also extends to our setting with edge labels obtained as linear combinations of vertex labels.

\begin{proof}[Proof of Theorem~\ref{thm:expproof}]
  For $\labV\sim\cD$, define $\tilde{\Gamma}$ to be the $\bF_q^t$-lift of $\Gamma$ with labeling $\liftlab$ as given in Lemma~\ref{lem:islift}. Then by Lemma~\ref{lem:specunion},
  \begin{equation}
    \label{eq:specunion}
    \spectrum(A_{\tilde{\Gamma}}) = \spectrum(A_{\Gamma})\sqcup\bigsqcup_{\text{nontrivial characters }\chi:\bF_q^t\rightarrow\bC}\spectrum(A_{\Gamma,\chi}).
  \end{equation}
  By definition the the spectrum of the adjacency matrix $A_\Gamma$ of the $\Delta$-regular complete bipartite graph $\Gamma$ consists of $\Delta$, $-\Delta$ (each with multiplicity $1$), and $0$ (with multiplicty $|V|-2$). Thus it remains to be shown that if $\chi$ is nontrivial, then all eigenvalues of $A_{\Gamma,\chi}$ are $<\eta\Delta$ with high probability.

  For this purpose, fix a nontrivial character $\chi:\bF_q^t\rightarrow\bC$. By Lemma~\ref{lem:tracebound},
  \begin{equation}
    \label{eq:lambdatowalks}
    \lambda(A_{\Gamma,\chi})^{2k} \leq \tr(A_{\Gamma,\chi}^{2k}) = \sum_{\text{walks }(e_1,b_1),\dots,(e_{2k},b_{2k})}\prod_{i=1}^{2k}\chi((-1)^{b_i}\liftlab(e_i)),
  \end{equation}
  where we recall from Definition~\ref{def:signedadjmat} that $\lambda(A_{\Gamma,\chi})$ denotes the largest absolute value of any eigenvalue of $A_{\Gamma,\chi}$.

  By Definition~\ref{def:walk}, a length-$2k$ walk $(e_1,b_1),\dots,(e_{2k},b_{2k})$ visits a sequence of vertices $v_1,\dots,v_{2k+1}\in V$ with $v_1=v_{2k+1}$, where $v_i=\dir{\ver}_0(e_i,b_i)$ and $v_{i+1}=\dir{\ver}_1(e_i,b_i)$ for $i\in[2k]$. Then letting $e_0=e_{2k}$, the RHS of~(\ref{eq:lambdatowalks}) can be expressed as
  \begin{align}
    \label{eq:chietov}
    \begin{split}
      \prod_{i=1}^{2k}\chi\left((-1)^{b_i}\liftlab(e_i)\right)
      &= \prod_{i=1}^{2k}\chi\left((-1)^{b_i}\cdot e_i\cdot (\labV(\ver_0(e_i))-\labV(\ver_1(e_i)))\right) \\
      &= \prod_{i=1}^{2k}\chi\left(e_i\cdot(\labV(\dir{\ver}_0(e_i,b_i)) - \labV(\dir{\ver}_1(e_i,b_i)))\right) \\
      &= \prod_{i=1}^{2k}\chi\left((e_i-e_{i-1})\cdot\labV(v_i)\right),
    \end{split}
  \end{align}
  where the first equality above holds by the definition of $\liftlab$, and the third equality above uses the fact that that $\chi$ is a homomorphism and that $v_1=v_{2k+1}$. The above equation motivates the following definition of a redundant walk:

  \begin{definition}
    A length-$2k$ walk $((e_1,b_1),\dots,(e_{2k},b_{2k}))\in(\dir{E})^{2k}$ is said to be \textbf{redundant} if the sequence of vertices $(v_1,\dots,v_{2k})\in V^{2k}$ that it visits satisfies the following: for every $i\in[2k]$, either there exists some $j\in[2k]\setminus\{i\}$ with $v_i=v_j$, or else $e_{i-1}=e_i$ (where we denote $e_0=e_{2k}$).
  \end{definition}

  The following claim explains the importance of this definition. Recall below that we assume $\chi:\bF_q^t\rightarrow\bC$ is a nontrivial character, and $\labV\in(\bF_q^t)^V$ is sampled from the $\beta$-biased distribution $\cD$.

  \begin{claim}
    \label{claim:nrvanish}
    Let $((e_1,b_1),\dots,(e_{2k},b_{2k}))\in(\dir{E})^{2k}$ be a length-$2k$ walk that is not redundant. Then
    \begin{equation*}
      \left|\bE_{\labV\sim\cD}\left[\prod_{i=1}^{2k}\chi((-1)^{b_i}\liftlab(e_i))\right]\right| \leq \beta.
    \end{equation*}
  \end{claim}
  \begin{proof}
    The function $\chi':(\bF_q^t)^V\rightarrow\bC$ given by
    \begin{equation*}
      \chi'(\labV) := \prod_{i=1}^{2k}\chi\left((e_i-e_{i-1})\cdot\labV(v_i)\right)
    \end{equation*}
    is by definition a character of the group $(\bF_q^t)^V$. Non-redundancy implies that there exists some $i\in[2k]$ where $v_i\neq v_j$ for every $j\in[2k]\setminus\{i\}$, and $e_{i-1}\neq e_i$. Therefore $\chi'(\labV)$ only depends on $\labV(v_i)$ through the factor $\chi\left((e_i-e_{i-1})\cdot\labV(v_i)\right)$, and this factor is itself a nontrivial character because $e_i-e_{i-1}\neq 0$. It follows that $\chi'$ is nontrivial, so by~(\ref{eq:chietov}) we have
    \begin{align*}
      \left|\bE_{\labV\sim\cD}\left[\prod_{i=1}^{2k}\chi((-1)^{b_i}\liftlab(e_i))\right]\right|
      &= |\bE_{\labV\sim\cD}\chi'(\labV)| \leq \beta.
    \end{align*}
  \end{proof}

  Applying Claim~\ref{claim:nrvanish} in~(\ref{eq:lambdatowalks}) yields the following simpler inequality:
  \begin{align}
    \label{eq:lambdaboundwalks}
    \bE_{\labV\sim\cD}[\lambda(A_{\Gamma,\chi})^{2k}]
    &\leq |\{\text{length-$2k$ redundant walks on $\Gamma$}\}| + |V|\Delta^{2k}\cdot\beta.
  \end{align}
  Here we have used the fact that there are at most $|V|\Delta^{2k}$ walks of length $2k$ on $\Gamma$, as such a walk is specified by the starting vertex $v_1\in V$, along with $2k$ values in $[\Delta]$, the $i$th of which is used to specify the choice of $e_i$ from the $\Delta$ edges incident to $v_i$.

  The following claim gives a bound on the RHS of~(\ref{eq:lambdaboundwalks}).

  \begin{claim}
    \label{claim:countwalks}
    The number of length-$2k$ redundant walks on $\Gamma$ is at most
    \begin{equation*}
      \left(\Delta\cdot\frac{3k}{|V|^{1/4}}\right)^{2k}.
    \end{equation*}
  \end{claim}
  \begin{proof}
    We first show that for a length-$2k$ redundant walk $(e_1,b_1),\dots,(e_{2k},b_{2k})$, the sequence of vertices $(v_1,\dots,v_{2k})\in V^{2k}$ that it visits consists of $\leq 3k/2$ distinct vertices. Specifically, let $I\subseteq[2k]$ be the set of indices of vertices that appear exactly once in the sequence $v_1,\dots,v_{2k}$. Then by the definition of redundancy, for every $i\in I$ we have $v_{i-1}=v_{i+1}$ (where as usual we denote $v_0=v_{2k}$ and $v_1=v_{2k+1}$), so that $v_{i-1},v_{i+1}\notin I$. Thus $I\subseteq[2k]$ contains no pair of consecutive integers, so $|I|\leq k$. It follows that the number of distinct vertices in $\{v_1,\dots,v_{2k}\}$ is at most
    \begin{equation*}
      |I|+\frac{|[2k]\setminus I|}{2} = |I|+\frac{2k-|I|}{2} = k+\frac{|I|}{2} \leq \frac{3k}{2}.
    \end{equation*}

    Now there are ${|V|\choose\lfloor 3k/2\rfloor}$ choices of a subset $S\subseteq V$ of size $|S|=\lfloor 3k/2\rfloor$. For each such set $S$, the number of walks on $\Gamma$ that only visit vertices in $S$ is at most
    \begin{equation*}
      |S|^{2k}\cdot\left(\frac{\Delta}{|V_0|}\right)^{2k} \leq \left(\frac{3k}{2}\cdot\frac{\Delta}{|V_0|}\right)^{2k}
    \end{equation*}
    as there are $|S|^{2k}$ choices for the vertex sequence $(v_1,\dots,v_{2k})\in S^{2k}$, and then for every $i\in[2k]$ there are at most $\Delta/|V_0|$ choices for the edge $e_i$ connecting the pair of vertices $v_i,v_{i+1}$ (as $\Gamma$ is a complete bipartite (multi)graph of degree $\Delta$). Thus the total number of length-$2k$ redundant walks on $\Gamma$ is at most
    \begin{align*}
      {|V|\choose\lfloor 3k/2\rfloor} \cdot \left(\frac{3k}{2}\cdot\frac{\Delta}{|V_0|}\right)^{2k}
      &\leq |V|^{3k/2} \cdot \left(\frac{3k\Delta}{|V|}\right)^{2k} \leq \left(\Delta\cdot\frac{3k}{|V|^{1/4}}\right)^{2k},
    \end{align*}
    as desired.
  \end{proof}

  Applying Claim~\ref{claim:countwalks} in~(\ref{eq:lambdaboundwalks}), along with the fact that $\beta=1/|V|^{2k+1}$ by definition, gives that
  \begin{align*}
    \bE_{\labV\sim\cD}[\lambda(A_{\Gamma,\chi})^{2k}]
    &\leq \left(\Delta\cdot\frac{3k}{|V|^{1/4}}\right)^{2k} + \left(\frac{\Delta}{|V|}\right)^{2k} < \left(\Delta\cdot\frac{4k}{|V|^{1/4}}\right)^{2k}.
  \end{align*}
  Therefore
  \begin{align*}
    \Pr_{\labV\sim\cD}[\lambda(A_{\Gamma,\chi})\geq\eta\Delta]
    &\leq \frac{\bE_{\labV\sim\cD}[\lambda(A_{\Gamma,\chi})^{2k}]}{(\eta\Delta)^{2k}} < \left(\frac{4k}{\eta\cdot|V|^{1/4}}\right)^{2k}.
  \end{align*}
  Union bounding over all nontrivial characters $\chi:\bF_q^t\rightarrow\bC$ and applying~(\ref{eq:specunion}), we obtain the desired inequality
  \begin{align*}
    \Pr_{\labV\sim\cD}[\lambda_2(\bar{\Gamma})\geq\eta\Delta]
    &< q^t \cdot \left(\frac{4k}{\eta\cdot|V|^{1/4}}\right)^{2k}.
  \end{align*}
\end{proof}

\begin{remark}
  For our purposes in this paper, we will just need expansion $\lambda_2(\bar{\Gamma})\leq\eta\Delta$ for a sufficiently small constant $\eta>0$ as $\Delta\rightarrow\infty$. Therefore we emphasized simplicity in the proof of Theorem~\ref{thm:expproof}, and did not attempt to optimize the expansion bound. However, if the walks in Definition~\ref{def:walk} are replaced with \textit{non-backtracking} walks, which are not allowed to traverse backwards along the previous edge, tighter bounds can be shown. The reader is referred to \cite{jeronimo_explicit_2022} for an example application of such non-backtracking walks to prove tight expansion bounds, albeit for different base graphs than the $\Gamma$ we consider. We believe that such techniques could improve the factor of $|V|^{1/4}$ in the bound in Theorem~\ref{thm:expproof} to roughly $|V|^{1/2}$, by allowing the sets $S$ is Claim~\ref{claim:countwalks} to have size roughly $k$ instead of $3k/2$. However, such improvements are not needed for the purpose of our paper.
\end{remark}

\section{Classical LDPC Codes on the Expanders}
\label{sec:cldpc}
In this section, we construct classical LDPC codes following the Sipser-Spielman paradigm \cite{sipser_expander_1996}, in which we will impose local Reed-Solomon codes on the expanders given in Section~\ref{sec:expander} (specifically Corollary~\ref{cor:expinst}).

\begin{definition}
  \label{def:RMplant}
  Fix real numbers $\nu\in(0,1/2)$, $\delta\in(\nu,1-\nu]$, $\eta\in(0,1)$, and define $\tau,q_0,\Delta,t,\;\Gamma=(V,E,\ver),\;\bar{\Gamma}=(\bar{V},\bar{E},\bar{\ver})$ for some prime power $q\geq q_0$ as in Corollary~\ref{cor:expinst}. Let $\bar{\cX}$ be the incidence complex associated to $\bar{\Gamma}$ (see Definition~\ref{def:graphinc}). Also let $\ell\in[0,\Delta]$ be an integer. We now consider the 1-dimensional chain complex $\cF_{\bar{\cX}}$ for local coefficient system $\cF$ defined as in Definition~\ref{def:sscode}, with three different choices of the parameters $m$ and $(h_{\bar{v}}\in\bF_q^{m\times\bar{E}(\bar{v})})_{\bar{v}\in\bar{V}}$. Specifically, we will define $h_v\in\bF_q^{m\times E(v)}$ for $v\in V$, and then we let $h_{\bar{v}}=h_v$ for every $\bar{v}\in\bF_q^{t+1}/\spn\{(1,\labV(v))\}$, using the isomorphism $\bar{E}(\bar{v})\xrightarrow{\sim}E(v)$ given by projection onto the first coordinate.
  \begin{enumerate}
  \item\label{it:ker} Let $m=\Delta-\ell$. For $v\in V$, let $h_v\in\bF_q^{m\times E(v)}$ be a matrix with $\ker(h_v)=\evl_{E(v)}(\bF_q[X]^{<\ell})$.\footnote{Here there is a slight clash of notation: the polynomial indeterminate variable $X$ is distinct from the graded poset $\bar{X}$ associated to $\bar{\Gamma}$.}
  \item\label{it:im} Let $m=\ell$. For $v\in V$, let $h_v\in\bF_q^{m\times E(v)}$ be a matrix with $\im(h_v^\top)=\evl_{E(v)}(\bF_q[X]^{<\ell})$.
  \item\label{it:imcom} Let $m=\Delta-\ell$. For $v\in V$, let $h_v\in\bF_q^{m\times E(v)}$ be a matrix with $\im(h_v^\top)=\evl_{E(v)}(\bF_q[X]^{<|E|-\ell}_{E\setminus E(v)})$.
  \end{enumerate}
  We call the complex $\cF_{\bar{\cX}}$ arising from the three instantiations above a \textbf{1-dimensional RM-planted complex} of \textbf{type} \ref{it:ker}, \ref{it:im}, \ref{it:imcom}, respectively, with parameters $\nu,\delta,\eta,q,\ell$.
\end{definition}

We will ultimately obtain our qLDPC codes in Section~\ref{sec:qldpc} from tensor products of type-\ref{it:im} and type-\ref{it:imcom} RM-planted cochain complexes, and we obtain our cLTCs in Section~\ref{sec:cltc} from balanced products of type-\ref{it:ker} RM-planted chain complexes.

The name ``RM-planted'' comes from the fact that certain cycle and cocycle spaces of these complexes contain ``planted'' copies of a Reed-Muller (RM) code, as shown in the following lemma.

\begin{lemma}
  \label{lem:RMplant}
  Define all variables as in Definition~\ref{def:RMplant}. The RM-planted complexes $\cF_{\bar{\cX}}$ of type \ref{it:ker}, \ref{it:im}, \ref{it:imcom} respectively satisfy the following properties:
  \begin{enumerate}
  \item[\ref{it:ker}.] $\evl_{E\times\bF_q^t}(\bF_q[X_0,\dots,X_t]^{<\ell}) \subseteq Z_1(\cF_{\bar{\cX}})$.
  \item[\ref{it:im}.] Define $\iota^0:\bF_q[X_0,\dots,X_t]^{<\ell}\rightarrow\cF_{\bar{\cX}}^0$ as follows. For every $f\in\bF_q[X_0,\dots,X_t]^{<\ell}$ and every vertex $\bar{v}=x+\spn\{(1,\labV(v))\}\in\bar{V}=\bar{X}(0)$ for $x\in\{0\}\times\bF_q^t$, $v\in V$, we let $\iota^0(f)_{\bar{v}}\in\bF_q^\ell$ be the unique vector satisfying
    \begin{equation*}
      h_v^\top(\iota^0(f)_{\bar{v}}) = \evl_{E(v)}(f(X_0,\;x_1+\labV(v)_1X_0,\;\dots\;,x_t+\labV(v)_tX_0)).
    \end{equation*}
    Then $\im(\iota^0)\subseteq Z^0(\cF_{\bar{\cX}})$. Furthermore, $\iota^0|_{\bF_q[X_1,\dots,X_t]^{<\ell}}$ is injective.
  \item[\ref{it:imcom}.] $B^1(\cF_{\bar{\cX}}) \subseteq \evl_{E\times\bF_q^t}(\bF_q[X_0,\dots,X_t]^{<t(q-1)+|E|-\ell})$.
  \end{enumerate}
\end{lemma}
\begin{proof}
  \begin{enumerate}
  \item[\ref{it:ker}.] By definition $Z_1(\cF_{\bar{\cX}})$ consists of those elements of $\bF_q^{E\times\bF_q^t}$ whose evaluation the intersection of every affine line $\bar{v}\in\bar{V}$ with $E\times\bF_q^{t-1}$ agrees with a univariate polynomial of degree $<\ell$. But by definition the restriction of every polynomial in $\bF_q[X_0,\dots,X_t]^{<\ell}$ to an affine line is a univariate polynomial of degree $<\ell$, so the desired inclusion follows.
  \item[\ref{it:im}.] Consider an arbitrary $f\in\bF_q[X_0,\dots,X_t]^{<\ell}$ and $\bar{e}\in\bar{E}=\bar{X}(1)=E\times\bF_q^t$. Let $e=\bar{e}_0\in E$. Also, for $b\in\{0,1\}$ let $\bar{v}_b=\bar{\ver}_b(\bar{e})$, and define $x^b\in\{0\}\times\bF_q^t$, $v_b\in V_b$ so that $\bar{v}_b=x^b+\spn\{(1,\labV(v))\}$. Then by definition,
    \begin{align*}
      (\delta^{\cF_{\bar{\cX}}}_0(\iota^0(f)))_{\bar{e}}
      &= \delta^{\bar{\cX}}_{\bar{e},\bar{v}_0}\cdot(h_{v_0}^\top(\iota^0(f)_{\bar{v}_0}))_e + \delta^{\bar{\cX}}_{\bar{e},\bar{v}_1}\cdot(h_{v_1}^\top(\iota^0(f)_{\bar{v}_1}))_e \\
      &= f(e,\; x_1^0+\labV(v_0)_1e,\; \dots\;, x_t^0+\labV(v_0)_te) - f(e,\; x_1^1+\labV(v_1)_1e,\; \dots\;, x_t^1+\labV(v_0)_te) \\
      &= f(\bar{e}) - f(\bar{e}) \\
      &= 0,
    \end{align*}
    where the third equality above holds because $\bar{\Gamma}$ is defined so that $\bar{e}\in E\times\bF_q^t$ equals the unique point in both the affine lines $\bar{v}_0$ and in $\bar{v}_1$ whose first (i.e.~index $0$) component equals $e$. Thus $\im(\iota^0)\subseteq\ker(\delta^{\cF_{\bar{\cX}}}_0)$.

    To see that $\iota^0|_{\bF_q[X_1,\dots,X_t]^{<\ell}}$ is injective, consider some nonzero $f\in\bF_q[X_1,\dots,X_t]^{<\ell}$. Because $\ell\leq q$, there exists some $\bar{e}\in E\times\bF_q^t$ such that $f(\bar{e})\neq 0$. Then defining $e,\bar{v}_0,v_0$ as above, we have $(h_{v_0}^\top(\iota^0(f)_{\bar{v}_0}))_e=f(\bar{e})\neq 0$, so $\iota^0(f)\neq 0$, as desired.
  \item[\ref{it:imcom}.] Consider an arbitrary $\bar{v}\in\bar{V}$ and $c\in\cF_{\bar{v}}=\bF_q^m=\bF_q^{\Delta-\ell}$. Let $\bar{v}=x+\spn\{(1,\labV(v))\}$ for $x\in\{0\}\times\bF_q^t$ and $v\in V$. Also let $f(X_0)\in\bF_q[X_0]^{<|E|-\ell}$ be the unique polynomial for which $h_v^\top c=\evl_{E(v)}(f)$. We may view $c$ as an element of $\cF_{\bar{\cX}}^0=\bigoplus_{\bar{v}'\in\bar{V}}\cF_{\bar{v}'}$ supported on $\cF_{\bar{v}}=\bF_q^m$. Then by definition, $\delta^{\cF_{\bar{\cX}}}_0(c)$ is an element of $\bF_q^{E\times\bF_q^t}$, consisting of evaluations of $f$ placed on points on the affine line $\bar{v}$, and with $0$s elsewhere. Formally, define $g(X_0,\dots,X_t)\in\bF_q[X_0,\dots,X_t]^{\leq t(q-1)}$ by
    \begin{equation*}
      g(X_0,\dots,X_t) = \prod_{i=1}^t(1-(X_t-x_t-\labV(v)_tX_0)^{q-1}),
    \end{equation*}
    so that for every $x\in E\times\bF_q^t$ we have $g(x_0,\dots,x_t)=\1_{x\in\bar{v}}$. Then by definition
    \begin{equation*}
      \delta^{\cF_{\bar{\cX}}}_0(c) = \evl_{E\times\bF_q^t}(f(X_0)g(X_0,\dots,X_t)).
    \end{equation*}
    Now the polynomial on the RHS above has degree $<t(q-1)+|E|-\ell$. As every $c'\in\cF_{\bar{\cX}}^1$ can be written as a linear combination of $c\in\cF_{\bar{v}}$ for $\bar{v}\in\bar{V}$ as considered above, it follows that $\im(\delta^{\cF_{\bar{\cX}}}_0)\subseteq\evl_{E\times\bF_q^t}(\bF_q[X_0,\dots,X_t]^{<t(q-1)+|E|-\ell})$, as desired.
  \end{enumerate}
\end{proof}

These planted RM codes will help us in two regards. First, the planted RM codes allow us to bound the dimension of the associated codes codes (Definition~\ref{def:cctocode}) of the chain complex $\cF_{\bar{\cX}}$, when naive counting of linear constraints fails to show positive dimension. Second, the planted RM codes provide explicit codewords that can be multiplied component-wise to give codewords of higher-degree RM codes; this property will be important for the transversal $C^{r-1}Z$ behavior.

This idea of planting Reed-Muller codes in a Sipser-Spielman construction with local Reed-Solomon codes was used in \cite{dinur_new_2023}, and more generally in the classical literature on ``lifted codes.'' In the quantum setting, \cite{golowich_nlts_2024} showed how to plant the all-1s vector (which can be viewed as a degree-$0$ Reed-Muller code) in qLDPC codes; our results will provide one way of obtaining higher-dimensional planted qLDPC codes.

Note that by Lemma~\ref{lem:RSparameters} and Corollary~\ref{cor:SRStoRS}, the local codes $\ker(h_v)$ and dual local codes $\im(h_v^\top)$ in Definition~\ref{def:RMplant} will have good distance for appropriate choices of $\ell$. Hence the global codes $Z^0(\cF_{\bar{\cX}})$ and $Z_1(\cF_{\bar{\cX}})$ will also have good distance by Lemma~\ref{lem:ssdis}.


\section{Quantum LDPC Codes with Transversal $C^{r-1}Z$ via Tensor Product}
\label{sec:qldpc}
In this section, we take tensor products of the classical LDPC codes from Section~\ref{sec:cldpc} to obtain quantum LDPC codes, which we show in appropriate parameter regimes have almost linear dimension, polynomial distance, polylogarithmic locality, and support transversal $C^{r-1}Z$ gates.

\begin{theorem}
  \label{thm:qldpcmain}
  For every fixed $\nu\in(0,1/2)$, $\delta\in(\nu,1-\nu]$, $r\in\bZ_{\geq 2}$, the following holds for every sufficiently large prime power $q$. Let $\eta=\eta(r)=1/50r$, and define $\tau,\Delta,t,\Gamma,\bar{\Gamma}$ as in Corollary~\ref{cor:expinst}. Let $\ell=\lfloor\Delta/2\rfloor$ and $\ell'=\lfloor\ell/10r\rfloor$. Let $\cF_{\bar{\cX}}$ be a 1-dimensional type-\ref{it:imcom} RM-planted complex with parameters $\nu,\delta,\eta,q,\ell$, and let $\cF_{\bar{\cX}}'$ be a 1-dimensional type-\ref{it:im} RM-planted complex with parameters $\nu,\delta,\eta,q,\ell'$. For $h\in[r]$, define the tensor products
  \begin{align*}
    \cC^{(h)} &= {\cF_{\bar{\cX}} '}^{\otimes h-1}\otimes \cF_{\bar{\cX}} \otimes {\cF_{\bar{\cX}}'}^{\otimes r-h}.
  \end{align*}
  Then the quantum codes at level $1$ of $\cC^{(h)}$ for $h\in[r]$ are
  \begin{equation*}
    [[N=q^{rt+O_r(1)},\; K\geq N^{\delta-\nu},\; D\geq N^{1/r}/\log(N)^{O_{\nu,\delta,r}(1)}]]_q
  \end{equation*}
  codes of locality $w^{\cC^{(h)}}\leq\log(N)^{O_{\nu,\delta,r}(1)}$. Furthermore, there exists an $r$-tuple $H'=({H^{(h)}}'\subseteq H^1(\cC^{(h)}))_{h\in[r]}$ of 1-cohomology subspaces on which there is a coboundary-invariant multilinear form
  \begin{equation*}
    \zeta:(\cC^{(1)})^1\times\cdots\times(\cC^{(r)})^1\rightarrow\bF_q
  \end{equation*}
  of locality $w^\zeta\leq\log(N)^{O_{\nu,\delta,r}(1)}$, and such that $\subrank(\zeta_{H'})\geq N^{\delta-\nu}$.
\end{theorem}

Before proving Theorem~\ref{thm:qldpcmain}, we present the following corollary.

\begin{corollary}
  \label{cor:qldpcmain}
  For every $r\in\bZ_{\geq 2}$, every $\epsilon>0$, and every prime power $q$ (including $q=2$), there exists an infinite family of $r$-tuples of 2-dimensional cochain complexes $(\cC^{(h)})_{h\in[r]}$ for which the associated quantum codes at level $1$ are
  \begin{equation*}
    [[N=\dim(\cC^{(h)})^1,\; K\geq N^{1-\epsilon},\; D\geq N^{1/r}/\log(N)^{O_{\epsilon,r}(1)}]]_q
  \end{equation*}
  codes of locality $w^{\cC^{(h)}}\leq\log(N)^{O_{\epsilon,r}(1)}$, which support a transversal $C^{r-1}Z$ gate inducing $N^{1-\epsilon}$ logical $C^{r-1}Z$ gates (in the sense of Definition~\ref{def:transversal}). Furthermore, there exists a $\poly(N)$-time algorithm to construct $(\cC^{(h)})_{h\in[r]}$.
\end{corollary}
\begin{proof}
  Set $\nu=\epsilon/4$ and $\delta=1-\epsilon/4$. Then for an arbitrary sufficiently large power $q'$ of $q$, we obtain $(\cC^{(h)})_{h\in[r]}$ by truncating the complexes from Theorem~\ref{thm:qldpcmain} with parameters $\nu,\delta,r,q$ to levels $0,1,2$, and then applying Lemma~\ref{lem:alphred} followed by Lemma~\ref{lem:locred} to these complexes to reduce the alphabet size to $q$ and the multilinear form locality $w^\zeta$ to $1$. As Lemma~\ref{lem:alphred} and Lemma~\ref{lem:locred} both preserve the code length, dimension, distance and locality, as well as the multilinear form subrank, up to factors of $\log(N)^{O_{\epsilon,r}(1)}$, the resulting construction has all parameters as stated in the corollary.

  Meanwhile, Corollary~\ref{cor:expinst} ensures that the graph $\bar{\Gamma}$ underlying the complexes from Theorem~\ref{thm:qldpcmain} is constructable in polynomial time. As the local Reed-Solomon codes used to construct these complexes are explicit, the complexes from Theorem~\ref{thm:qldpcmain} are constructable in polynomial time. The transformations given by Lemma~\ref{lem:alphred} and Lemma~\ref{lem:locred} can also by definition be performed in polynomial time, so we have a $\poly(N)$-time algorithm to construct $(\cC^{(h)})_{h\in[r]}$, as desired.
\end{proof}

\begin{remark}
  \label{remark:isomorphiccodes}
  The complexes $\cC^{(h)}$ for $h\in[r]$ in Corollary~\ref{cor:qldpcmain} are by definition all isomorphic, with isomorphisms given by appropriate permutations of the basis elements of the cochain spaces (or equivalently, permutations of the qudits of the associated quantum codes).
\end{remark}

We now prove Theorem~\ref{thm:qldpcmain} in the sections below. Throughout these sections, we define all variables as in Theorem~\ref{thm:qldpcmain}.

\subsection{Proof of Basic Code Parameters}
\label{sec:basicparamproof}
In this section, we bound the length, dimension, distance, and locality of the codes in Theorem~\ref{thm:qldpcmain}:

\begin{lemma}
  \label{lem:basicparam}
  The quantum codes at level $1$ of $\cC^{(h)}$ for $h\in[r]$ as defined in Theorem~\ref{thm:qldpcmain} are
  \begin{equation*}
    [[N=q^{rt+O_r(1)},\; K\geq N^{\delta-\nu},\; D\geq N^{1/r}/\log(N)^{O_{\nu,\delta,r}(1)}]]_q
  \end{equation*}
  codes of locality $w^{\cC^{(h)}}\leq\log(N)^{O_{\nu,\delta,r}(1)}$.
\end{lemma}
\begin{proof}
  By definition, for every $h\in[r]$ the quantum code at level $1$ of $\cC^{(h)}$ has length
  \begin{align}
    \label{eq:codelenbound}
    \begin{split}
      N
      &= \dim{\cC^{(h)}}^1 \\
      &= (r-1)\cdot\dim({\cF'_{\bar{\cX}}}^1)\cdot\dim(\cF_{\bar{\cX}}^0)\cdot\dim({\cF_{\bar{\cX}}'}^0)^{r-2} + \dim(\cF_{\bar{\cX}}^1)\cdot\dim({\cF_{\bar{\cX}}'}^0)^{r-1} \\
      &= (r-1)\cdot|\bar{E}|\cdot(\Delta-\ell)|\bar{V}|\cdot(\ell'|\bar{V}|)^{r-2} + |\bar{E}|\cdot(\ell'|\bar{V}|)^{r-1} \\
      &= |\bar{E}|\cdot\Theta_r(\Delta^{r-1})\cdot|\bar{V}|^{r-1} \\
      &= \Theta_r(|\bar{E}|^r) \\
      &= \Theta_r(|E|^rq^{rt}) \\
      &= q^{rt+O_r(1)},
    \end{split}
  \end{align}
  where the final equality above holds because $1\leq|E|\leq q$.\footnote{As a point of notation, here we let $O_r(1)$ be positive or negative, as long as its absolute value is bounded by a value depending only on $r$.}

  By the K\"{u}nneth formula (Lemma~\ref{lem:kunneth}), this code has dimension
  \begin{align}
    \label{eq:codedimbound}
    \begin{split}
      K
      &= \dim H^1(\cC^{(h)}) \\
      &\geq \dim(H^1(\cF_{\bar{\cX}}))\cdot\dim(H^0(\cF_{\bar{\cX}}'))^{r-1} \\
      &\geq {\ell+t\choose t}\cdot{\ell'-1+t\choose t}^{r-1} \\
      &\geq (\ell/t)^t\cdot(\ell'/t)^{(r-1)t} \\
      &\geq (\Delta/21rt)^{rt} \\
      &\geq (q^{\delta-\nu/2})^{rt} \\
      &\geq N^{\delta-\nu}
    \end{split}
  \end{align}
  The second inequality in~(\ref{eq:codedimbound}) follows from Lemma~\ref{lem:RMplant}. Specifically, the bound on type~\ref{it:imcom} complexes in Lemma~\ref{lem:RMplant} gives that
  \begin{equation*}
    B^1(\cF_{\bar{\cX}}) \subseteq \evl_{E\times\bF_q^t}(\bF_q[X_0,\dots,X_t]^{<t(q-1)+|E|-\ell}),
  \end{equation*}
  which implies that the evaluations on $E\times\bF_q^t$ of the ${\ell+t\choose t}$ monomials $X_0^{|E|-1}X_1^{q-1-j_1}\cdots X_t^{q-1j_t}$ for $(j_i\in\{0,\dots,\ell\})_{i\in[t]}$ with $\sum_{i=1}^tj_i\leq\ell$ generate linearly independent cohomology classes in $H^1(\cF_{\bar{\cX}})=\bF_q^{E\times\bF_q^t}/B^1(\cF_{\bar{\cX}})$. Similarly, the bound on type~\ref{it:im} complexes in Lemma~\ref{lem:RMplant} implies that
  \begin{equation*}
    \dim(H^0(\cF_{\bar{\cX}}')) = \dim(Z^0(\cF_{\bar{\cX}}')) \geq \dim(\bF_q[X_1,\dots,X_t]^{<\ell}) = {\ell-1+t\choose t}.
  \end{equation*}
  The fifth inequality in~(\ref{eq:codedimbound}) follows because by definition, for fixed $\nu,\delta,r$ we have $\Delta=q^{\delta-o(1)}$ and $t\leq q^{\nu/32}$. The sixth inequality in~(\ref{eq:codedimbound}) then follows from~(\ref{eq:codelenbound}).

  By Lemma~\ref{lem:proddis}, the quantum code at level $1$ of $\cC^{(h)}$ has distance
  \begin{align*}
    D
    &= \min\{d^1(\cC^{(h)}),d_1(\cC^{(h)})\} \\
    &\geq \min\{d^0(\cF_{\bar{\cX}}),d^0(\cF_{\bar{\cX}}'),d_1(\cF_{\bar{\cX}}),d_1(\cF_{\bar{\cX}}')\}.
  \end{align*}
  By Lemma~\ref{lem:RSparameters} and Corollary~\ref{cor:SRStoRS}, all of the local codes $\ker(h_{\bar{v}}),\ker(h_{\bar{v}}')$ and dual local codes $\im(h_{\bar{v}}^\top),\im({h_{\bar{v}}'}^\top)$ used to define $\cF$ and $\cF'$ have distance at least $\min\{\ell+1,\Delta-\ell+1,\ell'+1,\Delta-\ell'+1\}=\ell'+1\geq\Delta/21r$. Therefore because $\eta=1/50r$, it follows by Lemma~\ref{lem:ssdis} that
  \begin{align*}
    D
    &\geq \left(\frac{\Delta}{21r}-\eta\Delta\right)\cdot\frac{1}{\Delta}\cdot|\bar{V}| \\
    &\geq \frac{\Delta}{50r}\cdot\frac{1}{\Delta}\cdot\frac{2}{\Delta}\cdot|\bar{E}| \\
    &\geq \frac{q^{t-\delta}}{25r} \\
    &\geq \frac{N^{1/r}}{\log(N)^{O_{\nu,\delta,r}(1)}},
  \end{align*}
  where the fourth inequality above holds by~(\ref{eq:codelenbound}), so that in particular $\log(N)\geq t\geq q^\tau=q^{\Theta_{\nu,\delta,r}(1)}$, and thus every polynomial $\poly(q)=q^{O_{\nu,\delta,r}(1)}$ is bounded above by $\log(N)^{O_{\nu,\delta,r}(1)}$.

  By the definition of a chain complex tensor product, we have the locality bound
  \begin{equation*}
    w^{\cC^{(h)}} \leq w^{\cF_{\bar{\cX}}}+(r-1)w^{\cF_{\bar{\cX}}'}.
  \end{equation*}
  Meanwhile, by definition we have
  \begin{equation}
    \label{eq:locbound}
    w^{\cF_{\bar{\cX}}},w^{\cF_{\bar{\cX}}'}\leq 2\Delta.
  \end{equation}
  Specifically, each basis element of $\cF_{\bar{\cX}}^0$ or of ${\cF_{\bar{\cX}}'}^0$ lies within $\cF_{\bar{v}}$ or $\cF'_{\bar{v}}$ respectively for some vertex $\bar{v}\in\bar{\Gamma}$, and hence such a basis element can only be incident to (i.e.~have a nonzero coboundary matrix element with) the $\Delta$ basis elements of $\cF_{\bar{\cX}}^1=\bF_q^{\bar{E}}={\cF_{\bar{\cX}}'}^1$ corresponding to the $\Delta$ edges incident to vertex $\bar{v}$. Similarly, each basis element of $\cF_{\bar{\cX}}^1=\bF_q^{\bar{E}}={\cF_{\bar{\cX}}'}^1$ corresponding to some edge $\bar{e}\in\bar{E}$ can only be incident in $\cF_{\bar{\cX}}$ or $\cF_{\bar{\cX}}'$ to the $\leq\Delta$ basis elements of $\cF_{\bar{v}}$ or $\cF_{\bar{v}}'$ respectively for the two vertices $\bar{v}$ in the edge $\bar{e}$. Thus
  \begin{align*}
    w^{\cC^{(h)}}
    &\leq r\cdot 2\Delta \leq 2rq^\delta \leq \log(N)^{O_{\nu,\delta,r}(1)},
  \end{align*}
  as desired.
\end{proof}

\subsection{Definition of Multilinear Form}
\label{sec:zetadef}
In this section, we define the multilinear form $\zeta$ in Theorem~\ref{thm:qldpcmain}, and we bound the locality $w^\zeta$ of $\zeta$.
Here we provide a formal presentation; for more intuition regarding the definition of $\zeta$, the reader is referred to Appendix~\ref{sec:ciformintuition}.

As described in (the $G=\{0\}$ case of) Example~\ref{example:graphproduct}, our complex $\cC^{(h)}$ for $h\in[r]$ is equal to $\cG^{(h)}_{\cY}$ for $\cY=\bar{\cX}^{\otimes r}$ and
\begin{align*}
  \cG^{(h)}_{x_1\times\cdots\times x_r} &= \cF'_{x_1}\otimes\cdots\otimes\cF'_{x_{h-1}}\otimes\cF_{x_h}\otimes\cF'_{x_{h+1}}\otimes\cdots\otimes\cF'_{x_r}
\end{align*}
with the local coefficient maps $\cG^{(h)}_{y'\leftarrow y}$ for $y'\triangleright y$ defined as in Example~\ref{example:graphproduct}. The reader is also referred to Example~\ref{example:graphproduct} for an interpretation of $\cY$ as a ``cubical complex.'' In particular, we will sometimes refer to elements of $Y(0)$ as ``vertices,'' edges of $Y(1)$ as ``edges,'' and elements of $Y(i)$ for $i\in[r]$ as ``$i$-dimensional faces/cubes.''

To begin, let
\begin{equation}
  \label{eq:adef}
  a=\lfloor\ell/10rt\rfloor,
\end{equation}
and fix arbitrary sets $A_0\subseteq E\setminus\{0\}$ and $A_i\subseteq\bF_q^*$ for $i\in[t]$ such that $|A_0|=1$ and $|A_i|=a$. Recall here that $E\subseteq\bF_q$ is the edge set of the graph $\Gamma$ as defined in Section~\ref{sec:expconstruct}. Let
\begin{equation}
  \label{eq:Adef}
  A=A_0\times A_1\times\cdots\times A_t.
\end{equation}

We now use this set $A$ to define a linear functional $\alpha:\bF_q^{Y(r)}\rightarrow\bF_q$ as follows. Given $c\in\bF_q^{Y(r)}$, because $Y^{(r)}=(E\times\bF_q^t)^r$ where $E\subseteq\bF_q$, there exists a unique polynomial $g(U_i^{(h)})_{i\in\{0,\dots,t\}}^{h\in[r]}\in\bF_q[(U_i^{(h)})_{i\in\{0,\dots,t\}}^{h\in[r]}]$ that has degree $<|E|$ in each variable $U_0^{(h)}$, has degree $<q$ in each variable $U_i^{(h)}$ for $i\in[t]$, and has $g(u)=c_u$ for every $u\in Y^{(r)}$. Let $g'\in\bF_q[(U_i^{(h)})_{i\in\{0,\dots,t\}}^{h\in[r]}]$ be the polynomial obtained from $g$ by removing (i.e.~zeroing out the coefficient of) every monomial $\prod_{i,h}(U_i^{(h)})^{j_i^{(h)}}$ in $g$ such that for some $h\in[r]$, it holds that $\sum_{i=0}^tj_i^{(h)}\geq t(q-1)+|E|-\ell+(r-1)\ell'$. Then define $\alpha(c)=\sum_{u\in A^r}g'(u)$. By construction $\alpha:\bF_q^{Y(r)}\rightarrow\bF_q$ is a linear functional, so it has the form $\alpha(c)=\sum_{y\in Y(r)}\alpha_yc_y$ for some coefficients $\alpha_y\in\bF_q$ for $y\in Y(r)$. Thus we can also view $\alpha=(\alpha_y)_{y\in Y(r)}$ as a vector in $\bF_q^{Y(r)}$.

We let
\begin{equation}
  \label{eq:zetadef}
  \zeta(f^{(1)},\dots,f^{(r)}) = \alpha(\xi(f^{(1)},\dots,f^{(r)}))
\end{equation}
for a multilinear function
\begin{equation*}
  \xi:(\cG_{\cY}^{(1)})^1\times\cdots\times(\cG_{\cY}^{(r)})^1 \rightarrow \bF_q^{Y(r)}
\end{equation*}
that we will now define. This function $\xi$ will bear some resemblance to the cup product from algebraic topology. More topological intuition is given in Appendix~\ref{sec:ciformintuition}.

To define $\xi$, we first define a ``local'' $r$-multilinear form over $\bZ$
\begin{equation*}
  \xi^{\loc}:({\cY^{\loc}}^1)^r \rightarrow {\cY^{\loc}}^r=\bZ^{Y^{\loc}(r)}=\bZ,
\end{equation*}
which acts on 1-cochains of the $O_r(1)$-sized ``local'' boolean hypercube incidence complex $\cY^{\loc}$ defined in Example~\ref{example:graphproduct}. We will then define $\xi$ globally by applying $\xi^{\loc}$ locally within each $r$-dimensional cube in $Y(r)$.

For an edge $e\in Y^{\loc}(1)$, we define the \textbf{direction} $\direc(e)\in[r]$ to equal the position of the unique '$*$' in $T(e)$, so that $T(e)_{\direc(e)}=*$ and $T(e)_i\in\{0,1\}$ for every $i\in[r]\setminus\{\direc(e)\}$. Similarly, for $e\in Y(1)$, we let $\direc(e)=\direc(T(e))$.

For a permutation $\pi\in \cS_r$ let $(e^1(\pi),\dots,e^r(\pi))\in Y^{\loc}(1)^r$ denote the unique length-$r$ path in the $r$-dimensional boolean hypercube starting at vertex $0^r$ and ending at vertex $1^r$, such that the $i$th edge $e^i(\pi)$ points in direction $\direc(e^i(\pi))=\pi(i)$. Then for $(f^{(1)},\dots,f^{(r)})\in({\cY^{\loc}}^1)^r$, we define
\begin{equation}
  \label{eq:xilocdef}
  \xi^{\loc}(f^{(1)},\dots,f^{(r)}) = \sum_{\pi\in \cS_r}\sgn(\pi)\cdot f^{(1)}_{e^1(\pi)}\cdots f^{(r)}_{e^r(\pi)}.
\end{equation}
Note that by replacing $\bZ^{Y(i)}$ with $\bF_q^{Y(i)}$ for $i=1,r$, we may also view $\xi^{\loc}$ as a bilinear form over $\bF_q$, and $\xi^{\loc}$ commutes with the natural ring homomorphism $\bZ\rightarrow\bF_q$.

Given $f^{(h)}\in(\cG_{\cY}^{(h)})^i$ and $y\in Y(j)$ for $i\leq j$, we define $f^{(h)}|_y\in(\cG^{(h)}_y)^{Y_{\preceq y}(i)}$ by $(f^{(h)}|_y)_{y'}=\cG^{(h)}_{y\leftarrow y'}f^{(h)}_{y'}$. In particular, for a 1-cochain $f^{(h)}\in(\cG_{\cY}^{(h)})^1$ and an $r$-dimensional cube $y\in Y(r)$, then $f^{(h)}|_y\in\bF_q^{Y_{\preceq y}(1)}$ assigns the value $\cG^{(h)}_{y\leftarrow y'}f^{(h)}_{y'}\in\bF_q$ to each edge $y'$ in the $r$-dimensional cube $y$. The type function $T:Y\rightarrow\{0,1,*\}^r=Y^{\loc}$ from Example~\ref{example:graphproduct} induces a natural isomorphism $T|_{Y_{\preceq y}}:Y_{\preceq y}\rightarrow Y^{\loc}$, which we can apply component-wise to $f^{(h)}|_y$ to obtain $T(f^{(h)})|_y\in(\cY^{\loc})^1$. We also let $e^i_y(\pi)=T|_{Y\preceq y}^{-1}(e^i(y))$. Then for $y\in Y(r)$ and $(f^{(1)},\dots,f^{(r)})\in(\cG_{\cY}^{(1)})^1\times\cdots\times(\cG_{\cY}^{(r)})^1$, we define
\begin{align*}
  \begin{split}
    \xi^{\loc}_y(f^{(1)}|_y,\dots,f^{(r)}|_y)
    &= \xi^{\loc}(T(f^{(1)}|_y),\dots,T(f^{(r)}|_y)) \\
    &= \sum_{\pi\in\cS_r}\sgn(\pi) \cdot (f^{(1)}|_y)_{e^1_y(\pi)}\cdots (f^{(r)}|_y)_{e^r_y(\pi)},
  \end{split}
\end{align*}
and we let
\begin{align}
  \label{eq:xidef}
  \xi(f^{(1)},\dots,f^{(r)})_y
  &= \xi^{\loc}_y(f^{(1)}|_y,\dots,f^{(r)}|_y).
\end{align}

Now that we have defined $\zeta$, we may bound its locality (in the sense of Definition~\ref{def:cobinv}).

\begin{lemma}
  \label{lem:zetaloc}
  The multilinear form $\zeta$ has locality $w^\zeta\leq q^{O_r(1)}\leq \log(N)^{O_{\nu,\delta,r}(1)}$.
\end{lemma}
\begin{proof}
  For some $h\in[r]$, let $b^{(h)}\in C^{(h)}(1)$ be a 1-cochain basis element of $\cC^{(h)}=\cG_{\cY}^{(h)}$, meaning that $b^{(h)}$ is one of the $\leq\Delta^{r-1}$ basis elements of the local coefficient space $\cG^{(h)}_{e^{(h)}}$ for some edge $e^{(h)}\in Y(1)$. Then for $(b^{(i)}\in C^{(i)}(1))_{i\in[r]\setminus\{h\}}$, we can only have $\zeta(b^{(1)},\dots,b^{(r)})\neq 0$ if each $b^{(i)}$ for $i\neq h$ corresponds to one of the $\leq\Delta^{r-1}$ basis elements of $\cG^{(i)}_{e^{(i)}}$ for some edges $(e^{(i)}\in Y(1))_{i\in[r]\setminus\{h\}}$ that satisfy $e^{(1)},\dots,e^{(r)}\prec y$ for some $y\in Y(r)$. As there are $\Delta^{r-1}$ faces $y\in Y(r)$ with $y\succ e^{(h)}$, and there are $r\cdot 2^{r-1}$ edges $e\prec y$, it follows that there are $\leq\Delta^{r-1}\cdot(r\cdot 2^{r-1}\cdot\Delta^{r-1})^{r-1}\leq q^{O_r(1)}$ choices of $(b^{(i)}\in C^{(i)}(1))_{i\neq h}$ that could possibly yield $\zeta(b^{(1)},\dots,b^{(r)})\neq 0$, where here we recall that $\Delta\leq q$. Hence $w^\zeta\leq q^{O_r(1)}$. Meanwhile, as observed previously, (\ref{eq:codelenbound}) implies that $q^{O_r(1)}\leq\log(N)^{O_{\nu,\delta,r}(1)}$, which yields the desired bound on $w^\zeta$ in the lemma statement.
\end{proof}



\subsection{Definition of 1-Cohomology Subspaces}
\label{sec:Hpdef}
In this section, we define the $r$-tuple $H'=({H^{(h)}}'\subseteq H^1(\cC^{(h)}))_{h\in[r]}$ of 1-cohomology subspaces in Theorem~\ref{thm:qldpcmain}. In the subsequent sections, we will show that $\zeta$ is coboundary-invariant on $H'$, and we will bound the subrank of $\zeta_{H'}$.

Define a monomial $M\in\bF_q[U_0,\dots,U_t]$ by
\begin{equation}
  \label{eq:Mdef}
  M(U_0,\dots,U_t) = U_0^{|E|-1}U_1^{q-\lfloor\ell/2t\rfloor}\cdots U_t^{q-\lfloor\ell/2t\rfloor}.
\end{equation}
Recalling from~(\ref{eq:adef}) that $a=\lfloor\ell/10rt\rfloor$, we then define sets of monomials $L^0,L^1\subseteq\bF_q[U_0,\dots,U_t]$ of size $|L|=|L'|=a^t$ by
\begin{align*}
  L^0 &= \{U_1^{j_1}\cdots U_t^{j_t}:0\leq j_i\leq a-1\;\forall i\in[t]\} \\
  L^1 &= \{M\cdot M':M'\in L^0\}.
\end{align*}
Define $\iota^0:\bF_q[U_0,\dots,U_t]^{<\ell}\rightarrow\cF_{\bar{\cX}}^0$ as in Lemma~\ref{lem:RMplant} (with the variables $U_i$ replacing $X_i$), and define $\iota^1:\bF_q[U_0,\dots,U_t]\rightarrow\cF_{\bar{\cX}}^1$ by $\iota^1(f)=\evl_{E\times\bF_q^t}(f)$. For $h\in[r]$, we then define
\begin{align*}
  {H^{(h)}}'
  &= \spn\left\{\bigotimes_{i\in[r]}\iota^{\1_{h=i}}(f_i)+B^1(\cG_{\cY}^{(h)}):f_i\in L^{\1_{h=i}}\;\forall i\in[r]\right\}.
\end{align*}
For $h\in[r]$, to see that ${H^{(h)}}'$ is a well-defined subspace of $H^1(\cC^{(h)})$, it suffices to show that $\bigotimes_{i\in[r]}\iota^{\1_{h=i}}(f_i)\in Z^1(\cG_{\cY}^{(h)})$ for every $(f_i\in L^{\1_{h=i}})_{i\in[r]}$. For this purpose, for every $i\neq h$, then $f_i$ is a monomial of total degree $<ta<\ell$, so item~\ref{it:im} of Lemma~\ref{lem:RMplant} implies that $\iota^0(f_i)\in Z^0(\cF_{\bar{\cX}}')$. Meanwhile, for $i=h$ by definition $\iota^1(f_h)\in Z^1(\cF_{\bar{\cX}})$, as $Z^1(\cF_{\bar{\cX}})=\cF_{\bar{\cX}}^1$ because $\cF_{\bar{\cX}}$ is a 1-dimensional complex. Therefore $\bigotimes_{i\in[r]}\iota^{\1_{h=i}}(f_i)$ must lie in the space $Z^1(\cG_{\cY}^{(h)})$ of 1-cochains of the product complex, as desired; this conclusion follows directly from the definition of tensor products of chain complexes, and can alternatively be seen from the K\"{u}nneth formula (Lemma~\ref{lem:kunneth}).

\subsection{Proof of Coboundary-Invariance}
We now show that $\zeta$ is coboundary-invariant on $H'$:

\begin{lemma}
  \label{lem:cobinv}
  The multilinear form $\zeta$ defined in Section~\ref{sec:zetadef} is coboundary-invariant on the tuple $H'$ defined in Section~\ref{sec:Hpdef}.
\end{lemma}

The following lemma is the key result for proving Lemma~\ref{lem:cobinv}

\begin{lemma}
  \label{lem:xilowdeg}
  Let $S\subseteq[r]$ be an arbitrary nonempty subset. Let $(g^{(1)},\dots,g^{(r)})\in(\cG_{\cY}^{(1)})^1\times\cdots\times(\cG_{\cY}^{(r)})^1$ be a tuple satisfying the following:
  \begin{enumerate}
  \item For every $h\in S$, there exists some $f^{(h)}\in(\cG_{\cY}^{(h)})^0$ supported on a single 0-dimensional face $v_h\in Y(0)$ such that $g^{(h)}=\delta_0^{\cG_{\cY}^{(h)}}(f^{(h)})$.
  \item For every $h\in[r]\setminus S$, there exists some $(f^{(h)}_i\in L^{\1_{h=i}})_{i\in[r]}$ such that $g^{(h)}=\bigotimes_{i\in[r]}\iota^{\1_{h=i}}(f^{(h)}_i)$.
  \end{enumerate}
  Then there exists some $h^*\in S$ such that
  \begin{equation}
    \label{eq:xiinprod}
    \xi(g^{(1)},\dots,g^{(r)}) \in (\bF_q^{E\times\bF_q})^{\otimes(h^*-1)}\otimes\evl_{E\times\bF_q^t}(\bF_q[U_0,\dots,U_t]^{<t(q-1)+|E|-\ell+(r-1)\ell'})\otimes(\bF_q^{E\times\bF_q})^{\otimes(r-h^*)}.
  \end{equation}
\end{lemma}

\begin{remark}
  \label{remark:equivdegbound}
  An equivalent way to state the condition~(\ref{eq:xiinprod}) is that there is some polynomial in $\bF_q[(U_i^{(h)})_{i\in\{0,\dots,t\}}^{h\in[r]}]$ that agrees with $\xi(g^{(1)},\dots,g^{(r)})$ at all inputs in $Y(r)=(E\times\bF_q^t)^r$, and such that every monomial $\prod_{i,h}(U_i^{(h)})^{j_i^{(h)}}$ with nonzero coefficient has $\sum_{i=0}^tj_i^{(h^*)}<t(q-1)+|E|-\ell+(r-1)\ell'$.
\end{remark}

We first prove coboundary-invariance assuming Lemma~\ref{lem:xilowdeg}, and then we subsequently prove Lemma~\ref{lem:xilowdeg}.

\begin{proof}[Proof of Lemma~\ref{lem:cobinv} assuming Lemma~\ref{lem:xilowdeg}]
  By the multilinearity of $\zeta$, it suffices to show that for every choice of $(f_i^{(h)}\in L^{\1_{h=i}})_{h,i\in[r]}$ and $(b^{(h)}\in B^1(\cG_{\cY}^{(h)}))_{h\in[r]}$ then letting $z^{(h)}=\bigotimes_{i\in[r]}\iota^{\1_{h=i}}(f_i^{(h)})$, we have
  \begin{equation}
    \label{eq:cobinvgoal}
    \zeta(z^{(1)},\dots,z^{(r)}) = \zeta(z^{(1)}+b^{(1)},\dots,z^{(r)}+b^{(r)}).
  \end{equation}
  Indeed, such $z^{(h)}$ of this form by definition comprise a generating set of ${H^{(h)}}'$, so~(\ref{eq:cobinvgoal}) immediately implies the sufficient condition for coboundary invariance described in Definition~\ref{def:cobinv}. But we may express each $b^{(h)}\in B^1(\cG_{\cY}^{(h)})$ as a sum of coboundaries of the form $\delta_0^{\cG_{\cY}^{(h)}}f^{(h)}$ for some $f^{(h)}\in(\cG_{\cY})^0$ supported on a single vertex $v_h\in Y(0)$. Then by multilinearity, the RHS of~(\ref{eq:cobinvgoal}) equals the LHS $\zeta(z^{(1)},\dots,z^{(r)})$ plus a sum of terms of the form $\zeta(g^{(1)},\dots,g^{(r)})$ for $g^{(1)},\dots,g^{(r)}$ satisfying the conditions in the statement of Lemma~\ref{lem:xilowdeg}; namely, every $g^{(h)}$ either equals some $z^{(h)}$ or some $\delta_0^{\cG_{\cY}^{(h)}}f^{(h)}$ for some $f^{(h)}\in(\cG_{\cY})^0$ supported on a single vertex $v_h\in Y(0)$, with the latter being the case for at least one $h\in[r]$. Thus to proven~(\ref{eq:cobinvgoal}), it suffices to show that every $g^{(1)},\dots,g^{(r)}$ satisfying the conditions in Lemma~\ref{lem:xilowdeg} has $\zeta(g^{(1)},\dots,g^{(r)})=0$. But~(\ref{eq:zetadef}) along with the definition of $\alpha:\bF_q^{Y(r)}\rightarrow\bF_q$ implies that $\zeta$ must vanish on every $g^{(1)},\dots,g^{(r)}$ for which~(\ref{eq:xiinprod}) holds. Hence Lemma~\ref{lem:xilowdeg} implies~(\ref{eq:cobinvgoal}), so $\zeta$ is coboundary-invariant on $H'$, as desired.
\end{proof}
  
\begin{proof}[Proof of Lemma~\ref{lem:xilowdeg}]
  If $\xi(g^{(1)},\dots,g^{(r)})=0$, the result holds trivially. Therefore we will fix $g^{(1)},\dots,g^{(r)}$ such that $\xi(g^{(1)},\dots,g^{(r)})\neq 0$. Under this assumption, we begin by showing the following claim, which characterizes the structure of the vertices $v_h$ for $h\in S$, and provides a simpler expression for $\xi(g^{(1)},\dots,g^{(r)})$.

  Below, we say that elements $h,h'\in S$ are \textbf{consecutive} if there does not exist $h''\in S$ with $h<h''<h'$. Also, we let $Y^{\loc,S} $ denote the $|S|$-dimensional boolean hypercube, and we let $\Pi_S:Y^{\loc}\rightarrow Y^{\loc,S}$ denote the projection given by $\Pi_s(y^{\loc})=y^{\loc}|_S$, so that $\Pi_S$ collapses all directions in $[r]\setminus S$. 

  \begin{claim}
    \label{claim:pathstructure}
    Assume that $\xi(g^{(1)},\dots,g^{(r)})\neq 0$. Then for every $y\in Y(r)$ with $\xi(g^{(1)},\dots,g^{(r)})_y\neq 0$, it holds that $v_h\preceq y$ for every $h\in S$. Also, for every consecutive $h,h'\in S$, the Hamming distance $|\Pi_S(T(v_{h'}))-\Pi_S(T(v_h))|\leq 1$. Furthermore, there exists some permutation $\pi\in\cS_r$ depending only on $(v_h)_{h\in S}$ such that $\pi(h)=h\;\forall h\notin S$, and such that for every $y\in Y(r)$,
    \begin{equation*}
      \xi(g^{(1)},\dots,g^{(r)})_y = \sgn(\pi) \cdot (g^{(1)}|_y)_{e^1_y(\pi)}\cdots (g^{(r)}|_y)_{e^r_y(\pi)}.
    \end{equation*}
    Specifically, $\pi\in\cS_r$ is the unique permutation with $\pi(h)=h\;\forall h\notin S$ such that $T(v_h)\prec e^h(\pi)$ for every $h\in S$.
  \end{claim}
  \begin{proof}
    For $y\in Y(r)$ with $\xi(g^{(1)},\dots,g^{(r)})_y\neq 0$, by definition
    \begin{align}
      \label{eq:xilocydef}
      \begin{split}
        \xi(g^{(1)},\dots,g^{(r)})_y
        &= \xi^{\loc}_y(g^{(1)}|_y,\dots,g^{(r)}|_y) \\
        &= \xi^{\loc}(T(g^{(1)}|_y),\dots,T(g^{(r)}|_y)) \\
        &= \sum_{\pi\in\cS_r}\sgn(\pi) \cdot (g^{(1)}|_y)_{e^1_y(\pi)}\cdots (g^{(r)}|_y)_{e^r_y(\pi)}.
      \end{split}
    \end{align}
    For $h\in S$, as $g^{(h)}$ is by definition supported on edges containing the vertex $v_h$, then by the assumption that the RHS of~(\ref{eq:xilocydef}) is nonzero, we have for some $\pi\in\cS_r$ that $v_h\prec e_y^h(\pi)\preceq y$. Thus we have shown the first statement in the claim, namely that if $\xi(g^{(1)},\dots,g^{(r)})_y\neq 0$ then every $h\in S$ has $v_h\preceq y$.


    For $h\in[r]\setminus S$, by definition $g^{(h)}$ is supported on edges $e\in Y(1)$ pointing in direction $\direc(e)=h$. Therefore if $\pi(h)\neq h$ for $h\notin S$, the RHS of~(\ref{eq:xilocydef}) vanishes, so we can write
    \begin{align}
      \label{eq:xigoodpi}
      \begin{split}
        \xi(g^{(1)},\dots,g^{(r)})_y
        &= \sum_{\pi\in\cS_r:\pi(h)=h\;\forall h\notin S}\sgn(\pi) \cdot (g^{(1)}|_y)_{e^1_y(\pi)}\cdots (g^{(r)}|_y)_{e^r_y(\pi)}.
      \end{split}
    \end{align}
    


    Assume that $y\in Y(r)$ is such that the sum on the RHS of~(\ref{eq:xigoodpi}) is nonzero, and consider some $\pi\in\cS_r$ with $\pi(h)=h\;\forall h\notin S$ whose associated term in this sum is nonzero. As observed previously, by the definition of $g^{(h)}$ for $h\in S$, it follows that for every $h\in S$, we have $v_h\prec e_y^h(\pi)$. Then $\Pi_S(e^h(\pi))_{h\in S}$ forms a length-$|S|$ path from $0^S$ to $1^S$ in $Y^{\loc,S}$, for which each edge $\Pi_S(e^h(\pi))$ contains vertex $\Pi_S(T(v_h))$. Therefore if $h,h'\in S$ are consecutive elements (meaning $\nexists h''\in S$ with $h<h''<h'$) , then we must have the Hamming distance $|\Pi_S(T(v_{h'}))-\Pi_S(T(v_h))|\leq 2$, as these two vertices lie in consecutive edges in a path.
    
    However, if $|\Pi_S(h')-\Pi_S(h)|=2$, then $\xi(g^{(1)},\dots,g^{(r)})_y=0$. To see that this statement holds, assume for a contradiction that $\xi(g^{(1)},\dots,g^{(r)})_y\neq 0$ and that $|\Pi_S(h')-\Pi_S(h)|=2$ for consecutive $h,h'\in S$. Consider the involution $W=W_{h,h'}$ acting on permutations $\pi\in\cS_r:\pi(h)=h\;\forall h\notin S$ that simply swaps the values of $\pi(h)$ and $\pi(h')$. By definition $\sgn(\pi)=-\sgn(W(\pi))$. Now for every $i\in S$, it holds that
    \begin{equation}
      \label{eq:pathssame}
      (g^{(i)}|_{y})_{e_y^i(\pi)} = (\delta_0^{\cY})_{e_y^i(\pi),v_i} \cdot \cG^{(i)}_{y\leftarrow v_i}(f^{(i)}_{v_i}) = (\delta_0^{\cY})_{e_y^i(W(\pi)),v_i} \cdot \cG^{(i)}_{y\leftarrow v_i}(f^{(i)}_{v_i}) = (g^{(i)}|_{y})_{e_y^i(W(\pi))}
    \end{equation}
    Indeed, the paths $e_y^1(\pi),\dots,e_y^r(\pi)$ and $e_y^1(W(\pi)),\dots,e_y^r(W(\pi))$ agree except for the segment from edge $h$ through edge $h'$, where they take different subpaths from vertex $v_h$ to vertex $v_{h'}$. Hence for every $i\in S\setminus\{h,h'\}$, we have $e_y^i(\pi)=e_y^i(W(\pi))$, and~(\ref{eq:pathssame}) holds. For $i\in\{h,h'\}$, (\ref{eq:pathssame}) holds because $(\delta_0^{\cY})_{e_y^h(\pi),v_h}=(\delta_0^{\cY})_{e_y^h(W(\pi)),v_h}=+1$ and $(\delta_0^{\cY})_{e_y^{h'}(\pi),v_{h'}}=(\delta_0^{\cY})_{e_y^{h'}(W(\pi)),v_{h'}}=-1$, as by Definition~\ref{def:graphinc} and Definition~\ref{def:product}, for $e\in Y(1)$ and $v\in Y(0)$ with $v\prec e$, then $(\delta_0^{\cY})_{e,v}=(-1)^{T(v)_{\direc(e)}}$. Projecting down to the (2-dimensional) square in directions $S'=\{\pi(h),\pi(h')\}$, then $(\Pi_{S'}(e^h(\pi)),\Pi_{S'}(e^{h'}(\pi)))$ and $(\Pi_{S'}(e^h(W(\pi))),\Pi_{S'}(e^{h'}(W(\pi))))$ are the two length-$2$ paths from $00=\Pi_{S'}(T(v_h))$ to $11=\Pi_{S'}(T(v_{h'}))$ in the square $Y^{\loc,S'}$. Hence $T(v_h)_{\direc(e^h(\pi))}=T(v_h)_{\direc(e^h(W(\pi)))}=0$ and $T(v_{h'})_{\direc(e^{h'}(\pi))}=T(v_{h'})_{\direc(e^{h'}(W(\pi)))}=1$, so~(\ref{eq:pathssame}) holds for $i\in\{h,h'\}$. Similarly, for every $i\in[r]\setminus S$, writing $y=y_1\times\cdots\times y_r\in(E\times\bF_q^t)^r=Y(r)$ so that each $y_i\in E\times\bF_q^t=\bar{X}(1)$, we have
    \begin{equation}
      \label{eq:logsame}
      (g^{(i)}|_{y})_{e_y^i(\pi)} = f_1^{(i)}(y_1)\cdots f_r^{(i)}(y_r) = (g^{(i)}|_{y})_{e_y^i(W(\pi))},
    \end{equation}
    where in the middle expression of~(\ref{eq:logsame}) we recall that each $f_j^{(i)}\in\bF_q[U_0,\dots,U_t]$. Specifically, the equalities in~(\ref{eq:logsame}) hold because by definition $g^{(i)}=\bigotimes_{j\in[r]}\iota^{\1_{i=j}}(f_j^{(i)})$, so for every direction-$i$ edge $e=x_1\times\cdots x_r\in Y(1)$, $e\preceq y$ (so each $x_j\in \bar{X}(\1_{i=j})$), then the definition of $\iota^0,\iota^1$ ensures that
    \begin{equation}
      \label{eq:logform}
      (g^{(i)}|_{y})_e = \cG^{(i)}_{y\leftarrow e}g^{(i)}_e = \iota^1(f_i^{(i)})_{y_i}\cdot\prod_{j\in[r]\setminus\{i\}}\cF'_{y_j\leftarrow x_j}(\iota^0(f_j^{(i)})) = f_i^{(i)}(y_i)\cdot\prod_{j\in[r]\setminus\{i\}}f_j^{(i)}(y_j).
    \end{equation}
    Setting $e=e_y^i(\pi)$ and $e=e_y^i(W(\pi))$ in the above equation yields~(\ref{eq:logsame}). Thus we have shown that for consecutive $h,h'\in S$ with $|\Pi_S(T(v_h'))-\Pi_S(T(v_h))|=2$, then for every $i\in[r]$,
    \begin{equation*}
      (g^{(i)}|_{y})_{e_y^i(\pi)} = (g^{(i)}|_{y})_{e_y^i(W(\pi))}.
    \end{equation*}
    Hence because $W$ is an involution on $\{\pi\in\cS_r:\pi(i)=i\;\forall i\notin S\}$ with $\sgn(W(\pi))=-\sgn(\pi)$, it follows from~(\ref{eq:xigoodpi}) that
    \begin{align*}
      \hspace{1em}&\hspace{-1em}\xi(g^{(1)},\dots,g^{(r)})_y \\
                  &= \sum_{\pi\in \cS_r:\pi(i)=i\;\forall i\notin S,\; \sgn(\pi)=+1} ((g^{(1)}|_{y})_{e_y^1(\pi)}\cdots (g^{(r)}|_{y})_{e_y^r(\pi)} - (g^{(1)}|_{y})_{e_y^1(W(\pi))}\cdots (g^{(r)}|_{y})_{e_y^r(W(\pi))}) \\
                  &= 0,
    \end{align*}
    contradicting the assumption that $\xi(g^{(1)},\dots,g^{(r)})_y\neq 0$. Thus indeed for every consecutive $h,h'\in S$, we must have $|\Pi_S(T(v_{h'}))-\Pi_S(T(v_h))|\leq 1$.

    Now assume for a contradiction that there exist two distinct permutations $\pi,\pi':\pi(i)=\pi'(i)=i\;\forall i\notin S$ for which $T(v_i)\prec e^i(\pi),e^i(\pi')\;\forall i\in S$ (which, as shown previously, is a necessary condition for having the terms associated to $\pi,\pi'$ on the RHS of~(\ref{eq:xigoodpi}) to be nonzero). Letting $h\in S$ be the least element such that $\pi(h)\neq\pi(h')$, then applying the projection $\Pi_S$, we see that the paths $(\Pi_S(e^i(\pi)))_{i\in S}$ and $(\Pi_S(e^i(\pi')))_{i\in S}$ from $0^S$ to $1^S$ first diverge at the unique shared vertex of $\Pi_S(e^h(\pi))$ and $\Pi_S(e^h(\pi'))$. Therefore this shared vertex must be $\Pi_S(T(v_h))$. Letting $h'\in S$ be the least element greater than $h$, so that $h,h'\in S$ are consecutive, then because $T(v_{h'})\prec e^{h'}(\pi),e^{h'}(\pi')$, the paths must converge at vertex $\Pi_S(T(v_{h'}))$, which must be the unique shared vertex of edges $\Pi_S(e^{h'}(\pi))$ and $\Pi_S(e^{h'}(\pi'))$. But then $|\Pi_S(T(v_{h'}))-\Pi_S(T(v_h))|=2$, which contradicts the bound $|\Pi_S(T(v_{h'}))-\Pi_S(T(v_h))|\leq 1$ we showed above. Thus there is a unique $\pi\in\cS_r$ with $\pi(i)=i\;\forall i\notin S$ such that $T(v_i)\prec e^i(\pi)\;\forall i\in S$. This $\pi$ also gives the unique term with a nonzero contribution to the sum in~(\ref{eq:xigoodpi}), so
    \begin{align*}
      \xi(g^{(1)},\dots,g^{(r)})_y
      &= \sgn(\pi) \cdot (g^{(1)}|_y)_{e^1_y(\pi)}\cdots (g^{(r)}|_y)_{e^r_y(\pi)},
    \end{align*}
    as desired.
  \end{proof}

  The following claim chooses the desired element $h^*\in S$ described in the lemma statement.

  \begin{claim}
    \label{claim:Yprime}
    Assume that $\xi(g^{(1)},\dots,g^{(r)})\neq 0$. Then there exists some $h^*\in S$ and some ${y^{\loc}}'\in Y^{\loc}(r-1)$ with $b:=T({y^{\loc}}')_{h^*}\in\{0,1\}$ (so that $T({y^{\loc}}')=*^{h^*-1}b*^{r-h^*}$) such that ${y^{\loc}}'\succ T(v_h)$ for every $h\in S$. Furthermore, the set $Y'\subseteq Y(r-1)$ defined by
    \begin{equation*}
      Y' = \{y'\in Y(r-1):T(y')={y^{\loc}}',\; y'\succeq v_h\;\forall h\in S\}.
    \end{equation*}
    is nonempty, and every $y\in Y(r)$ with $y\succ v_h\;\forall h\in S$ lies above a unique element $y'\in Y'$, which is given by $y'=T|_{Y_{\preceq y}}^{-1}({y^{\loc}}')$.
  \end{claim}
  \begin{proof}
    By Claim~\ref{claim:pathstructure}, $(\Pi_S(T(v_h)))_{h\in S}$ is a sequence of $|S|$ points in the $|S|$-dimensional boolean hypercube $Y^{\loc,S}$, such that each consecutive pair of points has Hamming distance $\leq 1$. Furthermore, letting $\pi$ be the permutation given by Claim~\ref{claim:pathstructure}, then $\Pi_S(T(v_h))$ lies on the edge $\Pi_S(e^h(\pi))$ of the length-$|S|$ path $(T(e^h(\pi)))_{h\in S}$ from $0^S$ to $1^S$. 


    Letting $\underline{h}=\min S$ and $\overline{h}=\max S$, it follows by the triangle inequality that $|\Pi_S(v_{\overline{h}})-\Pi_S(v_{\underline{h}})|\leq|S|-1$. Therefore because $|1^S-0^S|=|S|$, it follows that either $\Pi_S(T(v_{\underline{h}}))\neq 0^S$ or that $\Pi_S(T(v_{\overline{h}}))\neq 1^S$. Let
    \begin{equation*}
      h^* = \begin{cases}
        \pi(\underline{h})=\direc(e^{\underline{h}}(\pi)),&\Pi_S(v_{\underline{h}})\neq 0^S \\
        \pi(\overline{h})=\direc(e^{\overline{h}}(\pi),&\Pi_S(v_{\underline{h}})=0^S.
      \end{cases}
    \end{equation*}
    Then $h^*\in S$, every $T(v_h)$ for $h\in S$ lies inside some $(r-1)$-dimensional face ${y^{\loc}}'\in Y^{\loc}(r-1)$ with $T({y^{\loc}}')_{h^*}\in\{0,1\}$. Indeed, if $h^*=\underline{h}$, then every $v_h$ for $h\in S$ has $T(v_h)_{h^*}=1$, so we can let ${y^{\loc}}'=*^{h^*-1}1*^{r-h^*}$. If instead $h^*=\overline{h}$, then every $v_h$ for $h\in S$ has $T(v_h)_{h^*}=0$, so we can let ${y^{\loc}}'=*^{h^*-1}0*^{r-h^*}$.

    By assumption there exists some $y\in Y(r)$ with $\xi(g^{(1)},\dots,g^{(r)})_y\neq 0$, which by Claim~\ref{claim:pathstructure} implies that $y\succ v_h$ for every $h\in S$. Therefore $T|_{Y_{\preceq y}}^{-1}({y^{\loc}}')\in Y'$, so $Y'$ is nonempty. Furthermore, for distinct $y',y''\in Y'$, there cannot exist $y\in Y(r)$ such that $y\succ y',y''$, as then we would have the isomorphism $T|_{Y_{\preceq y}}(y')={y^{\loc}}'=T|_{Y_{\preceq y}}(y'')$, so applying $T|_{Y_{\preceq y}}^{-1}$ gives $y'=y''$.
  \end{proof}

  Fix $\pi$ to be the permutation given by Claim~\ref{claim:pathstructure}. Then by Claim~\ref{claim:pathstructure} and Claim~\ref{claim:Yprime},
  \begin{align}
    \label{eq:isolateinner}
    \begin{split}
      \xi(g^{(1)},\dots,g^{(r)})
      &= \sum_{y'\in Y'}\sum_{y\triangleright y'} \ind{y}\cdot\xi(g^{(1)},\dots,g^{(r)})_y \\
      &= \sgn(\pi) \cdot \sum_{y'\in Y'}\sum_{y\triangleright y'} \ind{y}\cdot(g^{(1)}|_{y})_{e_y^1(\pi)}\cdots (g^{(r)}|_{y})_{e_y^r(\pi)}.
    \end{split}
  \end{align}
  By definition, for $y'\in Y'$ and $y\triangleright y'$,
  \begin{align*}
    (g^{(1)}|_{y})_{e_y^1(\pi)}\cdots (g^{(r)}|_{y})_{e_y^r(\pi)}
    &= \prod_{h\in[r]} \cG^{(h)}_{y\leftarrow e_y^h(\pi)}(g^{(h)}_{e_y^h(\pi)}) \\
    &= \prod_{h\in S} \left((\delta^{\cY}_0)_{e_y^h(\pi),v_h}\cdot\cG^{(h)}_{y\leftarrow v_h}(f^{(h)}_{v_h})\right) \cdot \prod_{h\in[r]\setminus S} \left(f^{(h)}_1(y_1)\cdots f^{(h)}_r(y_r)\right),
  \end{align*}
  where we write $y=y_1\times\cdots\times y_r\in(E\times\bF_q^t)^{\times r}=\bar{E}^{\times r}=\bar{X}(1)^{\times r}=Y(r)$. Note that the second equality above holds by the definition of the $g^{(h)}$ along with~(\ref{eq:logform}). Observe that
  \begin{equation*}
    \prod_{h\in S}(\delta^{\cY}_0)_{e_y^h(\pi),v_h} = \prod_{h\in S}(\delta^{\cY^{\loc}}_0)_{e^h(\pi),T(v_h)} \in \{\pm 1\}
  \end{equation*}
  does not depend on $y$. Therefore letting $s\in\{\pm 1\}$ be the value above, then the innermost sum on the RHS of~(\ref{eq:isolateinner}) becomes
  \begin{align}
    \label{eq:pulloutincidence}
    \sum_{y\triangleright y'}\ind{y}\cdot(g^{(1)}|_{y})_{e_y^1(\pi)}\cdots (g^{(r)}|_{y})_{e_y^r(\pi)}
    &= s\cdot\sum_{y\triangleright y'}\ind{y}\cdot\prod_{h\in S} \cG^{(h)}_{y\leftarrow v_h}(f^{(h)}_{v_h}) \cdot \prod_{h\in[r]\setminus S} \left(f^{(h)}_1(y_1)\cdots f^{(h)}_r(y_r)\right),
  \end{align}
  where above we used~(\ref{eq:logform}) to express $(g^{(h)}_y)_{e^h_y(\pi)}=f^{(h)}_1(y_1)\cdots f^{(h)}_r(y_r)$ for $h\in[r]\setminus S$.
  Now we may write each $v_h=v_{h,1}\times\cdots\times v_{h,r}\in\bar{V}^{\times r}=\bar{X}(0)^{\times r}=Y(0)$, and for $y'\in Y'$ we write $y'=y_1'\times\cdots\times y_r'\in\bar{E}^{\times h^*-1}\times\bar{V}\times\bar{E}^{\times r-h^*}$. Because $y'\succeq v_h$ for each $h\in S$, we must have $y'_{h^*}=v_{h,h^*}$ for very $h\in S$. Thus the $\Delta$ different $r$-dimensional faces $y\triangleright y'$ all have $y_h=y'_h\;\forall h\neq h^*$, and then have any of the $\Delta$ choices of $y_{h^*}\in\bar{E}(y_{h^*}')$.

  Let $\Pi_0$ denote the natural projection from $\bar{\Gamma}$ to $\Gamma$, so that $\Pi_0:\bar{E}=E\times\bF_q^t\rightarrow E$ denotes projection onto the first coordinate, and $\Pi_0:\bar{V}=\bigsqcup_{v\in V}\bF_q^{t+1}/\spn\{(1,\labV(v))\}\rightarrow V$ maps cosets of $\spn\{(1,\labV(v))\}$ to $v$. Then by definition, for every $h\in S$, we have\footnote{Recall that for vertices $v$, here the variables $h_v$ and $h'_v$ denote the matrices from Definition~\ref{def:RMplant} used to define local codes within $\cF$ and $\cF'$ respectively, and in particular are distinct from the variables $h,h^*\in[r]$.}
  \begin{align*}
    \cG^{(h)}_{y\leftarrow v_h}(f^{(h)}_{v_h})
    &= \left(\bigotimes_{i=1}^{h-1}\cF'_{y_i\leftarrow v_{h,i}} \otimes \cF_{y_{h}\leftarrow v_{h,h}} \otimes \bigotimes_{i=h+1}^{r}\cF'_{y_i\leftarrow v_{h,i}}\right)(f^{(h)}_{v_{h}}) \\
    &= \ind{y}^\top\left(\bigotimes_{i=1}^{h-1}{h'_{v_{h,i}}}^\top \otimes h_{v_{h,h}}^\top \otimes \bigotimes_{i=h+1}^{r}{h'_{v_{h,i}}}^\top\right)(f^{(h)}_{v_{h}}) \\
    &= \ind{\Pi_0^{\times r}(y)}^\top\left(\bigotimes_{i=1}^{h-1}{h'_{\Pi_0(v_{h,i})}}^\top \otimes h_{\Pi_0(v_{h,h})}^\top \otimes \bigotimes_{i=h+1}^{r}{h'_{\Pi_0(v_{h,i})}}^\top\right)(f^{(h)}_{v_{h}}).
  \end{align*}
  Now by definition, for every $v\in V$, we have $\im(h_{v}^\top)=\evl_{E(v)}(\bF_q[U_0]^{<|E|-\ell}_{E\setminus E(v)})$, while $\im({h'_{v}}^\top)=\evl_{E(v)}(\bF_q[U_0]^{<\ell'})$. Therefore the vector
  \begin{align*}
    c^{(h)}
    &:= \left(\bigotimes_{i=1}^{h-1}{h'_{\Pi_0(v_{h,i})}}^\top \otimes h_{\Pi_0(v_{h,h})}^\top \otimes \bigotimes_{i=h+1}^{r}{h'_{\Pi_0(v_{h,i})}}^\top\right)(f^{(h)}_{v_{h}}) \\
    &\in \bigotimes_{i=1}^{h-1}\evl_{E(\Pi_0(v_{h,i}))}(\bF_q[U_0]^{<\ell'}) \otimes \evl_{E(\Pi_0(v_{h,h}))}(\bF_q[U_0]^{<|E|-\ell}_{E\setminus E(\Pi_0(v_{h,h}))}) \otimes \bigotimes_{i=h+1}^{r}\evl_{E(\Pi_0(v_{h,i}))}(\bF_q[U_0]^{<\ell'}).
  \end{align*}
  satisfies $c^{(h)}_{\Pi_0^{\times r}(y)}=\cG^{(h)}_{y\leftarrow v_h}(f^{(h)}_{v_h})$ for every $y\triangleright y'$. Therefore for some $n\in\bN$, there exist polynomials $(f^{(h)}_{i,j}\in\bF_q[U_0])_{i\in[r],j\in[n]}$ such that $f^{(h)}_{h,j}\in\bF_q[U_0]^{<|E|-\ell}_{E\setminus E(\Pi_0(v_{h,h}))}$ and $f^{(h)}_{i,j}\in\bF_q[U_0]^{<\ell'}$ for $i\neq h$, and such that for every $y\triangleright y'$ we have
  \begin{equation}
    \label{eq:tensortopoly}
    \cG^{(h)}_{y\leftarrow v_h}(f^{(h)}_{v_h}) = c^{(h)}_{\Pi_0^{\times r}(y)} = \sum_{j\in[n]}f^{(h)}_{1,j}(\Pi_0(y_1))\cdots f^{(h)}_{r,j}(\Pi_0(y_r)).
  \end{equation}
  
  Also for $h\neq h^*$ let $I^{(h)}_{y'}\in\bF_q[U_0,\dots,U_t]$ be some polynomial (of arbitrary degree) with $I^{(h)}_{y'}(u)=\1_{u=y'_h}$, and let $I^{(h^*)}_{y'}\in\bF_q[U_0,\dots,U_t]^{\leq t(q-1)}$ be the polynomial with $I^{(h^*)}_{y'}(u)=\1_{u\in y'_{h^*}}$. Specifically, recalling that $y'_{h^*}$ is an affine line in $\bF_q^{t+1}$ given by $y'_{h^*}=x'+\spn\{(1,\labV(v'))\}$ for some $x'\in\{0\}\times\bF_q^t$ and $v'\in V$, then
  \begin{equation*}
    I^{(h^*)}_{y'}(U_0,\dots,U_t) = \prod_{i\in[r]}(1-(U_i-\labV(v')_iU_0-x'_i)^{q-1}).
  \end{equation*}
  Then for every $y\in Y(r)=(E\times\bF_q^t)^r$, it follows that $\xi(g^{(1)},\dots,g^{(r)})_y$ equals the evaluation of the $r(t+1)$-variate polynomial $G(U^{(h)}_i)^{h\in[r]}_{i\in\{0,\dots,t\}}\in\bF_q[(U^{(h)}_i)^{h\in[r]}_{i\in\{0,\dots,t\}}]$ given by
  \begin{align}
    \label{eq:Gdef}
    \begin{split}
      G(U^{(h)}_i)^{h\in[r]}_{i\in\{0,\dots,t\}}
      \hspace{10em}&\hspace{-10em}= \sgn(\pi)\cdot s\cdot\sum_{y'\in Y'}\sum_{(j^{(h)})_{h\in S}\in[n]^S}
      I^{(1)}_{y'}(U^{(1)})\cdots I^{(r)}_{y'}(U^{(r)}) \\
      &\cdot\prod_{h\in S} \left(f^{(h)}_{1,j^{(h)}}(U_0^{(1)})\cdots f^{(h)}_{r,j^{(h)}}(U_0^{(r)})\right) \\
      &\cdot\prod_{h\in[r]\setminus S} \left(f^{(h)}_1(U^{(1)})\cdots f^{(h)}_r(U^{(r)})\right).
    \end{split}
  \end{align}
  at $U=(U^{(1)},\dots,U^{(r)})=(y_1,\dots,y_r)=y$, where we let $U^{(h)}=(U^{(h)}_0,\dots,U^{(h)}_t)$. Specifically, given $y'\in Y'$, then $I^{(1)}_{y'}(U^{(1)})\cdots I^{(r)}_{y'}(U^{(r)})$ is the indicator function for the affine line $y'$ in $(\bF_q^{t+1})^r$, and for every $j\in[n]$ by definition $f^{(h^*)}_{h^*,j}(U_0^{(h^*)})$ vanishes at every $U_0^{(h^*)}\in E\setminus E(\Pi_0(y'_{h^*}))$. Therefore the sum on the RHS of~(\ref{eq:Gdef}) can only have a nonzero term at $y'\in Y'$ for evaluation points $U=y$ such that $y\in y'$ and $\Pi_0(y_{h^*})\in E(\Pi_0(y'_{h^*}))$, or equivalently, such that $y\triangleright y'$. Thus~(\ref{eq:Gdef}) with $U=y$ is equivalent to
  \begin{align*}
    G(y)
    &= \sgn(\pi)\cdot s\cdot\sum_{y'\in Y'} \1_{y\triangleright y'} \cdot \sum_{(j^{(h)})_{h\in S}\in[n]^S} \prod_{h\in S} \left(f^{(h)}_{1,j^{(h)}}(\Pi_0(y_1))\cdots f^{(h)}_{r,j^{(h)}}(\Pi_0(y_r))\right) \\
    &\hspace{16em}\cdot \prod_{h\in[r]\setminus S} \left(f^{(h)}_1(y_1)\cdots f^{(h)}_r(y_r)\right).
  \end{align*}
  The fact that $\xi(g^{(1)},\dots,g^{(r)})_y=G(y)$ then follows directly by applying~(\ref{eq:isolateinner}), (\ref{eq:pulloutincidence}), and (\ref{eq:tensortopoly}) with the above equation.

  Now for a polynomial in $F\in\bF_q[(U^{(h)}_i)^{h\in[r]}_{i\in\{0,\dots,t\}}]$, let $\deg^{(h^*)}(F)$ denote the maximum value of $k_0^{(h^*)}+\cdots+k_t^{(h^*)}$ over all monomials $\prod_{i,h}(U^{(h)}_i)^{k^{(h)}_i}$ with a nonzero coefficient in $F$. Then~(\ref{eq:Gdef}), implies that
  \begin{align*}
    \deg^{(h^*)}(G)
    &\leq \max_{y'\in Y',\;(j^{(h)})_{h\in S}\in[n]^S} \deg(I^{(h^*)}_{y'}) + \sum_{h\in S}\deg(f^{(h)}_{h^*,j^{(h)}}) + \sum_{h\in[r]\setminus S}\deg(f^{(h)}_{h^*}) \\
    &\leq t(q-1) + (|E|-\ell-1) + (|S|-1)(\ell'-1) + (r-|S|-1)t(a-1) \\
    &< t(q-1) + |E| - \ell + (r-1)\ell',
  \end{align*}
  where the second inequality above holds because by definition $f^{(h^*)}_{h^*,j^{(h^*)}}\in\bF_q[U_0^{(h^*)}]^{<|E|-\ell}$, $f^{(h)}_{h^*,j^{(h)}}\in\bF_q[U_0^{(h^*)}]^{<\ell'}$ for $h\in S\setminus\{h^*\}$, and $f^{(h)}_{h^*}\in L^0$ for $h\in[r]\setminus S$; the third inequality above holds because by definition $a-1\leq\ell'/t$. Thus the desired bound~(\ref{eq:xiinprod}) follows by Remark~\ref{remark:equivdegbound}.
\end{proof}

\subsection{Proof of Subrank Bound}
\label{sec:subrankproof}
In this section, we bound the subrank of $\zeta_{H'}$ for $\zeta$ defined in Section~\ref{sec:zetadef} and $H'$ defined in Section~\ref{sec:Hpdef}. Recall below the definitions of $a$ and $A$ from~(\ref{eq:adef}) and~(\ref{eq:Adef}), respectively.
  
\begin{lemma}
  \label{lem:subrankbound}
  $\subrank(\zeta_{H'})\geq|A|^r=a^{rt}\geq N^{\delta-\nu}$.
\end{lemma}
\begin{proof}
  For $h\in[r]$, we define a linear map
  \begin{equation*}
    \phi^{(h)}:\bF_q^{A^r}=(\bF_q^A)^{\otimes r}\rightarrow{H^{(h)}}'
  \end{equation*}
  as follows. For $j\in\{0,\dots,t\}$, the map $\evl_{A_j}:\bF_q[U_j]^{<|A_j|}\xrightarrow{\sim}\bF_q^{A_j}$ is an isomorphism, so letting
  \begin{equation*}
    \cA = \spn\{L^0\}^{\otimes r} = \bigotimes_{i=1}^r\bigotimes_{j=0}^t\bF_q[U_j^{(i)}]^{<|A_j|} \subseteq \bF_q[(U_j^{(i)})_{j\in\{0,\dots,t\}}^{i\in[r]}]
  \end{equation*}
  denote the span of all $r(t+1)$-variate monomials whose degree in each $U_j^{(i)}$ is $<|A_j|$, then
  \begin{equation*}
    \evl_{A^r}:\cA\xrightarrow{\sim}\bF_q^{A^r}
  \end{equation*}
  is an isomorphism. Therefore given $c\in\bF_q^{A^r}$, recalling the definition of $M(U_0,\dots,U_t)$ from~(\ref{eq:Mdef}), we define $c/M^{(h)}\in\bF_q^{A_r}$ by $(c/M^{(h)})_y=c_y/M(y_h)$ for $y=(y_1,\dots,y_r)\in A^r$; the quotient $c_y/M(y_h)\in\bF_q$ is well defined because by assumption $0\notin A_i$ for each $0\leq i\leq t$, so $M(y_h)\neq 0$ for every $y_h\in A$. Now we define
  \begin{align}
    \label{eq:Fhdef} F^{(h)}(c) &= M(U_0^{(h)},\dots,U_t^{(h)})\cdot\evl_{A^r}^{-1}(c/M^{(h)}) \in \bigotimes_{i\in[r]}\spn\{L^{\1_{h=i}}\}\\
    \nonumber \phi^{(h)}(c) &= \left(\bigotimes_{i\in[r]}\iota^{\1_{h=i}}\right)(F^{(h)}(c)) \in Z^1(\cG_{\cY}^{(h)})\\
    \nonumber \phi_{H'}^{(h)}(c) &= \phi^{(h)}(c) + B^1(\cG_{\cY}^{(h)}) \in {H^{(h)}}'.
  \end{align}
  To clarify the notation in the definition of $\phi^{(h)}(c)$ above, for each monomial $\prod_{j,i}(U_j^{(i)})^{k_j^{(i)}}$, then
  \begin{equation*}
    \left(\bigotimes_{i\in[r]}\iota^{\1_{h=i}}\right)\left(\prod_{j,i}(U_j^{(i)})^{k_j^{(i)}}\right) = \bigotimes_{i\in[r]}\iota^{\1_{h=i}}\left((U_0^{(i)})^{k_0^{(i)}},\dots,(U_t^{(i)})^{k_t^{(i)}}\right).
  \end{equation*}

  The following claim implies that $\zeta_{H'}$ diagonalizes in the basis $\{\phi^{(h)}_{H'}(\ind{y}):y\in A^r\}$, which in turn implies that $\subrank(\zeta_{H'})\geq|A^r|$.

  \begin{claim}
    \label{claim:zetaphi}
    It holds for every $(c^{(1)},\dots,c^{(r)})\in(\bF_q^{A^r})^r$ that
    \begin{equation*}
      \zeta_{H'}(\phi_{H'}^{(1)}(c^{(1)}),\dots,\phi_{H'}^{(r)}(c^{(r)})) = \sum_{y\in A^r}c^{(1)}_y\cdots c^{(r)}_y.
    \end{equation*}
  \end{claim}
  \begin{proof}
    Because $\zeta$ is coboundary-invariant on $H'$, it suffices to show that
    \begin{equation}
      \label{eq:zetagoal}
      \zeta(\phi^{(1)}(c^{(1)}),\dots,\phi^{(r)}(c^{(r)})) = \sum_{y\in A^r}c^{(1)}_y\cdots c^{(r)}_y.
    \end{equation}
    For $y\in Y(r)$, by definition
    \begin{align*}
      \xi(\phi^{(1)}(c^{(1)}),\dots,\phi^{(r)}(c^{(r)}))
      &= \sum_{\pi\in\cS_r}\sgn(\pi) \cdot (\phi^{(1)}(c^{(1)})|_y)_{e_y^1(\pi)}\cdots(\phi^{(r)}(c^{(r)})|_y)_{e_y^r(\pi)}.
    \end{align*}
    By definition $\phi^{(h)}(c^{(h)})$ is supported on direction-$h$ edges, so the only nonvanishing term in the sum above is given by the identity permutation $\pi=\text{it}$, and for $y=(y_1,\dots,y_r)\in Y(r)$ we have
    \begin{align}
      \label{eq:xitopeval}
      \begin{split}
        \xi(\phi^{(1)}(c^{(1)}),\dots,\phi^{(r)}(c^{(r)}))_y
        &= (\phi^{(1)}(c^{(1)})|_y)_{e_y^1(\text{id})}\cdots(\phi^{(r)}(c^{(r)})|_y)_{e_y^r(\text{id})} \\
        &= F^{(1)}(c)(y)\cdots F^{(r)}(c)(y).
      \end{split}
    \end{align}
    Here the second equality above holds because for every $h\in[r]$ and every $(f^{(h)}_i(U_0^{(i)},\dots,U_t^{(i)})\in L^{\1_{h=i}})_{i\in[r]}$, we showed in~(\ref{eq:logform}) that $g^{(h)}:=\bigotimes_{i\in[r]}\iota^{\1_{h=i}}(f^{(h)}_i)$ satisfies $(g^{(h)}|_y)_{e_y^h(\text{id})}=f_1^{(h)}(y_1)\cdots f_r^{(h)}(y_r)=(f_1^{(h)}\cdots f_r^{(h)})(y)$. As $\phi^{(h)}(c^{(h)})$ is by definition a linear combination of such 1-cochains $g^{(h)}$, where $F^{(h)}(c)$ is the associated linear combination of polynomials $f_1^{(h)}(U^{(1)})\cdots f_r^{(h)}(U^{(r)})$ with each $U^{(i)}=(U_0^{(i)},\dots,U_t^{(i)})$, it follows that $(\phi^{(h)}(c^{(h)})|_y)_{e_y^h(\text{id})}=F^{(h)}(c)(y)$, and hence~(\ref{eq:xitopeval}) holds.

    Now because each $F^{(h)}(U_j^{(i)})_{j\in\{0,\dots,t\}}^{i\in[r]}\in\bigotimes_{i\in[r]}\spn\{L^{\1_{h=1}}\}$, we have
    \begin{align*}
      F^{(1)}(c)\cdots F^{(r)}(c)
      &\in \bigotimes_{i\in[r]}\spn\{f^{(1)}_i(U^{(i)})\cdots f^{(r)}_i(U^{(i)}):f^{(h)}_i\in L^{\1_{h=i}}\;\forall h\in[r]\} \\
      &= \bigotimes_{i\in[r]}\spn\{M(U^{(i)})\cdot f^{(1)}_i(U^{(i)})\cdots f^{(r)}_i(U^{(i)}):f^{(h)}_i\in L^{0}\;\forall h\in[r]\}.
    \end{align*}
    The RHS above is precisely the span of all monomials $\prod_{j,i}(U_j^{(i)})^{k_j^{(i)}}$ where for every $i\in[r]$, we have $k_0^{(i)}=|E|-1$ and $k_j^{(i)}\in[q-\lfloor\ell/2t\rfloor,\; q-\lfloor\ell/2t\rfloor+r(a-1)]$ for $j\in[t]$. Therefore if the monomial $\prod_{j,i}(U_j^{(i)})^{k_j^{(i)}}$ has nonzero coefficient in $F^{(1)}(c)\cdots F^{(r)}(c)$, then for every $i\in[r]$, we have
    \begin{align}
      \label{eq:Fproddeg}
      \begin{split}
        \sum_{j=0}^tk_j^{(i)}
        &\in [|E|-1+t(q-\lfloor\ell/2t\rfloor),\; |E|-1+t(q-\lfloor\ell/2t\rfloor+r(a-1))] \\
        &\subseteq [t(q-1)+|E|-1-\ell/2,\; t(q-1)+|E|-1-\ell/2+tra] \\
        &\subseteq [t(q-1)+|E|-\ell+(r-1)\ell',\; t(q-1)+|E|-1],
      \end{split}
    \end{align}
    where the third inclusion above holds assuming $q$ is sufficiently large because by definition $\ell=\lfloor\Delta/2\rfloor$ with $\Delta=\lfloor q^\nu\rfloor\cdot\lfloor q^{\delta-\nu}\rfloor$ and $t\in(q^\tau,q^{\nu/32})$, $\ell'=\lfloor\ell/10r\rfloor$, and $a=\lfloor\ell/10rt\rfloor$. Now it follows that
    \begin{align*}
      \zeta(\phi^{(1)}(c^{(1)}),\dots,\phi^{(r)}(c^{(r)}))
      &= \alpha(\xi(\phi^{(1)}(c^{(1)}),\dots,\phi^{(r)}(c^{(r)}))) \\
      &= \alpha(\evl_{Y(r)}(F^{(1)}(c)\cdots F^{(r)}(c))) \\
      &= \sum_{y\in A^r}F^{(1)}(c)(y)\cdots F^{(r)}(c)(y) \\
      &= \sum_{y\in A^r}c^{(1)}_y\cdots c^{(r)}_y,
    \end{align*}
    where the first equality above holds by the definition of $\zeta$ in~(\ref{eq:zetadef}), the second equality holds by~(\ref{eq:xitopeval}), the third equality holds because~(\ref{eq:Fproddeg}) implies that $\alpha(\evl_{Y(r)}(F^{(1)}(c)\cdots F^{(r)}(c)))$ computes the polynomial $g'=F^{(1)}(c)\cdots F^{(r)}(c)$ (see the definition of $\alpha$ above) and then outputs $\sum_{y\in A^r}g'(y)$, and the fourth equality holds because~(\ref{eq:Fhdef}) implies that $F^{(h)}(c)(y)=M(y_h)\cdot c_y/M(y_h)=c_y$ for every $h\in[r]$. Thus~(\ref{eq:zetagoal}) holds, as desired.
  \end{proof}

  Claim~\ref{claim:zetaphi} implies that for every $(y^{(1)},\dots,y^{(r)})\in(A^r)^r$,
  \begin{equation*}
    \zeta_{H'}(\phi_{H'}^{(1)}(\ind{y^{(1)}}),\dots,\phi_{H'}^{(r)}(\ind{y^{(r)}})) = \1_{y^{(1)}=\cdots=y^{(r)}}.
  \end{equation*}
  It follows by Definition~\ref{def:subrank} that $\subrank(\zeta_{H'})\geq|A|^r$. Therefore we obtain the desired bound
  \begin{align*}
    \subrank(\zeta_{H'})
    &\geq a^{rt} \\
    &= \lfloor\ell/10rt\rfloor^{rt} \\
    &\geq (\Delta/21rt)^{rt} \\
    &\geq (q^{\delta-\nu/2})^{rt} \\
    &\geq N^{\delta-\nu},
  \end{align*}
  where the third inequality above holds because $\Delta=q^{\delta-o(1)}$ and $t\leq q^{\nu/32}$, and the fourth inequality holds by~(\ref{eq:codelenbound}).
\end{proof}

\subsection{Putting it All Together}
Combining the results in Sections~\ref{sec:basicparamproof}-\ref{sec:subrankproof} immediately yields Theorem~\ref{thm:qldpcmain}:

\begin{proof}[Proof of Theorem~\ref{thm:qldpcmain}]
  The result follows directly from Lemma~\ref{lem:basicparam}, Lemma~\ref{lem:zetaloc}, Lemma~\ref{lem:cobinv}, and Lemma~\ref{lem:subrankbound}.
\end{proof}

\section{Classical LTCs with Multiplication Property via Balanced Product}
\label{sec:cltc}
In this section, we take balanced products of the classical LDPC codes from Section~\ref{sec:cldpc} to obtain classical LDPC codes, which we show in appropriate parameter regimes have almost linear dimension, nearly linear distance, polylogarithmic locality, inverse polylogarithmic soundness, and exhibit the multiplication property. Recall that this multiplication property, which requires products of codewords to belong to a larger code with similar paramters, is in some sense a classical analogue of transversal $C^{r-1}Z$ gates on quantum codes.

\begin{theorem}
  \label{thm:cltcmain}
  There exists a sufficiently small constant $\eta>0$ such that for every fixed $\nu\in(0,1/2)$, $\delta\in(\nu,1-\nu]$, the following holds for every sufficiently large prime power $q$. Define $\tau,\Delta,t,\Gamma,\bar{\Gamma}$ as in Corollary~\ref{cor:expinst}. For $\ell\in[\Delta]$, let $\cF_{\bar{\cX}}^{(\ell)}$ be a 1-dimensional type-\ref{it:ker} RM-planted complex with parameters $\nu,\delta,\eta,q,\ell$. Define the balanced product
  \begin{equation*}
    \cC^{(\ell)} = \cF_{\bar{\cX}}^{(\ell)}\otimes_{\bF_q^t}\cF_{\bar{\cX}}^{(\ell)}
  \end{equation*}
  using the natural free action of $\bF_q^t$ on $\cF_{\bar{\cX}}^{(\ell)}$ given by Lemma~\ref{lem:expaction} and Example~\ref{example:graphproduct}. If $\ell\leq\Delta/4$, then the classical code $Z_2(\cC^{(\ell)})$ at level $2$ of $\cC^{(\ell)}_*$ is a
  \begin{equation*}
    \left[N=q^{t+O(1)},\; K\geq{(\ell-1)+(t+2)\choose t+2},\; D\geq\frac{N}{\log(N)^{O_{\nu,\delta}(1)}}\right]_q
  \end{equation*}
  code of locality $w^{\cC^{(\ell)}}\leq\log(N)^{O_{\nu,\delta}(1)}$ that is locally testable with soundness $\rho_2(\cC^{(\ell)})\geq\log(N)^{-O_{\nu,\delta}(1)}$. Furthermore, for every $r,\ell'\in\bN$ with $r(\ell-1)\leq\ell'-1$, we have
  \begin{equation}
    \label{eq:cltcmult}
    Z_2(\cC^{(\ell)})^{*r} \subseteq Z_2(\cC^{(\ell')}).
  \end{equation}
\end{theorem}

In Theorem~\ref{thm:cltcmain}, we typically think of taking $\ell=\Theta(\Delta)$, which yields the following corollary.

\begin{corollary}
  \label{cor:cltcmain}
  For an arbitrary fixed constant $r_0\in\bN$, define all variables as in Theorem~\ref{thm:cltcmain} and set $\ell_0=\lfloor\Delta/4r_0\rfloor$. Then for every sufficiently large prime power $q$ and every $\ell_0\leq\ell\leq\Delta/4$, the code $Z_2(\cC^{(\ell)})$ is a
  \begin{equation*}
    \left[N=q^{t+O(1)},\; K\geq N^{\delta-\nu},\; D\geq\frac{N}{\log(N)^{O_{\nu,\delta}(1)}}\right]_q
  \end{equation*}
  code of locality $w^{\cC^{(\ell)}}\leq\log(N)^{O_{\nu,\delta}(1)}$ that is locally testable with soundness $\rho_2(\cC^{(\ell)})\geq\log(N)^{-O_{\nu,\delta}(1)}$. Furthermore, (\ref{eq:cltcmult}) holds for every $r,\ell,\ell'$ with $r(\ell-1)\leq\ell'-1$.
\end{corollary}
\begin{proof}
  The corollary follows immediately from the fact that
  \begin{align*}
    K
    &\geq {(\ell-1)+(t+2)\choose t+2} \geq (\ell/t)^t \geq (\Delta/5r_0t)^t \geq (q^{\delta-\nu/2})^t \geq N^{\delta-\nu},
  \end{align*}
  where the fourth inequality above holds because $\Delta=q^{\delta-o(1)}$ and $t\leq q^{\nu/32}$.
\end{proof}

Note that for an arbitrarily small constant $\epsilon>0$, setting $\nu=\epsilon/2$ and $\delta=1-\epsilon/2$ ensures that the codes in Corollary~\ref{cor:cltcmain} have close-to-linear dimension $K\geq N^{1-\epsilon}$. 
As described in Section~\ref{sec:intro}, Corollary~\ref{cor:cltcmain} achieves similar parameters (up to polylog factors) as previously known constructions of qLTCs with the multiplication property, such as Reed-Muller codes with the derandomized low-degree test of \cite{ben-sasson_randomness-efficient_2003}, as well as the codes of \cite{dinur_new_2023}. Therefore rather than achieving new parameters, Theorem~\ref{thm:cltcmain} and Corollary~\ref{cor:cltcmain} highlight the generality of our techniques, as we are able to obtain classical codes (complementing the quantum codes in Section~\ref{sec:qldpc}) with desirable fault-tolerance properties.

We now turn to proving Theorem~\ref{thm:cltcmain}. To prove the distance and soundness bounds, we will apply the results of \cite{polishchuk_nearly-linear_1994,dinur_expansion_2024}, as described below. To begin, we will need the following definition.

\begin{definition}
  \label{def:prodexp}
  A pair of classical codes $C_1,C_2\subseteq\bF_q^n$ is said to be \textbf{$\rho$-product-expanding} if for every $c\in C_1\otimes\bF_q^n+\bF_q^n\otimes C_2$, there exists a decomposition $c=c_1+c_2$ with $c_1\in C_1\otimes\bF_q^n$, $c_2\in\bF_q^n\otimes C_2$ such that
  \begin{equation*}
    |c| \geq \rho n(|c_1|_1+|c_2|_2),
  \end{equation*}
  where $|c_1|_1$ (resp.~$|c_2|_2$) denotes the number of nonzero columns (resp.~rows) in the $n\times n$ matrix $c_1$ (resp.~$c_2$).
\end{definition}

As stated below, \cite{polishchuk_nearly-linear_1994} showed that Reed-Solomon codes have good product-expansion, though to the best of our knowledge \cite{kalachev_two-sided_2023} were the first to translate their result to the language of product-expansion.

\begin{theorem}[\cite{polishchuk_nearly-linear_1994}]
  \label{thm:RSpe}
  There exists an absolute constant $\rho>0$ such that for every prime power $q$, every $E_1,E_2\subseteq\bF_q$ with $\Delta:=|E_1|=|E_2|$, and every $\ell\leq\Delta/4$, the pair of codes $\evl_{E_1}(\bF_q[X]^{<\ell}),\evl_{E_2}(\bF_q[X]^{<\ell})$ is $\rho$-product-expanding.
\end{theorem}

\cite{dinur_expansion_2024} showed the following result bounding the distance and local testability (i.e.~soundness) of codes obtained from balanced products of 1-dimensional chain complexes given by Definition~\ref{def:sscode}. The presentation in \cite{dinur_expansion_2024} is fairly general, but we only state the case that is relevant to this paper.

\begin{theorem}[\cite{dinur_expansion_2024}]
  \label{thm:dlv}
  For every $\rho>0$, there exists a sufficiently small $\eta=\eta(\rho)>0$ such that the following holds. Let $G$ be an abelian group, and let $\cX$ be the 1-dimensional incidence complex associated to some $\Delta$-regular bipartite graph $\Gamma=\Gamma^{(1)}=\Gamma^{(2)}=(V,E,\ver)$ with $\lambda_2(\Gamma)\leq\eta\Delta$ that respects a free action $\sigma$ of an abelian group $G$. Let $\cF^{(1)}_{\cX},\cF^{(2)}_{\cX}$ be 1-dimensional complexes from Definition~\ref{def:sscode} such that for every $i\in\{1,2\}$, $v\in V$, and $g\in G$, the local code parity-check matrices $h_v^{(i)}\in\bF_q^{m\times E(v)}$ and $h_{\sigma(g)v}^{(i)}\in\bF_q^{m\times E(\sigma(g)v)}$ are equal, under the isomorphism $E(v)\cong E(\sigma(g)v)$ given by $\sigma(g)$. (This setup is precisely the $r=2$ case of Example~\ref{example:graphproduct}.) Also assume that for every $v^{(1)},v^{(2)}\in V$, the pair of local codes $\ker(h^{(1)}_{v^{(1)}}),\ker(h^{(2)}_{v^{(2)}})$ is $\rho$-product expanding. Then the classical code at level $2$ of the balanced product complex $\cC=\cF^{(1)}_{\cX}\otimes_G\cF^{(2)}_{\cX}$ has distance
  \begin{equation*}
    d_2(\cC) \geq \frac{|E|}{\Delta^{O(1)}}
  \end{equation*}
  and is locally testable with soundness
  \begin{equation*}
    \rho_2(\cC) \geq \frac{1}{(\Delta\cdot|E|/|G|)^{O(1)}}.
  \end{equation*}
\end{theorem}

\begin{remark}
  While \cite{dinur_expansion_2024} proves a significantly more general result than the statement in Theorem~\ref{thm:dlv}, they make two minor (and fortunately, unnecessary) assumptions that we needed to relax to make Theorem~\ref{thm:dlv} applicable to proving Theorem~\ref{thm:cltcmain}. First, \cite{dinur_expansion_2024} assume that the bipartite graph $\Gamma$ is the double cover of a non-bipartite graph (see the assumption in \cite[Section~3]{dinur_expansion_2024} that the permutation sets $A_i$ are closed under inverse). Second, \cite{dinur_expansion_2024} assume that for each $i\in\{1,2\}$, the local code parity-check matrices $h^{(i)}_v$ are the same for each $v\in V$. Fortunately, the proof of \cite{dinur_expansion_2024} applies equally well without these assumptions, and hence implies Theorem~\ref{thm:dlv}.
\end{remark}

\begin{proof}[Proof of Theorem~\ref{thm:cltcmain}]
  As described in the $G=\bF_q^t$ case of Example~\ref{example:graphproduct}, our complex $\cC^{(\ell)}$ for $\ell\in[\Delta]$ is equal to $\cG_{\cY}^{(\ell)}$ for $\cY=\bar{\cX}\otimes_{G}\bar{\cX}$ and
  \begin{equation*}
    \cG_{x_1\times_Gx_2}^{(\ell)} = \cF_{x_1}^{(\ell)}\otimes\cF_{x_2}^{(\ell)}
  \end{equation*}
  with the local coefficient maps $\cG^{(\ell)}_{y'\leftarrow y}$ for $y'\triangleright y$ defined as in Example~\ref{example:graphproduct}.

  By definition, the classical code $Z_2(\cC^{(\ell)})=Z_2(\cG_{\cY}^{(\ell)})$ at level $2$ of $\cG_{\cY}^{(\ell)}$ has length
  \begin{align}
    \label{eq:ltclenbound}
    N
    &= \dim(\cG_{\cY}^{(\ell)})_2 = |Y(2)| = |X(1)\times_G X(1)| = |E^2\times\bF_q^t| = q^{t+O(1)},
  \end{align}
  where we use the fact that $1\leq |E|\leq q$, and that
  \begin{equation}
    \label{eq:setbalprod}
    (E\times\bF_q^t)\times_G(E\times\bF_q^t) \cong E\times E\times\bF_q^t,
  \end{equation}
  with the isomorphism above given by $(e,x)\times_G(e',x')\mapsto(e,e',x+x')$. We will use this isomorphism throughout this proof below.  

  By the definition of the boundary map $\partial_2^{\cG_{\cY}^{(\ell)}}$, the code $Z_2(\cG_{\cY}^{(\ell)})=\ker\partial_2^{\cG_{\cY}^{(\ell)}}$ is the space of all elements $g\in\bF_q^{Y(2)}=\bF_q^{E^2\times\bF_q^t}$ whose restriction to components within certain affine lines agree with polynomials of degree $<\ell$. Formally, $g\in Z_2(\cG_{\cY}^{(\ell)})$ if and only if the following two constraints hold for every $\bar{v}=(0,x)+\spn\{(1,\labV(v))\}\in\bar{V}$ for $x\in\bF_q^t$ and $v\in V$, and every $\bar{e}=(\bar{e}_0,\bar{e}_{[t]})\in\bar{E}$:
  \begin{enumerate}
  \item\label{it:edgeconstraint1} There exists some polynomial $f(U)\in\bF_q[U]^{<\ell}$ such that for every $u\in E(v)$,
    \begin{equation*}
      g_{(u,x+\labV(v)\cdot u)\times_G\bar{e}} = g_{(u,\bar{e}_0,x+\labV(v)\cdot u+\bar{e}_{[t]})} = f(u).
    \end{equation*}
  \item\label{it:edgeconstraint2} There exists some polynomial $f'(U)\in\bF_q[U]^{<\ell}$ such that for every $u\in E(v)$,
    \begin{equation*}
      g_{\bar{e}\times_G(u,x+\labV(v)\cdot u)} = g_{(\bar{e}_0,u,x+\labV(v)\cdot u+\bar{e}_{[t]})} = f'(u).
    \end{equation*}
  \end{enumerate}
  Note that the first equality in each of the two equations above simply applies the isomorphism~(\ref{eq:setbalprod}).

  Now by definition, for every multivariate polynomial $g'\in\bF_q[X_0,\dots,X_{t+1}]^{<\ell}$, and every affine line $L(U)=(L_0U,\dots,L_{t+1}U_{t+1})$ for $(L_0,\dots,L_{t+1})\in\bF_q^{t+2}$, then $g'(L(U))\in\bF_q[U]^{<\ell}$. Therefore $g=\evl_{E^2\times\bF_q^t}(g')$ satisfies the conditions~\ref{it:edgeconstraint1} and~\ref{it:edgeconstraint2} above, and thus lies inside $Z_2(\cG_{\cY}^{(\ell)})$. Furthermore, if $g'\neq 0$, then beacuse $|E|\geq\Delta\geq\ell$, we have $\evl_{E^2\times\bF_q^t}(g')$. Thus the classical code at level~$2$ of $\cG_{\cY}^{(\ell)}$ has dimension
  \begin{equation*}
    K = \dim(Z_2(\cG_{\cY}^{(\ell)})) \geq \dim(\bF_q[X_0,\dots,X_{t+1}]^{<\ell}) = {(\ell-1)+(t+2)\choose t+2}.
  \end{equation*}

  For some $r\in\bN$ and $\ell\in[\Delta]$, consider some $g_1,\dots,g_r\in Z_2(\cG_{\cY}^{(\ell)})$, so that conditions~\ref{it:edgeconstraint1} and~\ref{it:edgeconstraint2} hold for $g=g_i$ for each $i\in[r]$. Fix any $\ell'\in\bN$ with $r(\ell-1)\leq\ell'-1$. Then because the product of $r$ polynomials of degree $\leq\ell-1$ is a polynomial of degree $\leq r(\ell-1)\leq\ell'-1$, it follows that conditions~\ref{it:edgeconstraint1} and~\ref{it:edgeconstraint2} with $\ell$ replaced by $\ell'$ must hold for $g=g_1*\cdots*g_r$, so $g_1*\cdots*g_r\in Z_2(\cG_{\cY}^{(\ell')})$. Thus~(\ref{eq:cltcmult}) holds.

  By the definition of a chain complex tensor product, we have that $w^{\cC^{(\ell)}}\leq 2w^{\cF_{\bar{\cX}}^{(\ell)}}$. Meanwhile, by definition $w^{\cF_{\bar{\cX}}^{(\ell)}}\leq 2\Delta$; such a bound was shown for type-\ref{it:im} and type-\ref{it:imcom} RM-planted complexes in~(\ref{eq:locbound}) the proof of Theorem~\ref{thm:qldpcmain}, and the same reasoning applies to the type-\ref{it:ker} complexes we consider here. Thus we obtain the desired locality bound
  \begin{align*}
    w^{\cC^{(\ell)}}
    &\leq 4\Delta \leq 4q^\delta \leq \log(N)^{O_{\nu,\delta}(1)},
  \end{align*}
  where the final inequality above holds because $t\in(q^\tau,q^{\nu/32})$, so by~(\ref{eq:ltclenbound}) we have $\log(N)=tq^{O(1)}=q^{\Theta_{\nu,\delta}(1)}$.

  Now by the definition our type-\ref{it:ker} RM-planted complex $\cF_{\bar{\cX}}^{(\ell)}$, every local code $\ker(h_{\bar{v}}^{(\ell)})$ is a Reed-Solomon code of length $\Delta$ and dimension $\leq\Delta/4$. Therefore Theorem~\ref{thm:RSpe} with Theorem~\ref{thm:dlv} implies that there exists some absolute constant $\eta>0$ such that as long as $\lambda_2(\bar{\Gamma})\leq\eta\Delta$, then the classical code at level~$2$ of $\cG_{\cY}^{(\ell)}$ has distance
  \begin{align*}
    D
    &= d_2(\cG_{\cY}^{(\ell)}) \geq \frac{|\bar{E}|}{\Delta^{O(1)}} \geq \frac{N}{\log(N)^{O_{\nu,\delta}(1)}}
  \end{align*}
  and soundness
  \begin{align*}
    \rho_2(\cG_{\cY}^{(\ell)})
    &\geq \frac{1}{(\Delta\cdot|\bar{E}|/|\bF_q^t|)^{O(1)}} \geq \frac{1}{q^{O(1)}} \geq \frac{1}{\log(N)^{O_{\nu,\delta}(1)}},
  \end{align*}
  as desired. Note that the final inequality in each of the two equations above follows by~(\ref{eq:ltclenbound}) and because $\log(N)=q^{\Theta_{\nu,\delta}(1)}$, as described above.
\end{proof}


\section{Acknowledgments}
We thank Niko Breuckmann for insightful discussions that helped motivate this work. We thank Omar Alrabiah for suggesting a way to simplify the proof of expansion in Theorem~\ref{thm:expproof}, and we thank Venkatesan Guruswami for valuable feedback on the manuscript that improved the presentation.

\bibliographystyle{alpha}
\bibliography{library}

\appendix

\section{Toplogical View of the Coboundary-Invariant Form}
\label{sec:ciformintuition}
In this section, we provide some topological intuition for the definition of the multilinear form $\zeta$ in the proof of Theorem~\ref{thm:qldpcmain}.

Recall from~(\ref{eq:zetadef}) that $\zeta=\alpha\circ\xi$, where $\alpha$ is a linear functional and $\xi$ is a multilinear function taking values in $\bF_q^{Y(r)}$. We first briefly comment on the definition of $\alpha$, before turning to provide intuition for the definition of $\xi$.

The functional $\alpha$ views the output of $\xi$ as the evaluation of a polynomial, so that $\alpha$ simply interpolates this polynomial, zeroes out all low-degree monomials, and then outputs the sum of the evaluations on a certain set of points. Therefore in the language of quantum codes, $\alpha$ enforces a sort of systematic encoding of the underlying logical (i.e.~message) qudits, meaning that the logical qudits correspond to evaluation points of polynomials. Similar techniques using systematic encodings have been previously used in constructions of (non-LDPC) quantum codes supporting transversal $CCZ$ gates \cite{krishna_towards_2019,wills_constant-overhead_2024,golowich_asymptotically_2024,nguyen_good_2024}. However, these prior works did not need to zero out low-degree monomials after the polynomial interpolation step. This additional step in our setting ultimately stems from the fact that the LDPC condition makes it more difficult to control the structure of the polynomials whose evaluations correspond to coboundaries.

We now discuss the definition of $\xi$. Recall from~(\ref{eq:xidef}) that for $(f^{(1)},\dots,f^{(r)})\in(\cG_{\cY}^{(1)})^1\times\cdots\times(\cG_{\cY}^{(r)})^1$, the value of $\xi(f^{(1)},\dots,f^{(r)})$ at a given point $y\in Y(r)$ equals the value of a local multilinear form $\xi^{\loc}$ applied to the restrictions of $f^{(1)},\dots,f^{(r)}$ to edges within the $r$-dimensional cube $y$. The definition of $\xi^{\loc}$ is given in~(\ref{eq:xilocdef}), and consists of a sum of terms corresponding to length-$r$ paths in the $r$-dimensional boolean hypercube from $0^r$ to $1^r$. Below, we provide a topological view of this definition of $\xi^{\loc}$; here we focus on intuition, and do not make any rigorous claims.

Consider a chain complex associated to an oriented, closed $r$-dimensional manifold. We can often obtain a coboundary-invariant $r$-multilinear form $\zeta(f^{(1)},\dots,\zeta(f^{(r)})$ on 1-cocycles $f^{(h)}$ of this manifold by taking the \textit{cup product} $f^{(1)}\cup\cdots\cup f^{(r)}$, which is a well-known operation from algebraic topology. Informally, the resulting multilinear form can be viewed as a sum of terms corresponding to points in the manifold, where the term at a given point is nonzero if the Poincar\'{e} duals of $f^{(1)},\dots,f^{(r)}$, viewed as $(r-1)$-dimensional submanifolds, all intersect at that point. Indeed, this approach is taken in~\cite{zhu_non-clifford_2023} to obtain qLDPC codes with transversal non-Clifford gates, albeit with low (at most logarithmic) distance.

To extend this approach to chain complexes from cubical complexes, which are more general than manifolds, we use a similar local multilinear form $\xi^{\loc}$ within each $r$-dimensional cube $y\in Y(r)$ that counts (signed) intersection points, and then sum up these local forms over all $y\in Y(r)$. Specifically, fix some values $1>s_1>\cdots>s_r>0$. For each $h\in[r]$, consider the $r$ hyperplanes $P^{(r,1)},\dots,P^{(r,r)}$ of dimension $r-1$ in the $r$-dimensional cube $[0,1]^r\subseteq\bR^r$ (parametrized by $u=(u_1,\dots,u_r)$), such that $P^{(r,i)}=\{u\in[0,1]^r:u_i=s_i\}$. For a given $i\in[r]$, by definition the hyperplanes $P^{(r,j)}$ for $j\neq i$ hyperplanes cut $P^{(r,i)}$ into $2^{r-1}$ pieces, where each piece touches a unique direction-$i$ edge in the boolean hypercube given by the boundary of $[0,1]^r$. Thus for each direction-$i$ edge $e\in Y^{\loc}(1)$ in the $r$-dimensional boolean hypercube, we have a piece $P^{(r)}_e$ of an $(r-1)$-dimensional hypercube intersecting edge $e$, so that $\bigcup_{e\in Y^{\loc}(1):\direc(e)=i}P^{(r)}_e=P^{(r,i)}$.

Now given $(f^{(1)},\dots,f^{(r)})\in({\cY^{\loc}}^1)^r=(\bZ^{Y^{\loc}(1)})^r$, for every sequence of edges $(e_1,\dots,e_r)\in Y^{\loc}(1)^r$ such that $P^{(1)}_{e_1}\cap\cdots\cap P^{(r)}_{e_r}\neq\emptyset$, we add a term of the form $\pm f^{(1)}_{e_1}\cdots f^{(r)}_{e_r}$ to $\xi^{\loc}(f^{(1)},\dots,f^{(r)})$. Note that we must have $\{\direc(e_1),\dots,\direc(e_r)\}=[r]$, as for $h,h',i\in[r]$ with $h\neq h'$, the hyperplanes $P^{(h,i)},P^{(h',i)}$ are parallel and do not intersect. To see that the resulting $\xi^{\loc}$ will have the form~(\ref{eq:xilocdef}), consider that if $e_1$ is not incident to the vertex $0^r$ in the $r$-dimensional boolean hypercube, then in some directon $i\neq\direc(e_1)\in[r]$, we have that the projection of $P^{(1)}_{e_1}$ to the $i$th coordinate lies inside $[s_1,1]$. But as $s_1>s_2>\cdots>s_r$, for whichever $h\in\{2,\dots,r\}$ has $\direc(e_h)=i$, then the projection to the $i$th coordinate of $P^{(h)}_{e_h}\subseteq P^{(h,i)}$ is $\{s_h\}\not\subseteq[s_1,1]$, and hence $P^{(1)}_{e_1}\cap P^{(h)}_{e_h}=\emptyset$. Thus $\xi^{\loc}(f^{(1)},\dots,f^{(r)})$ only has a term of the form $\pm f^{(1)}_{e_1}\cdots f^{(r)}_{e_r}$ if $e_1$ is incident to the vertex $0^r$. Inductively applying this reasoning to $f^{(2)},\dots,f^{(r)}$ on the restriction to the $(r-1)$-dimensional subcube with $u_{\direc(e_1)}=1$, we find that $\xi^{\loc}(f^{(1)},\dots,f^{(r)})$ only has a term of the form $\pm f^{(1)}_{e_1}\cdots f^{(r)}_{e_r}$ if $e_1,\dots,e_r$ forms a path from $0^r$ to $1^r$, as indeed is the case in~(\ref{eq:xilocdef}). The sign $\sgn(\pi)\in\{\pm 1\}$ in~(\ref{eq:xilocdef}) is then chosen to make the analysis go through; note that if $\bF_q$ has characteristic $2$, the sign becomes meaningless and the above description provides a complete definition of $\xi^{\loc}$.

Our approach described above can be viewed as expressing a cubical complex as a gluing of individual $r$-dimensional cubes. We then view each such cube as a manifold with boundary, and endow it with a local multilinear form based on intersecting submanifolds within the cube. A similar view was taken in~\cite{scruby_quantum_2024}, who also expressed certain cubical-like complexes as a gluing of individual cubes, in order to obtain chain complexes with a coboundary-invariant multilinear form. However, the quantum LDPC codes obtained in \cite{scruby_quantum_2024} are similar to those of \cite{zhu_non-clifford_2023} in that their distance remains low (i.e.~logarithmic). In contrast, by applying these techniques to more general cubical complexes, we are able to obtain qLDPC codes with polynomial distance and almost linear dimension, albeit with polylogarithmic locality.
\end{document}